\documentclass[11pt, a4paper]{article}
\usepackage{fullpage}

\usepackage[persection, figure]{belov_paper}
\usepackage{scalerel,stackengine}
\stackMath
\newcommand\reallywidehat[1]{%
\savestack{\tmpbox}{\stretchto{%
  \scaleto{%
    \scalerel*[\widthof{\ensuremath{#1}}]{\kern.1pt\mathchar"0362\kern.1pt}%
    {\rule{0ex}{\textheight}}%WIDTH-LIMITED CIRCUMFLEX
  }{\textheight}% 
}{2.4ex}}%
\stackon[-6.9pt]{#1}{\tmpbox}%
}\def\chit#1{\reallywidehat{#1}} % for Fourier Transform

\usepackage{tikz}
\usetikzlibrary{quantikz, decorations.pathreplacing, arrows.meta}
\definecolor{brown}{rgb}{0.65, 0.16, 0.16}

\newcommand{\hookmapsto}{\hookrightarrow}
\DeclareMathOperator\supp{supp}
\def\maps#1{\stackrel{#1}{\longmapsto}}
\newcommand{\OO}{\mathcal{O}}

\title{Tight Quantum Lower Bound for $k$-Distinctness}
\author{
  Aleksandrs Belovs\thanks{Center for Quantum Computing Science, Faculty of Science and Technology, University of Latvia.}
}
\date{}

\hypersetup{
    pdftitle={Tight Quantum Lower Bound for k-Distinctness}, % title
    pdfauthor={Aleksandrs Belovs}, % author
}

\def\nodegap{0.7}
\def\rowgap{0.9}
\def\marginsize{0.25}

\def\drawelement#1#2{ % row, column
    \draw (#2*\nodegap, #1*\rowgap) node[draw, circle, inner sep=2] {};
}
\def\partitionblock#1#2#3{ %Internal function
    (#2*\nodegap-\marginsize, #1*\rowgap-\marginsize) rectangle (#3*\nodegap+\marginsize, #1*\rowgap+\marginsize);
}
\def\drawblock#1#2#3{ % row, column start, column end
    \draw\partitionblock{#1}{#2}{#3}
}
\def\fillblock#1#2#3{ % row, column start, column end
    \draw[fill=gray!50]\partitionblock{#1}{#2}{#3}
}
\def\drawbrace#1#2#3#4{
    \draw [decorate, decoration={brace,amplitude=5pt}] (#2*\nodegap-0.3, #1*\rowgap + 0.3) --node[above=4pt]{$#4$} (#3*\nodegap+0.3, #1*\rowgap +0.3);
}

\begin{document}
\maketitle

\begin{abstract}
In this paper, we introduce a new quantum query lower bound framework.
It is inspired by Zhandry's compressed oracle technique~\cite{zhandry:record}, but it also subsumes the polynomial method~\cite{beals:pol} as a special case.
Compared to Zhandry's technique, our approach has two key differences.
First, we do not use any oracles (except for the standard input oracle), and define ``knowledge'' directly through the expansion of the state of the algorithm in the Fourier basis.
Second, we allow arbitrary probability distributions of inputs.

We show how this framework behaves on the problem of finding equal elements in the input string.
In particular, we demonstrate its power by proving a first tight quantum query lower bound for the $k$-Distinctness problem, matching the upper bound from~\cite{belovs:learningKDist}.
\end{abstract}

\mycutecommand{\q}{\bZ_q}

\section{Introduction}
The $k$-Distinctness problem asks to find $k$ equal elements in the input string $x\in [q]^n$, where $q$ is the size of the input alphabet.
For many years, this problem has been on the front-line of quantum algorithm research.
It is sufficiently simple, but complicated enough to inspire many important algorithmic and lower bound techniques.

The first quantum algorithm for this problem was due to Buhrman et al.~\cite{buhrman:distinctness} in 2000 using quantum amplitude amplification~\cite{brassard:amplification}.
It achieved $\OO\sA[n^{3/4}]$ query and nearly the same time complexity for the special case of $k=2$, which is known as the Element Distinctness problem.
Shortly after, Aaronson and Shi~\cite{shi:collisionLower} proved an $\Omega\sA[n^{1/3}]$ query lower bound for the closely related Collision problem, which immediately implied an $\Omega\sA[n^{2/3}]$ query lower bound for Element Distinctness.
(See also~\cite{kutin:collisionLower} for a different version of this proof.)
This lower bound only worked in the assumption that the size of the input alphabet $q = \Omega\sA[n^2]$.
Ambainis~\cite{ambainis:collisionLower} removed this restriction and demonstrated that the lower bound $\Omega\sA[n^{2/3}]$ holds even for $q=n$.
(Ambainis consider the decision version of the problem, for which the regime $q<n$ is nonsensical.)
All these developments used the polynomial method~\cite{beals:pol}.

Ambainis was also the first to construct the matching upper bound.
In 2003, he used quantum walks on the Johnson graph to solve $k$-Distinctness in $\OO\sA[n^{k/(k+1)}]$ queries~\cite{ambainis:distinctness}.
For $k=2$, this gives $\OO\sA[n^{2/3}]$, matching the aforementioned lower bound.
The time complexity of the algorithm is close to its query complexity assuming a powerful form of quantum random access memory.

The upper bound on the quantum query complexity of $k$-Distinctness for larger $k$ was improved to $\OO\sB[n^{\frac 34 - \frac{1}{4(2^k-1)}} ]$ in 2012 by Belovs~\cite{belovs:learningKDist} using learning graphs~\cite{belovs:learning}.
This algorithm was only query-efficient.
The time-efficient version was obtained for $k=3$ in 2013 independently by Childs, Jeffery, Kothari, and Magniez~\cite{childs:walk3Dist} and Belovs~\cite{belovs:electicityQuantumWalks} (published as a merged paper~\cite{belovs:mergedWalk3Dist}).
A time-efficient version for general $k$ was constructed only in 2022 by Jeffery and Zur~\cite{jeffery:kDist} using multidimensional quantum walks.

On the lower bound front, the best known lower bound for all $k\ge 2$ had been the same $\Omega(n^{2/3})$ for more than a decade, until, in 2017, Bun, Kothari, and Thaler~\cite{bun:polynomialStrikesBack} showed a lower bound of $\Omega\sA[n^{\frac 34 - \frac1{2k}}]$ for all $k$.
This was later slightly improved to $\Omega\sA[n^{\frac 34 - \frac1{4k}}]$ by Mande, Thaler, and Zhu~\cite{mande:kDistinctness}.
While, for large $k$, it approaches the correct asymptotic of $n^{3/4}$, these lower bounds give no improvement over $\Omega(n^{2/3})$ for the case $k=3$.
Both papers used the polynomial method in the form of dual polynomials.

In this paper, we prove a tight quantum query lower bound for $k$-Distinctness for all values of $k$, matching the aforementioned upper bound by Belovs~\cite{belovs:learningKDist}:

\begin{thm}
\label{thm:main}
Any quantum algorithm solving the $k$-Distinctness problem with bounded error makes
\[
\Omega\sB[n^{\frac 34 - \frac{1}{4(2^k-1)}} ]
\]
queries to the input string (assuming $k=\OO(1)$ and that the size of the input alphabet is at least $\Omega(n^2)$).
\end{thm}

As it can be deduced from the above exposition, the polynomial method played an important role in lower bounding quantum query complexity of $k$-Distinctness.
Another important lower bound method is Zhandry's compressed oracle technique~\cite{zhandry:record} developed in 2018.
It is a powerful tool, which allows one to talk about algorithm's knowledge of the input string in a very precise and formal way.
It was used both to prove lower bounds for concrete computational problems~\cite{liu:multiCollisions, hamoudi:multipleCollisionPairs} and to show security of a number of classical cryptographic schemes in the post-quantum world~\cite{liu:revisitingFiatShamir, chung:compressedOracle, alagic:EvenMansour}.
However, its important limitation is that it is tailored for the average-case settings assuming the uniform (or at least i.i.d.) distribution on the input strings.
Only very recently the compressed oracle method was generalised to uniformly random permutations~\cite{carolan:compressedPermutationOracles}.

In particular, Liu and Zhandry~\cite{liu:multiCollisions} proved matching upper and lower bounds on quantum query complexity of finding $k$ equal elements in \emph{a uniformly random} input string.
This result was later extended by Hamoudi and Magniez to prove time-space tradeoffs for finding multiple collision pairs~\cite{hamoudi:multipleCollisionPairs}.
Let us note that the assumption that the input string comes from the uniformly random distribution is different from the worst-case settings we consider in this paper, and it also significantly simplifies some aspects of the problem (see, e.g., discussion in~\cite{belovs:learningKdistPrior}).
Therefore, our result, \rf{thm:main}, is incomparable with prior work by Liu and Zhandry.

To prove \rf{thm:main}, we develop a new framework for proving quantum query lower bounds that combines both the polynomial and (with some modifications) Zhandry's method.
Contrary to Zhandry's method, we are not limited to the uniform distribution, nor do we use any oracles besides the standard input oracle.
Strictly speaking, Zhandry's method is not a special case of our framework, but there is a close affinity between the two methods, see \rf{sec:zhandry}.

\section{Overview of the Paper}
\label{sec:overview}

The results of our paper can be divided into three layers of generality:
\begin{itemize}\noseps
\item A quantum query lower bound framework that extends Zhandry's method to other probability distributions, and which, thereafter, consumes the polynomial method as a special case.
\item Application of the above framework to the task of finding equal elements in the input string.
\item A lower bound for the $k$-Distinctness problem that is a special case of the task of finding equal elements. 
\end{itemize}
The organisation of the paper reflects this division.
\rf{sec:framework} corresponds to the first bullet above, \SuperRef Sections\ref{sec:equalElements} and \ref{sec:kDistAntiConcentration} to the second one, while \rf{sec:kDistQueryGain} mostly deals with just the $k$-Distinctness problem.
We try to keep our framework modular with different lemmata of various generality.
We separate them from lower bounds for concrete problems and explicitly state assumptions for all of them.
In this overview, \rf{sec:overviewFramework} deals with the first point above, while \rf{sec:overviewSearch} explains the last two.

\subsection{Framework}
\label{sec:overviewFramework}

In this section, we describe our general framework, a thorough exposition of which will follow in \rf{sec:framework}.
Throughout the paper, we consider a quantum query algorithm running in space $\cA$ with the initial state $\ket A|0>$.
It has oracle access to an input string $x\in \q^n$, where $\q$ is the input alphabet (we identify the input alphabet with the additive group modulo $p$, which is without loss of generality, but allows us to define the Fourier basis in $\bC^q$).
We consider relational problems, i.e., on each input string $x$, there can be a whole subset $R_x\subseteq R$ of correct outputs, and every output in this subset is equally good.

\subsubsection{Common Foundations}
\label{sec:commonFoundations}

We start with what we identify as the common foundations of both the polynomial and Zhandry's methods, which corresponds to the content of \rf{sec:uniform}.
The foundations uses the following four ideas.

First, the algorithm is executed in the joint space $\cX\otimes \cA$, where $\cX = (\bC^q)^{\otimes n}$ stores the input string $x\in \q^n$ and $\cA$ is the space of the algorithm.
The register $\cX$ is only accessible to the algorithm through applications of the input oracle $O$, which thus intertwines the state between these two registers.
(This is an old idea, e.g., it was already used in Ambainis' original adversary method~\cite{ambainis:adv}.)

Second, the Fourier basis is used in $\cX$.
Recall that it is given by vectors $\ket X|\chit \sigma>$ with $\sigma\in \q^n$.
Also, it is essential that the input string $x=(x_1,\dots,x_n)$ is sampled uniformly at random from $\q^n$.
This means that the algorithm is executed on the uniform initial state 
\[
\psi_0 = \frac{1}{\sqrt{q^n}} \sum_{x\in \q^n} \ket X|x>\ket A|0> = \ket X|\chit \emptyset> \ket A|0>.
\]
(Here $\emptyset$ stands for the identically zero function on $[n]$, see \rf{sec:prelim}.)
We denote by $\psi_t$ the state of the algorithm after $t$ queries when executed on this initial state.
We call them \emph{uniform states of the algorithm}, and they will be essential in our construction.

Third, due to the phase kickback trick, we may assume that the input oracle acts in $\cX$ instead of $\cA$.
Moreover, it acts by modifying the entry $\sigma(i)$ of the function $\sigma\in\q^n$ in $\ket X|\chit \sigma>$, where $i$ is the index of the input variable being queried.
As one consequence of this, it is easy to show that, after $t$ queries, the state $\psi_t$ of the algorithm satisfies
\[
\psi_t\in \cX_{\le t}\otimes \cA,
\]
where $\cX_{\le t}$ is the subspace of $\cX$ spanned by all $\ket X|\chit \sigma>$ with $\sigma$ having support size at most $t$.

Fourth, this gives the main piece of intuition which we use extensively in our paper.
The state $\ket X|\chit \sigma>$ corresponds to the algorithm \emph{``knowing'' the values of the variables in the support of $\sigma$} and, even more importantly, \emph{not ``knowing'' the values outside of it}.
In particular, the above observation that $\psi_t \in \cX_{\le t}\otimes \cA$ corresponds to a piece of intuition that, after $t$ queries, the algorithm knows the values of at most $t$ variables.
More precisely, what we mean by this is that the state of the algorithm is a superposition of several substates, where each substate knows the values of at most $t$ variables.

\subsubsection{Transfer Operator, Anti-Concentration, and the Polynomial Method}
\label{sec:overviewTransform}
In this section, we switch from the uniform probability distribution we used in \rf{sec:commonFoundations} to an arbitrary probability distribution, which corresponds to the content of \rf{sec:spaceY}.
The transition is based on the following simple idea.
Let $Y = \q^n$ equipped with some probability distribution $p = (p_y)$, and let $\cY = \bC^Y$ be another register isomorphic to $(\bC^q)^{\otimes n}$.
(We use a different register to emphasise a different probability distribution.  Also, we will shorty switch to a register $\cY$ that is not isomorphic to $\cX$.)
We define the corresponding \emph{transfer operator} as a linear map 
\[
\Upsilon\colon \cX\to\cY\colon\quad \frac{1}{\sqrt{q^n}} \ket X|x>\longmapsto \sqrt{p_x}\ket Y|x>
\]
for all $x\in \q^n$.
If $\psi_t$ is the uniform state of the algorithm after $t$ queries, then it is trivial to see that $\Upsilon \psi_t$ is the state after $t$ queries when the initial state of the algorithm is $\sum_{y\in \q^n} \sqrt{p_y}\ket Y|y>\ket A|0>$.
We will call them \emph{non-uniform states}.
In particular, we get that the non-uniform state $\Upsilon \psi_t$ lives in the space $\Upsilon\cX_{\le t}\otimes \cA$.

This might seem too na\"\i ve to possibly work, but the polynomial method (in the form of dual polynomials) can be interpreted as working exactly in this way, which we demonstrate in \rf{sec:polynomial}.
The idea of the polynomial method is to come up with a carefully constructed probability distribution $p$ so that every state $\phi\in \Upsilon\cX_{\le t}\subseteq \cY$ is \emph{anti-concentrated}: No matter what response $\rho \in R$ you choose, the probability that measuring $\phi$ in the computation basis yields $y\in \q^n$ for which $\rho$ is a correct answer is small.
This is enough to prove small success probability of the algorithm, see \rf{sec:overallAntiConcentration}.

\subsubsection{Knowledge and Zhandry's Method}
\label{sec:overviewKnowledge}
Unfortunately, coming up with a probability distribution $p$ in the polynomial method is tricky.
Informally, this is because the polynomial method only uses that the algorithm knows the values of at most $t$ variables, but it does not keep track of \emph{which} input variables these are.
Zhandry's method goes beyond that by proving that it is hard for the algorithm to learn the relevant variables for every input string out there.
This is quantified by the concept of \emph{knowledge}, which we will now describe.
For each input string $y \in \q^n$, let $L^+_y \subseteq 2^{[n]}$ be the set of certificates of $y$ (in some sense).
We define the \emph{knowledge operator} as a linear map 
$
\Upsilon^+\colon \cX\to\cY
$
given by 
\begin{equation}
\label{eqn:overviewKnowledgeAlternative}
\Upsilon^+\ket X|\chit \sigma> = \Pi_\sigma^+ \Upsilon \ket X|\chit\sigma>
\end{equation}
for all $\sigma\in \q^n$, 
where $\Pi^+_\sigma$ projects onto the span of $\ket Y|y>$ satisfying $\supp(\sigma) \in L^+_y$.
This is very similar to the projectors used in Zhandry's method, however, they are not identical, see \rf{sec:zhandry}.
Similarly, we define the operator $\Upsilon^-\colon \cX\to\cY$ so that we have the following decomposition
\begin{equation}
\label{eqn:overviewDecmposition}
\Upsilon \psi_t = \Upsilon^+\psi_t  + \Upsilon^-\psi_t
\end{equation}
of the state of the algorithm into two parts: one with knowledge and one without.
Now, to upper bound the success probability of the algorithm after $t$ queries, it suffices to show the following two things:
\begin{itemize}\noseps
\item \emph{anti-concentration:} every vector $\phi \in \Upsilon^-\cX_{\le t}$ is anti-concentrated.
\item \emph{small knowledge:} the \emph{knowledge of the algorithm}, which we quantify by $\|\Upsilon^+ \psi_{t}\|$, is small.
\end{itemize}

The anti-concentration part is analogous to what we did in \rf{sec:overviewTransform} for the polynomial method, but with $\Upsilon$ replaced by $\Upsilon^-$.
Removal of the $\Upsilon^+$-part of the state makes this task much more feasible, see, e.g., the example in \rf{sec:overviewKDistAntiConcentration}.
We will discuss estimation of the knowledge in more detail in \rf{sec:overviewQueryGain}.

Summing up, the polynomial method can be obtained from our framework by letting $\Upsilon^+=0$ and proving anti-concentration for $\Upsilon = \Upsilon^-$.
It requires quite some effort in the construction of the probability distribution $p$ on $Y$.
Zhandry's method works quite well and was, in fact, limited to a very simple probability distribution: the uniform one; and it works by bounding the knowledge of the algorithm.
%It can be obtained from our framework by letting $\Upsilon = I$.
Our framework naturally combines both approaches using the transfer operator to allow arbitrary (but still natural) probability distributions $p$, and simultaneously keeping track of the knowledge using the knowledge operator.

\subsubsection{Types of Inputs}
\label{sec:overviewTypes}
We also introduce \emph{types} of inputs.
Informally, a type $\mu$ is a set of inputs that we treat as one indivisible entity.
In particular, the correct output of the algorithm only depends on the type, as well as the knowledge system $L^+$.
Thus, we will write $R_\mu$ and $L^+_\mu$ instead of $R_x$ and $L^+_x$, where $\mu$ is the type of $x$.
We formalise this by defining the Hidden Computational Problem.
In this problem, the input is a pair $(\mu,y)\in Y$, where the required output only depends on $\mu$ (the type), and the algorithm only has access to $y$ (the input string in $\q^n$ as before).
We also assume some probability distribution $p_{\mu,y}$ on such pairs.

This definition captures a wide variety of hidden problems, a framework popular in quantum algorithms, most famous being the Hidden Subgroup Problem~\cite{kitaev:HSP}, where the type $\mu$ is the hidden subgroup.
In the $k$-Distinctness problem, we say that two input strings $x$ and $y$ have the same type if they induce the same partition on $[n]$, i.e., $x_i = x_j$ if and only if $y_i = y_j$.
(I.e., a type in this case is a partition of $[n]$ into subsets of equal elements.)

Let $M$ be a collection of possible types.
We represent each input by a pair $\ket Y|\mu,y>$ with $\mu \in M$ and $y\in \q^n$.
This gives a decomposition of the space $\cY = \bC^Y$ into a direct sum $\bigoplus_\mu \cY_\mu$, where the subspace $\cY_\mu$ is spanned by all the inputs of type $\mu$.
For each $\mu$, we can define its own transfer operator $\Upsilon_\mu\colon \cX\to\cY_\mu$ (where the underlying probability distribution $p$ is conditioned on having type $\mu$).
Thus, for every $\phi\in \cX$, it holds that 
\begin{equation}
\label{eqn:overviewUpsilonDirectSum}
\Upsilon \phi = \bigoplus_{\mu} \sqrt{p_\mu}\; \Upsilon_\mu \phi,
\end{equation}
where $p_\mu = \sum_{y} p_{\mu,y}$ is the probability of having type $\mu$.
Similarly, we define a knowledge operator restricted to a type $\mu$ as a linear map 
\begin{equation}
\label{eqn:overviewKnowledge}
\Upsilon^+_\mu\colon \cX\to\cY_\mu\colon\quad
\ket X|\chit \sigma> \longmapsto 1_{\supp(\sigma)\in L^+_\mu} \Upsilon_\mu \ket X|\chit\sigma>.
\end{equation}
The operator $\Upsilon^+\colon \phi\mapsto \bigoplus_{\mu} \sqrt{p_\mu}\; \Upsilon^+_\mu \phi$ combines all knowledge operators for individual $\mu$.
Similarly, $\Upsilon^-\mu$ maps $\ket X|\chit \sigma>$ into $1_{\supp(\sigma)\notin L^+_\mu} \Upsilon_\mu \ket X|\chit\sigma>$, and $\Upsilon^-$ maps $\phi$ into $\bigoplus_{\mu} \sqrt{p_\mu}\; \Upsilon^-_\mu \phi$.

The definition of $\Upsilon^+$ above complies with the one in~\rf{eqn:overviewKnowledgeAlternative} if we assume that each input has its own type.
Introduction of types serves two main purposes: first, abstracting away unnecessary details; and, second, obtaining a nice expression for the transfer operator $\Upsilon_\mu$, see \rf{sec:overviewPartition} below.

\subsubsection{Query Gain}
\label{sec:overviewQueryGain}

A natural way to bound the knowledge $\|\Upsilon^+ \psi_{t}\|$ is to show that it does not grow much with each query.
This can be achieved using the \emph{query gain operator} defined as follows (see \rf{sec:queryGain} for more detail).
For each type $\mu \in M$ and input variable index $i\in [n]$, we define the system $L^{\partial i}_\mu$ that consists of those subsets of $[n]$ that are added to $L^+_\mu$ when the $i$-th input variable is queried.
Formally, $S\in L^{\partial i}_\mu$ iff $S\notin L^+_\mu$ but $S\cup\{i\} \in L^+_\mu$.
Using this system, we define $\Upsilon_\mu^{\partial i}$ similarly to $\Upsilon^+_\mu$ and essentially arrive at the following \emph{query gain bound}:
\begin{equation}
\label{eqn:overviewQueryGain}
\|\Upsilon^+ \psi_{t+1}\| - \|\Upsilon^+ \psi_{t}\| \le \|\Psi^\partial \psi_t\| + \|\Psi^\partial \psi_{t+1}\|,
\end{equation}
where $\Psi^\partial$ combines $\Upsilon^{\partial i}_\mu$ over all $\mu$ and $i$.
Informally, the first term, $\|\Psi^\partial \psi_t\|$, corresponds to learning values of the input variables, while the second one, $\|\Psi^\partial \psi_{t+1}\|$, to unlearning.
Since $L^{\partial i}_\mu$ uses only a small fraction of the subsets $S\subseteq [n]$, we expect $\|\Psi^\partial \psi_t\|$ to have small norm.

It makes sense to define several knowledge systems and knowledge operators for the same problem.
For example, in the $k$-Distinctness problem, the main task is to find a $k$-tuple of equal elements.
We represent this task by the knowledge operator $\Upsilon_k^+$, which is the same $\Upsilon^+$ as above.
Informally, the hardness of finding a $k$-tuple during the next query depends on how many $(k-1)$-tuples of equal elements the algorithm already knows.
We represent the task of finding $(k-1)$-tuples by another knowledge operator $\Upsilon_{k-1}^+$.
Then, it is possible to bound the query gain $\|\Psi_k^\partial \psi_t\|$ of the former in terms of the knowledge $\|\Upsilon_{k-1}^+\psi_t\|$ of the latter.
Similarly, we define the knowledge operator $\Upsilon_{k-2}^+$ and bound $\|\Psi_{k-1}^\partial \psi_t\|$ in terms of $\|\Upsilon_{k-2}^+\psi_t\|$, and so on.
This gives a recursion, solving which, we can bound the success probability of the algorithm.

While all the knowledge operators above share the same domain, they can have different ranges.
That is: $\Upsilon_\ell^+\colon \cX\to \cY_\ell$, where each $\cY_\ell$ is spanned by some set $Y_\ell = \sfigA{\ket |\mu, y>}$ of input strings $y\in\q^n$ together with their types $\mu$.

\subsection{\texorpdfstring
    {Search for Equal Elements and $k$-Distinctness}
    {Search for Equal Elements and k-Distinctness}
}
\label{sec:overviewSearch}

Next, we apply our framework described in \rf{sec:overviewFramework} to the problem of finding equal elements in the input string, in particular, to the $k$-Distinctness problem.
We start with the transfer operator in \rf{sec:overviewPartition}, which works for a a general task of finding equal elements.
Then, we concentrate on the $k$-Distinctness problem.
In \rf{sec:overviewKDistKnowledge}, we define the knowledge operators $\Upsilon^+_\ell$ mentioned at the end of \rf{sec:overviewQueryGain}.
In \SuperRef Sections\ref{sec:overviewKDistAntiConcentration} and~\ref{sec:overviewKDistQueryGain}, we tackle anti-concentration and the query gain.

\subsubsection{Partitions and the Transfer Operators}
\label{sec:overviewPartition}

For the $k$-Distinctness problem, the type of an input string $y\in\q^n$ is given by its partition into blocks of equal elements.
That is, types $\mu\in M$ are given by partitions $\mu = \{B_1,\dots, B_m\}$ of $[n]$, and we will use both terms interchangeably.
For a fixed $\mu$, input strings of type $\mu$ are in 1-to-1 correspondence with injective functions in $\q^\mu$.

As a first step, we consider a relaxed version of the problem, where we allow all functions in $\q^\mu$ (not necessarily injective).
In this way, some input strings can have multiple types (e.g., the string $0^n$ has all possible types).
Nonetheless, the problem stays well-defined, moreover, if the size of the alphabet $q = \Omega(n^2)$, there is no significant difference between the two versions of the problem (as collisions in $\q^\mu$ are rare).
A similar relaxation idea was used, e.g., in the quantum $k$-sum lower bound~\cite{belovs:kSumLower} and dual learning graphs~\cite{belovs:onThePower}.

The assumption $q = \Omega(n^2)$ is not an inherent limitation of our approach, but it allows us to obtain a very nice form of the transfer operator $\Upsilon_\mu$.
First, we assume all the input string in $\q^\mu$ have the same probability.
Then, the values of different blocks become independent, and we can define the Fourier basis in $\cY$ consisting of the vectors $\ket Y|\mu, \chit \tau>$ with $\tau\in \q^\mu$.
Second, we get that
\begin{equation}
\label{eqn:overviewUpsilonOnSigma}
\Upsilon_\mu\colon 
\ket X|\chit \sigma> \mapsto \ket Y|\mu, \chit \tau>\quad\text{given by}\quad
\tau(B) = \sum_{i\in B} \sigma(i)\quad \text{for all $B\in\mu$}.
\end{equation}

Since each partition $\mu$ uniquely specifies all the inputs therein, it suffices to specify the collection $M$ of partitions in order to define the set of inputs $Y$ and the space $\cY = \bC^Y$.
We will assume that all partitions inside the collection have the same probability $p_\mu = 1/|M|$.
% and every input $(\mu, y)$ inside the partition $\mu$ has the same probability $p_{\mu, y} = q^{-|\mu|}/|M$, where $|\mu|$ is the number of blocks in $\mu$.

\subsubsection{Knowledge Operators}
\label{sec:overviewKDistKnowledge}

To define knowledge, we introduce \emph{highlighted} partitions.
A highlighted partition $\mu$ is the same partition as in \rf{sec:overviewPartition}, but with exactly one block highlighted.
Informally, the task of the algorithm is to find the highlighted block.
Therefore, the corresponding knowledge system $L^+_\mu$ consists of all supersets of the highlighted block, and we define the knowledge operator $\Upsilon^+_\mu$ as in~\rf{eqn:overviewKnowledge}.

As mentioned at the end of \rf{sec:overviewQueryGain}, we define several collections of highlighted partitions: $M_1,\dots, M_k$.
Each collection $M_\ell$ consists of all possible permutations of a single highlighted partition $\mu_\ell \in M_\ell$.
The collection $M_k$ corresponds to the $k$-Distinctness problem itself, and it is defined by a highlighted partition $\mu_k$ that has a unique block of size $k$, which is highlighted, as well as $\Omega(n)$ blocks of size $i$ for all $i=1,\dots,k-1$.
Since the highlighted block is the unique $k$-block of the partition, there is no principal difference between highlighted and usual partitions here.

\myfigure
\label{fig:overview3Dist}
=====
\[
\begin{tikzpicture}
\node at (-0.5,4*\rowgap) {$M_{3}$};
\node at (-0.5,3*\rowgap) {$M_{2}$};
\node at (-0.5,2*\rowgap) {$M_{1}$};
\fillblock413
\fillblock312
\drawblock333
\fillblock211
\drawblock222
\drawblock233
\foreach \y in {2,3,4}
{
    \foreach \x in {1,...,7,9,10,11,12,13,15} {\drawelement \y\x};
    \foreach \x/\z in {4/5,6/7,9/10} {\drawblock\y\x\z};
    \foreach \x in {11,12,13,15} {\drawblock\y\x\x};
    \draw (8*\nodegap, \y*\rowgap) node {$\cdots$};
    \draw (14*\nodegap, \y*\rowgap) node {$\cdots$};
}
\drawbrace44{10}{\Omega(n)}
\drawbrace4{11}{15}{\Omega(n)}
\end{tikzpicture}
\]
-----
For the 3-Distinctness problem, the collections $M_1$, $M_2$, and $M_3$ are obtained by all possible permutations of these highlighted partitions.
The highlighted blocks are filled in grey.
=====

But this is not the case for $M_\ell$ with $\ell = 1,\dots, k-1$.
The collection $M_\ell$ consists of all permutations of the highlighted partition $\mu_\ell$ defined as follows: take the highlighted $k$-block of $\mu_k$ and replace it by a highlighted $\ell$-block together with $k-\ell$ singletons, see \rf{fig:overview3Dist}.
Computationally, these collections do not seem to make much sense.
Indeed, given an input string $y\in \q^n$ corresponding to a highlighted partition $\mu$, the algorithm has absolutely no information which of the $\Omega(n)$ of its $\ell$-blocks is highlighted.
Nonetheless, there is a clear intuition here.
Each non-highlighted partition $\mu$ is counted in $\Upsilon^+_\ell \ket X|\chit \sigma>$ as many times as many $\ell$-tuples of equal elements are covered by the set $\supp(\sigma)$ of known input variables.
Thus, $\|\Upsilon_\ell^+ \psi_t\|^2$ can be informally interpreted as the average fraction of the $\ell$-tuples of equal elements found by the algorithm.

\subsubsection{Anti-Concentration}
\label{sec:overviewKDistAntiConcentration}

We have to prove anti-concentration only for the collection $M = M_k$ since the collections $M_\ell$ with $\ell < k$ are auxiliary, used only to bound query gain in $M_k$.
We will show that if $t$ is sublinear, then every vector in the subspace $\Upsilon^-\cX_{\le t}$ is $\OO\s[\frac1{\sqrt{n}}]$-anti-concentrated.
Let us explain what this means in this context.
Let $\rho = \{a_1,\dots,a_k\}$ be a $k$-subset of $[n]$, and denote by $M_\rho$ both the subset of partitions $\mu\in M$ that have $\rho$ as their unique $k$-block, as well as the orthogonal projector $M_\rho\colon \cY\to\cY$ onto the direct sum $\bigoplus_{\mu\in M_\rho} \cY_\mu$.
The claim is that
\begin{equation}
\label{eqn:overviewAnticoncentration}
\|M_\rho\Upsilon^-\phi\|^2 = \OO \s[1/n]\|\Upsilon^-\phi\|^2
\end{equation}
for every $\phi \in \cX_{\le t}$.
Before we proceed with anti-concentration of $\Upsilon^-\cX_{\le t}$, it makes sense to explain why the space $\Upsilon\cX_{\le t}$ is \emph{not} anti-concentrated, and, therefore, why we need the knowledge decomposition~\rf{eqn:overviewDecmposition} to begin with.

We will consider the simplest case $k=2$.
Then, each partition $\mu\in M$ contains exactly one pair and $n-2$ singletons.
We denote by $\mu_{a,b}$ the partition having $\{a,b\}$ as its unique pair.
We equip $M$ with the uniform probability distribution: $p_\mu = 1/\binom{n}{2}$ for all $\mu\in M$.
Consider the following vector:
\begin{equation}
\label{eqn:overlapExample}
\phi = \sum_{v = 1}^{q-1} \ket X|\chit {\{1\mapsto v,\, 2\mapsto -v\}}> \in \cX_{\le 2},
\end{equation}
where $\{1\mapsto v,\, 2\mapsto -v\}$ is the function mapping 1 to $v$, 2 to $-v$, and everything else to zero.
Applying $\Upsilon$, we have by \rf{eqn:overviewUpsilonDirectSum} and~\rf{eqn:overviewUpsilonOnSigma}:
\begin{equation}
\label{eqn:overlapExample2}
\Upsilon \phi = \frac1{\sqrt{\binom n2}} \sC[
(q-1) \ket Y|\mu_{1,2}, \chit \emptyset>
+
\sum_{\{a,b\}\ne \{1,2\}} \sum_{v=1}^{q-1} \ket Y|\mu_{a,b} ,\chit {\{1\hookmapsto v,\; 2\hookmapsto -v\}}>
],
\end{equation}
where $i\hookmapsto v$ denotes $B\mapsto v$ for the block $B\in \mu_{a,b}$ containing $i$, and $\emptyset$ stands for the identically zero function again.
We have $\norm| M_{1,2}\Upsilon\phi|^2 = (q-1)^2/\binom n2$, while the norm squared of the remaining part of the state is $\sA[\binom n2-1](q-1)/\binom n2 < q$. 
Under the assumption $q\gg n^2$, if we measure the normalised state~\rf{eqn:overlapExample2}, we obtain the output $\{1,2\}$ with overwhelming probability, hence, the vector $\Upsilon\phi \in \Upsilon\cX_{\le 2}$ is not anti-concentrated.

Recall that the knowledge system $L^+$ is defined by $L^+_{\mu_{a,b}}$ consisting of all supersets of $\{a,b\}$.
The state $\Upsilon \phi$ from~\rf{eqn:overlapExample2} gets decomposed into the sum of two vectors
\[
\Upsilon^+ \phi = \frac{q-1}{\sqrt{\binom n2}}
\ket Y|\mu_{1,2}, \chit \emptyset>
\qqand
\Upsilon^- \phi = \frac1{\sqrt{\binom n2}} \sum_{\{a,b\}\ne \{1,2\}} \sum_{v=1}^{q-1} \ket Y|\mu_{a,b} ,\chit {\{1\hookmapsto v,\; 2\hookmapsto -v\}}>.
\]
It is easy to see that the second of them is anti-concentrated, which hints that~\rf{eqn:overviewAnticoncentration} could be true.

The main reason why anti-concentration fails in the above example is that the vectors $\ket X|\chit {\{1\mapsto v,\, 2\mapsto -v\}}>$ gets mapped to the same vector $\ket Y|\mu_{1,2}, \chit \emptyset>$ by $\Upsilon_{\mu_{1,2}}$, while all other $\Upsilon_{\mu_{a,b}}$ map them to orthogonal vectors within $\cY_{\mu_{a,b}}$.
Let us explain why we avoid this issue for the operator $\Upsilon^-$.
Let $\phi = \sum_\sigma \alpha_\sigma \ket X|\chit \sigma>$, and fix some $\rho\in R$.
We can decompose $\phi$ as follows
\begin{equation}
\label{eqn:overviewphi}
\phi = \sum_{\sigma\colon \rho \subseteq \supp(\sigma)} \alpha_\sigma \ket X|\chit \sigma>
+
\sum_{i=1}^k \sum_{\sigma\colon \rho \setminus \supp(\sigma) = \{a_i,\dots\}} \alpha_\sigma \ket X|\chit \sigma>,
\end{equation}
where, in the last sum, we mean that $i$ is the minimal integer such that $a_i\in \rho\setminus \supp(\sigma)$.
Let $\phi_0$ be the first term in~\rf{eqn:overviewphi} and $\phi_i$ be the $i$-th term in the second sum.
For the first term, $M_\rho \Upsilon^- \phi_0 = 0$, so we can ignore it.
For the term $\phi_i$, we have that its contribution for any $\mu \in M_\rho$ is comparable to its contribution to $\mu_{a_i \leftrightarrow b}$ for any singleton $b$ of $\mu$, where $\mu_{a_i \leftrightarrow b}$ stands for the partition obtained from $\mu$ by exchanging $a_i$ and $b$.
Since, for every $\mu\in M_\rho$, there are $\Omega(n)$ choices of $b$, we get the $O(1/n)$ factor in~\rf{eqn:overviewAnticoncentration}.
We make this argument rigorous in \rf{sec:anticoncentration}, where we prove a rather general anti-concentration result.

\subsubsection{Query Gain}
\label{sec:overviewKDistQueryGain}
For the query gain in the $k$-Distinctness problem, we have the following estimate:
\begin{equation}
\label{eqn:overviewPsiPartial}
\|\Psi^\partial_{\ell} \psi_t \| 
\le 
\OO \s[\frac{1}{\sqrt n}]
\|\Upsilon^+_{\ell-1} \psi_t \|.
\end{equation}
To show this inequality, recall that $S\in L^{\partial i}_{\mu_\ell}$ for $\mu_\ell \in M_\ell$ if and only if $S \notin L^+_{\mu_\ell}$ but $S\cup\{i\} \in L^+_{\mu_\ell}$.
This implies that $S \in L^+_{\mu_{\ell-1}}$, where $\mu_{\ell-1}\in M_{\ell-1}$ is obtained from $\mu_\ell$ by removing $i$ from the highlighted block of $\mu_\ell$ and making it a singleton.
This already implies that $\|\Upsilon^{\partial i}_{\mu_\ell} \phi \| \le \| \Upsilon^{+}_{\mu_{\ell-1}} \phi\|$ for every $\phi\in\cX$.
The $1/\sqrt{n}$ factor in~\rf{eqn:overviewPsiPartial} is solely due to normalisation because $|M_{\ell}| \ge \Omega(n)|M_{\ell-1}|$.

Let $w = \Psi^\partial_{\ell} \psi_t $ be the vector from~\rf{eqn:overviewPsiPartial}.
In the usual query gain bound of \rf{eqn:overviewQueryGain}, we bound the difference between the norms of $\Upsilon^+_\ell \psi_{t}$ and $\Upsilon^+_\ell \psi_{t+1}$ essentially by the norm of $w$ using the triangle inequality.
Unfortunately, it is not sufficient to prove a tight lower bound for $k$-Distinctness, see the calculations in \rf{sec:3DPartI}.

To get a tight lower bound, we have to take into account the special structure of the vector $\Upsilon^+_\ell \psi_{t}$. 
Note that $\Upsilon^{\partial i}_\mu$ is zero if $i$ is not contained in the highlighted block of $\mu$.
(Indeed, this amounts to saying that, if $S$ does not contain the highlighted block of $\mu$ as a subset and $i$ is not in the highlighted block, then, certainly, $S\cup\{i\}$ still does not cover the whole highlighted block.)
Therefore, we can partition $\Upsilon^+_\ell \psi_t$ as $u+v$, where $u$ is the part of the state where $i$ is not in the highlighted block, while, in $v$, $i$ is contained in the highlighted block.
By the above discussion, $u$ is orthogonal to $w$.
Moreover $v$ has small norm: Since $\mu$ contains $\Omega(n)$ blocks of size $\ell$, each of which can be highlighted, we have $\|v\| = \OO\sA[1/\sqrt{n}]$.
A simple linear algebra shows that we essentially have the following refined query gain bound valid for all $\ell < k$:
\begin{equation}
\label{eqn:overviewRefined}
\|\Upsilon_\ell^+\psi_{t+1}\|^2 - \|\Upsilon_\ell^+\psi_{t}\|^2 
\le
\OO\s[\frac1{\sqrt n}] \|\Psi^\partial_{\ell} \psi_t\|.
\end{equation}

As an illustration, let us use this bound to evaluate $\|\Upsilon^+\psi_t\|$ for the case of $k=3$.
Using a variant of~\rf{eqn:overviewPsiPartial}, we have that
$
\norm|\Psi_1^\partial \psi_t|  = \OO\s[\frac1{\sqrt n}].
$
Applying the refined query gain bound~\rf{eqn:overviewRefined}, we obtain
$
\norm|\Upsilon_1^+ \psi_t|  = \OO\s[\sqrt {\frac tn}].
$
By~\rf{eqn:overviewPsiPartial} again,
$
\norm|\Psi_2^\partial \psi_t|  = \OO\s[\frac{\sqrt t}n].
$
By~\rf{eqn:overviewRefined},
$
\norm|\Upsilon_2^+ \psi_t|  = \OO\s[\frac{t^{3/4}}{n^{3/4}}].
$
Using~\rf{eqn:overviewPsiPartial} one final time,
$
\norm|\Psi_2^\partial \psi_t|  = \OO\s[\frac{t^{3/4}}{n^{5/4}}].
$
To find $\norm|\Upsilon_3^+ \psi_t|$, we use the usual query gain bound~\rf{eqn:overviewQueryGain}, which gives
\[
\norm|\Upsilon_3^+ \psi_t| = \OO\s[\frac{t^{7/4}}{n^{5/4}}].
\]
To have constant success probability, we should have $\norm|\Upsilon_3^+ \psi_t| = \Omega(1)$, which means that $t = \Omega(n^{5/7})$, which agrees to \rf{thm:main}.
The case of general $k$ is similar, see \rf{sec:kDist}.

%%%%%%%%%%%%%%%%%%%
\section{Framework}
\label{sec:framework}

In this section, we lay out the basics of our framework for proving quantum query lower bounds.
The \emph{input string} to the algorithm is given by $x = (x_1,\dots,x_n)\in \q^n$, and the algorithm has query access to this string using the input oracle $O_x$ as described in \rf{sec:modelAlgorithm}.
For the purposes of this section, we fix a relational problem.
It is given by a set $R$ of possible \emph{responses} of the algorithm, and, for each input $x$, there is a set of \emph{correct outputs} $R_x\subseteq R$.
The goal of the algorithm is, given oracle access to the input string $x$, output a response $\rho\in R_x$.

We consider average-case settings.
Each input is assigned probability $p_x$, and the success probability of the algorithm is measured with respect to this probability distribution.
Lower bounds in the average-case settings imply the same lower bounds in the worst-case settings, but it is often tricky to come up with the right probability distribution to prove the lower bound.

\subsection{Preliminaries}
\label{sec:prelim}
For a positive integer $n$, we use $[n]$ to denote the set $\{1,\dots,n\}$, and $\bZ_n$ to denote the set $\{0,\dots,n-1\}$.
We treat the latter as an additive group modulo $n$.
For a predicate $P$, we use $1_P$ to denote the indicator variable that equals 1 if $P$ is true, and 0 otherwise.

We use calligraphic letters like $\cX$ or $\cH$ to denote complex inner product spaces.
We generally do not use the ket notation to write vectors.
However, there are three cases when we adopt it.
First, to signify that a particular vector belongs to a particular quantum register, i.e., $\ket A|\psi>$ denotes the vector $\psi$ located in the register $\cA$.
Second, to denote elements of the computational basis like $\ket I|i>$.
Third, to denote elements of the Fourier basis like $\ket X|\chit \sigma>$.

Let us introduce some nomenclature related to the set $\q^A$ of functions $\sigma\colon A\to \q$.
We will extensively use such functions with $A = [n]$ starting with this section, and with $A$ being a partition of $[n]$ starting with \rf{sec:equalElements}.
We identify such functions with sets of mappings
\begin{equation}
\label{eqn:functionOnA}
\sigma = \sfigB{ a_1\mapsto v_1, a_2\mapsto v_2, \dots, a_\ell\mapsto v_\ell },
\end{equation}
with $a_i\in A$ all distinct and $v_i\in \q\setminus\{0\}$.
This notation means that $\sigma(a_i) = v_i$ for all $i\in[\ell]$, and $\sigma(a) = 0$ for all $a\in A\setminus\{a_1,\dots,a_\ell\}$.
(The corresponding $A$ can be deduced from the context.)
In particular, $\emptyset$ stands for the identically zero function.
We are ignoring mappings of the form $a\mapsto 0$.
Adding (or removing) such an assignment to the list in~\rf{eqn:functionOnA} will not change the latter.
The \emph{support} of $\sigma$ is $\supp(\sigma) = \{a_1,\dots,a_\ell\}$: the set of elements of $A$ on which $\sigma$ is non-zero.
We write $|\sigma|$ for the size of its support, $\ell$.

We also use the representation in~\rf{eqn:functionOnA} to modify such functions.
E.g., $(\sigma\setminus\{a\mapsto v\})\cup \{b\mapsto w\}$ denotes the function $\sigma'$ that agrees to $\sigma$ everywhere except $a$ and $b$, where $\sigma'(a)=0$ and $\sigma'(b)=w$ (unless $a=b$, in which case, we simply have $\sigma'(a)=w$).
When we use such notation, we always guarantee that the things removed were actually in the value table of the function, and there are no contradictions: e.g., in the above example, it is actually the case that $\sigma(a)=v$ and $\sigma(b)=0$ (assuming $a\ne b$).

\subsection{
\texorpdfstring
	{Quantum Query Algorithm and the Space $\cA$}
	{Quantum Query Algorithm and the Space A}
}
\label{sec:modelAlgorithm}

In this section, we define our model of quantum query algorithms.
We follow a standard model with the input oracle in the phase, which is equivalent to the input oracle in the register up to a change of basis.

\myfigure
\label{fig:alg1}
=====
\[
\begin{quantikz}
\lstick[3]{$\ket A|0>$} 
&\lstick{$\cI$}\slice{$\psi_{x,0}$}
& \gate[3]{U_0}       & \gate[2]{O_x}      & \gate[3]{U_1}       & \gate[2]{O_x}      & \midstick[3, brackets=none]{$\cdots$} \qw & \gate[3]{U_{T-1}}     & \gate[2]{O_x}      & \gate[3]{U_T}       & \qw &\\
&\lstick{$\cC$}
&\slice{$\psi'_{x,0}$}&\slice{$\psi_{x,1}$}&\slice{$\psi'_{x,1}$}&\slice{$\psi_{x,2}$}& \qw  \slice{$\psi_{x,T-1}$}               &\slice{$\psi'_{x,T-1}$}&\slice{$\psi_{x,T}$}&\slice{$\psi'_{x,T}$}& \qw &\\
&\lstick{$\cW$}
&                     & \qw                & \qw                 & \qw                & \qw                                       &                       & \qw                &                     & \qw &
\end{quantikz}
\]
-----
The quantum query algorithm from~\rf{eqn:program}.
The sequence of states $\psi_{x,t}$ and $\psi'_{x,t}$ from~\rf{eqn:algorithmStates} and~\rf{eqn:algorithmSequenceOfStates} is depicted.
=====

The space of the quantum query algorithm is $\cA = \cI\otimes \cC\otimes \cW$, where $\cI$ is an $n$-qudit storing the index of the queried input variable, 
$\cC$ is an $q$-qudit,
and $\cW$ is an arbitrary working register.
The algorithm is a unitary in $\cA$ of the form (see also \rf{fig:alg1})
\begin{equation}
\label{eqn:program}
U_T O_x U_{T-1} \cdots O_x U_1 O_x U_0.
\end{equation}
Here $U_t$ are arbitrary input-independent unitaries in $\cA$, and $O_x$ is the query to the input string 
$x = (x_1,\dots,x_n) \in\q^n$.
We assume the query works in the phase, i.e., it is a unitary operator defined by
\begin{equation}
\label{eqn:oracle}
O_x\colon \cI\otimes \cC\to \cI\otimes \cC \colon \quad \ket I|i> \ket C|c> \mapsto \omega_q^{cx_i} \ket I|i> \ket C|c>,
\end{equation}
where $\omega_q = \ee^{2\pi\ii/q}$ is the primitive $q$-th root of unity.
The \emph{query complexity} of the algorithm is $T$, the number of times $O_x$ is executed.
We will always assume $T\le n$, because any problem can be solved in $n$ queries.

We use $\ket A|0>$ to denote the initial state of the algorithm.
We define the intermediate states $\psi_{x,t}$ and $\psi'_{x,t}$ of the algorithm~\rf{eqn:program} on the input $x$ as follows.
First, $\psi_{x,0} = \ket A|0>$, and then, recursively, for $t=0,\dots, T$:
\begin{equation}
\label{eqn:algorithmStates}
\psi'_{x,t} = U_t \psi_{x,t}
\qqand
\psi_{x,t+1} = O_x \psi'_{x,t}.
\end{equation}

Thus, $\psi_{x,t}$ is the state after the $t$-th query in~\rf{eqn:program}, and $\psi'_{x,t}$ before the $(t+1)$-st query.
In other words, the algorithm proceeds as follows:
\begin{equation}
\label{eqn:algorithmSequenceOfStates}
\psi_{x,0} \maps{U_0} \psi'_{x,0} \maps{O_x} \psi_{x,1} \maps{U_1} \psi'_{x,1} \maps{O_x} \cdots\maps{U_{T-1}} \psi'_{x,T-1} \maps{O_x} \psi_{x,T} \maps{U_{T}} \psi'_{x,T}.
\end{equation}
In particular, $\psi'_{x,T}$ is the final state of the algorithm.
%(And we do not need $\psi_{x,T+1}$).

At the end, a measurement is performed in $\cA$ that yields an output $\rho\in R$.
Without loss of generality, we may assume it is an orthogonal measurement given by $\{\Pi_\rho\}$ with $\rho$ ranging over the set $R$ of responses.
Thus,
\begin{equation}
\label{eqn:successProbabilityIndividualInput}
\Pr\skA[\text{The algorithm outputs $\rho$ on input $x$}] = \norm|{ \Pi_\rho \psi'_{x,T} }|^2.
\end{equation}

From now on, we will fix the algorithm~\rf{eqn:program}.
The idea is that it is going to be an arbitrary quantum algorithm for which we prove small success probability (assuming the number of queries $T$ is small).

\subsection{
\texorpdfstring
	{Uniform Probability Distribution and the Space $\cX$}
	{Uniform Probability Distribution and the Space X}
}
\label{sec:uniform}
In this section, we consider the action of the algorithm assuming the uniform probability distribution on all input strings.
Later, we will transfer it to arbitrary probability distributions using the transfer operators.
The content of this section is independent of the problem being solved.

We define the register
\begin{equation}
\label{eqn:X}
\cX = \cX_1\otimes \cdots \otimes \cX_n
\end{equation}
storing the input string, and which is not directly accessible to the algorithm.
Here each $\cX_i$ is isomorphic to $\bC^q$ (stores an element of $\q$).
Thus, the whole register $\cX$ stores the input string $x\in \q^n$.
The algorithm is executed on the initial state
\begin{equation}
\label{eqn:zhandryInitialOriginal}
\psi_0 = \frac{1}{\sqrt{q^n}} \sum_{x\in \q^n} \ket X|x>\otimes \ket A|0>
\end{equation}
involving the uniform superposition over all inputs in $\cX$ (corresponding to the uniform probability distribution).
We define the combined input oracle $O$ that simulates the input oracle $O_x$ from~\rf{eqn:oracle} for all $x$ by interweaving registers $\cX$ and $\cA$.
Formally, it is a unitary operator
\begin{equation}
\label{eqn:zhandryOracleStandard}
O\colon \cX\otimes \cI\otimes \cC \to \cX\otimes \cI\otimes \cC
\colon\quad
\ket X|x>\ket I|i> \ket C|c> \mapsto 
\omega_q^{cx_i} \ket X|x>\ket I|i> \ket C|c>,
\end{equation}
where $x_i$ is the $i$-th symbol of $x=(x_1,\dots,x_n)$.
The operator $O$ is the direct sum $O = \bigoplus_x O_x$ with $O_x$ as in~\rf{eqn:oracle}
It is often convenient to write it as a different direct sum $O = \bigoplus_{i,c} O_{i,c}$ of unitary operators
\begin{equation}
\label{eqn:zhandryOic}
O_{i,c}\colon \cX\to \cX\colon \quad \ket X|x> \mapsto \omega_q^{cx_i}\ket X|x>.
\end{equation}
The action of the algorithm~\rf{eqn:program} on the combined space $\cX\otimes\cA$ is given by (see also \rf{fig:alg2})
\begin{equation}
\label{eqn:zhandryAlgorithm}
U_T O U_{T-1} \cdots O U_1 O U_0.
\end{equation}
This is similar to~\rf{eqn:program} but with $O_x$ replaced by $O$.
Note that we generally omit identities on extra registers, therefore, $U_t$ in~\rf{eqn:zhandryAlgorithm} should, more strictly, be $I_{\cX}\otimes U_t$, and $O$ should be $O\otimes I_{\cW}$.

\myfigure
\label{fig:alg2}
=====
\[
\begin{quantikz}
\lstick{$\frac{1}{\sqrt{q^n}} \sum_{x\in \q^n} \ket |x>$}
&\lstick{$\cX$}
& \qw                 & \gate[3]{O}        & \qw                 & \gate[3]{O}        & \midstick[4, brackets=none]{$\cdots$} \qw & \qw                   & \gate[3]{O}        & \qw                 & \qw &\\
\lstick[3]{$\ket A|0>$} 
&\lstick{$\cI$}\slice{$\psi_{0}$}
& \gate[3]{U_0}       &                    & \gate[3]{U_1}       &                    &                                       \qw & \gate[3]{U_{T-1}}     &                    & \gate[3]{U_T}       & \qw &\\
&\lstick{$\cC$}
&\slice{$\psi'_{  0}$}&\slice{$\psi_{  1}$}&\slice{$\psi'_{  1}$}&\slice{$\psi_{  2}$}& \qw  \slice{$\psi_{T-1}$}                 &\slice{$\psi'_{  T-1}$}&\slice{$\psi_{  T}$}&\slice{$\psi'_{  T}$}& \qw &\\
&\lstick{$\cW$}
&                     & \qw                & \qw                 & \qw                & \qw                                       &                       & \qw                &                     & \qw &
\end{quantikz}
\]
-----
A quantum algorithm from~\rf{eqn:zhandryAlgorithm} combining 
the quantum algorithm from \rf{fig:alg1} with the extra register $\reg X$ storing input strings.
The latter is initialised to the uniform superposition over all inputs in $\q^n$.
The states $\psi_t$ and $\psi'_t$ from \rf{defn:algorithmStates} and~\rf{eqn:zhandrySequenceOfStates} are depicted.
=====

We call the states the algorithm~\rf{eqn:zhandryAlgorithm} goes through on the initial state $\psi_0$ of~\rf{eqn:zhandryInitialOriginal} the \emph{uniform states of the algorithm}.

\begin{defn}[Uniform States of the Algorithm]
\label{defn:algorithmStates}
For $t=0,\dots, T$, define the following states recursively starting with $\psi_0$ as in~\rf{eqn:zhandryInitialOriginal}:
\begin{equation}
\label{eqn:zhandryStates}
\psi_t' = U_t \psi_t
\qqand
\psi_{t+1} = O \psi'_t.
\end{equation}
\end{defn}

By construction, 
\begin{equation}
\label{eqn:zhandryToIndividual}
\psi_t = \frac{1}{\sqrt{q^n}} \sum_{x\in \q^n} \ket X|x> \ket A|\psi_{x,t}>
\qqand
\psi'_t = \frac{1}{\sqrt{q^n}} \sum_{x\in \q^n} \ket X|x> \ket A|\psi'_{x,t}>
\end{equation}
with $\psi_{x,t}$ and $\psi'_{x,t}$ being the states of the algorithm on individual inputs as in~\rf{eqn:algorithmStates}.
Again, we have the following sequence of states:
\begin{equation}
\label{eqn:zhandrySequenceOfStates}
\psi_0 \maps{U_0} \psi'_0 \maps{O} \psi_1 \maps{U_1} \psi'_1 \maps{O} \cdots\maps{U_{T-1}} \psi'_{T-1} \maps{O} \psi_{T} \maps{U_{T}} \psi'_{T}.
\end{equation}

Let us now define the Fourier basis.
The \emph{Fourier basis in $\bC^q$} consists of the following vectors
\begin{equation}
\label{eqn:FourierDef}
\ket |\chit a> = \frac1{\sqrt q} \sum_{b\in \q} \omega_q ^{ab} \ket |b>,
\end{equation}
where $a$ ranges over $\q$.
The most important one is $\ket |\chit 0>$, which is the uniform superposition in $\bC^q$.
The \emph{Fourier basis in $\cX$} is given by
\begin{equation}
\label{eqn:FourierSigma}
\ket X|\chit \sigma> = \ket X_1|\chit {\sigma(1)}>\otimes \cdots\otimes \ket X_n|\chit {\sigma(n)}>,
\end{equation}
where $\sigma$ ranges over all the functions in $\q^n$.
In particular, the uniform initial state $\psi_0$ of the algorithm 
can be written as
\begin{equation}
\label{eqn:zhandryInitialFourier}
\psi_0 = \ket X|\chit \emptyset>\ket A|0>,
\end{equation}
where $\emptyset$ is the identically zero function on $[n]$ (per our notational conventions of \rf{sec:prelim}).
\medskip

We will now describe the phase kickback trick.
Consider the constituent $O_{i,c}$ of the input oracle defined in~\rf{eqn:zhandryOic}.
It is a tensor product of the following operator in $\cX_i$:
\[
O_{i,c}
\colon \cX_i\to \cX_i\colon\quad
\ket X_i |x_i> \mapsto \omega_q^{cx_i }\ket X_i|x_i>
\]
and the identities in the remaining $\cX_j$.
Applying this to a Fourier basis state in $\cX_i = \bC^q$:
\[
O_{i,c} 
\colon \cX_i\to \cX_i\colon\quad
\frac1{\sqrt q} \sum_{b\in \q} \omega_q ^{ab} \ket X_i |b>
\mapsto 
\frac1{\sqrt q} \sum_{b\in \q} \omega_q ^{(a+c)b}  \ket X_i |b>,
\]
which reads in the the Fourier basis as 
$
O_{i,c} \colon \ket X_i |\chit a>
\mapsto 
\ket X_i|\chit {a+c}>,
$
where addition is performed in $\bZ_q$.
Attaching the remaining registers from $\cX$, we get the following action of $O_{i,c}$ on the Fourier basis of $\cX$:
\begin{equation}
\label{eqn:zhandryOicFourier}
O_{i,c}
\colon \cX\to \cX\colon\quad
\ket X |\chit \sigma> 
\mapsto 
\ket X |\chit {\sigma+\{i\mapsto c\}}>,
\end{equation}
where the function $\{i\mapsto c\}$ maps $i$ to $c$ and is zero everywhere else.
Addition of functions in~\rf{eqn:zhandryOicFourier} is performed element-wise.
For the whole input oracle $O$, we have
\begin{equation}
\label{eqn:zhandryOracleFourier}
O \colon \cX\otimes \cI\otimes \cC \to\cX\otimes \cI\otimes \cC
\colon\quad\ket X |\chit \sigma> \ket I|i> \ket C|c>
\mapsto 
\ket X |\chit {\sigma+\{i\mapsto c\}}> \ket I|i> \ket C|c>.
\end{equation}
Thus, we see that the input oracle $O$ acts in $\cX$ by performing addition in the Fourier basis.
\medskip

One can see from~\rf{eqn:zhandryOicFourier} that (assuming all $a_i$ are pairwise distinct):
\[
O_{a_t, c_t}\cdots O_{a_2,c_2} O_{a_1,c_1} \ket X|\chit \emptyset>
=
\ket X|\chit {\{a_1\mapsto c_1, a_2\mapsto c_2,\dots, a_t\mapsto c_t\}}>.
\]
From this, we draw intuition that a vector $\ket X|\chit \sigma>$ corresponds to the algorithm ``knowing'' the values of the input variables in the support of $\sigma$, and not knowing the remaining ones.
Motivated by this, for a non-negative integer $t$, we define the following subspace of $\cX$, which corresponds to the algorithm ``knowing'' the values of at most $t$ input variables: 
\begin{equation}
\label{eqn:Xlet}
\cX_{\le t} = \spn_{\sigma\in \q^n\colon |\sigma|\le t} \ket X|\chit \sigma>.
\end{equation}

These subspaces form an increasing sequence as $t$ grows:
\[
\cX_{\le 0} \subseteq \cX_{\le 1} \subseteq \cX_{\le 2} \subseteq \cdots \subseteq \cX_{\le t} \subseteq \cdots.
\]

\begin{prp}
\label{prp:inH}
For the uniform states of the algorithm, we have: $\psi_t, \psi'_t \in \cX_{\le t}\otimes \cA$.
\end{prp}

\begin{proof}
By \rf{eqn:zhandryInitialFourier}, $\psi_0 \in \cX_{\le 0}\otimes \cA$, which is the base of our induction on $t$.

Now assume that $\psi_t\in \cX_{\le t}\otimes \cA$ for some value of $t$, and we will prove a similar statement for $t+1$.
An application of the unitary $U_t$ in $\cA$ does not move the state outside of this subspace, hence, $\psi'_t\in \cX_{\le t}\otimes \cA$ and can be written as
\[
\psi'_t = \sum_{\sigma\colon |\sigma|\le t} \sum_{i,c} \ket X|\chit \sigma> \ket I|i> \ket C|c> \ket W|w_{\sigma,i,c}>
\]
for some (non-normalised) vectors $w_{\sigma,i,c}$.  By~\rf{eqn:zhandryOracleFourier}, we get that the state of the algorithm after the query is
\[
\psi_{t+1} = O\psi'_t = \sum_{\sigma\colon |\sigma|\le t} \sum_{i,c} \ket X|\chit {\sigma+\{i\mapsto c\}}> \ket I|i> \ket C|c> \ket W|w_{\sigma,i,c}>.
\]
As $\absA|\sigma+\{i\mapsto c\}| \le |\sigma|+1 \le t+1$, we get $\psi_{t+1}\in \cX_{\le t+1}\otimes \cA$.
\end{proof}

\mycutecommand{\S}{\mathfrak{S}}

\subsection{
\texorpdfstring
	{Formulation of the Problem and the Space $\cY$}
	{Formulation of the Problem and the Space Y}
}
\label{sec:spaceY}

In this section, we define the problem being solved, and introduce the notion of types of inputs.
We formalise this via the following general version of a computational problem.

\begin{defn}[Hidden Computational Problem]
\label{defn:problem}
A \emph{hidden computational problem} is given by the following:
a collection $M$ of types; 
a set $Y\subseteq M\times \q^n$ of inputs;
a set $R$ of responses; 
and, for each $\mu\in M$, a set $R_\mu \subseteq R$ of correct answers.
For a fixed input $(\mu,y)\in Y$, the algorithm solving the problem has query access to the input string $y$, and its \emph{success probability} on $(\mu,y)$ is the probability it outputs $\rho\in R_\mu$.
In the worst-case settings, the total success probability is the minimum over all inputs in $Y$.
In the average-case settings, the total success probability is its average success probability on an input sampled from a given probability distribution $p$ on $Y$.
\end{defn}

We call this problem hidden, because the output depends on $\mu\in M$, called the \emph{type} of the input, but the algorithm only has access to $y$ that is correlated to $\mu$.
It is reminiscent of a number of hidden problems in quantum algorithms, most famous being the hidden subgroup problem~\cite{kitaev:HSP}.
If the problem has nothing to hide, this can be achieved by taking $M=D$ and $Y = \sfigA{(y,y) \mid y\in D}$ with $D$ being the domain of the problem.
It is allowed for an input string $y\in \q^n$ to appear several times in $Y$ with different types, moreover, $R_\mu$ can be different for different $(\mu, y)\in Y$ with fixed $y$.

We consider relational problems (the same input can have several correct outputs), which are more general than functional problems (an input can only have one correct output).
We prove lower bounds in the average-case, which imply the same lower bounds in the worst-case.
Using~\rf{eqn:successProbabilityIndividualInput}, we get the following formula for the total success probability of the algorithm in the average-case settings:
\begin{equation}
\label{eqn:successProbability}
\sum_{(\mu, y)\in Y} p_{\mu,y} \sum_{\rho\in R_\mu} \normA|{ \Pi_\rho \psi'_{y,T} }|^2.
\end{equation}

Similarly as we fixed the algorithm in \rf{sec:modelAlgorithm}, in this section, we fix an instance of a hidden computational problem from \rf{defn:problem}.
\medskip

Mimicking~\rf{sec:uniform}, we introduce another register $\cY$ inaccessible to the algorithm directly.
Its standard basis is given by the states $\ket Y|\mu, y>$ with $(\mu, y)\in Y$.
In exactly the same way as in~\rf{eqn:zhandryOracleStandard} and~\rf{eqn:zhandryOic}, we define
\begin{equation}
\label{eqn:Yoracle}
O \colon \cY\otimes \cI\otimes \cC \to\cY\otimes \cI\otimes \cC\colon\quad
\ket Y |\mu,y> \ket I|i> \ket C|c> 
\mapsto 
\omega_q^{cy_i} \ket Y |\mu, y>  \ket I|i> \ket C|c>
\end{equation}
and
\begin{equation}
\label{eqn:YOic}
O_{i,c}\colon \cY\to\cY\colon \quad \ket Y |\mu, y> \mapsto \omega_q^{cy_i} \ket Y |\mu, y>.
\end{equation}
Note that we use the same letter $O$ for the input oracle as we did in~\rf{eqn:zhandryOracleStandard} and~\rf{eqn:zhandryOic}. 
We interpret it as $O$ being the direct sum of the input oracles over all $x\in \q^n$ and $(\mu,y)\in Y$, and we restrict $O$ as necessary.
The action of the algorithm on the combined space $\cY\otimes \cA$ is still given by~\rf{eqn:zhandryAlgorithm}.

For $\mu\in M$, we define $Y_\mu = \{y\mid (\mu,y)\in Y\}$ and $\cY_\mu$ as a subspace of $\cY$ spanned by all $\ket Y|\mu, y>$ with $y\in Y_\mu$.
Also, let 
\begin{equation}
\label{eqn:inducedProbabilities}
p_\mu = \sum_{y\in Y_\mu} p_{\mu, y}
\qqand
p'_{\mu,y} = p_{\mu,y}/p_{\mu},
\end{equation}
which give probability distributions $(p_\mu)_{\mu\in M}$ on $M$ and, for each $\mu\in M$, $(p'_{\mu,y})_{y\in Y_\mu}$ on $Y_\mu$.
We assume here that all $p_\mu$ are non-zero.

\begin{defn}
[Transfer Operators]
\label{defn:transfer}
The \emph{transfer operators} are linear operators given by
\begin{align}
\Upsilon\colon \cX\to\cY\colon\quad &
\frac{1}{\sqrt{q^n}} \ket X|x> \mapsto \sum_{(\mu, y)\in Y\colon y=x} \sqrt{p_{\mu, y}}\ket Y|\mu, y>.
\label{eqn:UpsilonStandard}
\\
\label{eqn:UpsilonMuStandard}
\Upsilon_\mu\colon \cX\to\cY_\mu\colon\quad &
\frac{1}{\sqrt{q^n}} \ket X| x> \mapsto \sqrt{p'_{\mu,x}} \ket Y|\mu, x>,
\end{align}
for all $x\in \q^n$.
In~\rf{eqn:UpsilonMuStandard}, we assume that $p'_{\mu,x}=0$ for $(\mu, x)\notin Y$.
\end{defn}

Thus, $\Upsilon$ transfers $\cX$ to the whole space $\cY$, while individual $\Upsilon_\mu$ only work in the confines of a fixed $\mu\in M$.
They are related by the following identity:
\begin{equation}
\label{eqn:UpsilonDirectSum}
\Upsilon\colon \cX\to\cY\colon\quad
\phi \mapsto \bigoplus_{\mu\in M} \sqrt{p_\mu}\; \Upsilon_\mu \phi.
\end{equation}

Another useful property is as follows.

\begin{fact}
\label{fact:Upsilon Oic commute}
The transfer operators commute with the input oracle: For all $i,c$ and $\mu$, it holds that 
$O_{i,c}\Upsilon = \Upsilon O_{i,c}$
and
$O_{i,c}\Upsilon_\mu = \Upsilon_\mu O_{i,c}$.
\end{fact}

\begin{proof}
For the second one, we just have by~\rf{eqn:zhandryOic}, \rf{eqn:YOic}, and~\rf{eqn:UpsilonMuStandard} that, for all $x\in\q^n$:
\[
\frac1{\sqrt{q^n}} O_{i,c}\Upsilon_\mu \ket X|x> 
=
\frac1{\sqrt{q^n}} \Upsilon_\mu O_{i,c} \ket X|x>
=
\omega_q^{cx_i} \sqrt{p'_{\mu,x}}  \ket Y|\mu, x>.
\]
The first one is shown similarly.
\end{proof}

The main point of the transfer operators is that they allow us to transform uniform states of the algorithm from $\cX$ to $\cY$ or $\cY_\mu$, which use different probability distributions.
This is formalised by the following trivial proposition.

\begin{prp}
\label{prp:UpsilonProperty}
Let $\mu$ be a type and let $t$ be an integer.
Then, whenever the uniform state of the algorithm $\psi_t$ from \rf{defn:algorithmStates} is defined,
$\Upsilon \psi_t$ and $\Upsilon_\mu \psi_t$ are the states of the algorithm~\rf{eqn:zhandryAlgorithm} after the $t$-th query on the initial states
\[
\sum_{(\mu, y)\in Y} \sqrt{p_{\mu,y}} \ket Y |\mu, y>\ket A|0>
\qqand
\sum_{y\in Y_\mu} \sqrt{p'_{\mu,y}}\ket Y |\mu, y>\ket A|0>,
\]
respectively.
Similarly, $\Upsilon \psi'_t$ and $\Upsilon_\mu \psi'_t$ give the states before the $(t+1)$-st query on the same initial states.
\end{prp}

In other words, similarly to~\rf{eqn:zhandryToIndividual}:
\begin{align}
\label{eqn:UpsilonStates1}
\Upsilon \psi_t &= \sum_{(\mu, y)\in Y} \sqrt{p_{\mu,y}} \ket Y |\mu, y>\ket A|\psi_{y,t}>
,
&&&\Upsilon_\mu \psi_t &=\sum_{y\in Y_\mu} \sqrt{p'_{\mu,y}}\ket Y |\mu, y>\ket A|\psi_{y,t}>,\\
\label{eqn:UpsilonStates2}
\Upsilon \psi'_t &= \sum_{(\mu, y)\in Y} \sqrt{p_{\mu,y}} \ket Y |\mu, y>\ket A|\psi'_{y,t}>
,
&\text{and}&&\Upsilon_\mu \psi'_t &=\sum_{y\in Y_\mu} \sqrt{p'_{\mu,y}}\ket Y |\mu, y>\ket A|\psi'_{y,t}>.
\end{align}
In particular, for all $\mu\in M$ and all $t$:
\begin{equation}
\label{eqn:UpsilonNorms}
\|\Upsilon \psi_t\| = \|\Upsilon_\mu \psi_t\| = 
\|\Upsilon \psi'_t\| = \|\Upsilon_\mu \psi'_t\| = 1.
\end{equation}

We use $M_\rho$ to denote the set of all $\mu\in M$ such that $\rho \in R_\mu$.
Slightly abusing notation, we use the same letter $M_\rho\colon \cY\to\cY$ for the orthogonal projector onto the span of $\ket Y|\mu,y>$ with $\mu \in M_\rho$.
Recall also from~\rf{sec:modelAlgorithm} that $\{\Pi_\rho\}$ is the final measurement of the algorithm.

\begin{clm}
\label{clm:successProbability}
The total success probability of the algorithm~\rf{eqn:zhandryAlgorithm} on the problem from \rf{defn:problem} equals 
$
\norm|\Pi \Upsilon \psi'_T|^2,
$
where $\Pi\colon \cY\otimes \cA \to \cY\otimes \cA$ is an orthogonal projector given by
\begin{equation}
\label{eqn:Pi}
\Pi = \sum_{\rho\in R} M_\rho\otimes \Pi_\rho.
\end{equation}
\end{clm}

\begin{proof}
Using~\rf{eqn:UpsilonStates2} and that the projectors $\Pi_\rho$ are mutually orthogonal:
\[
\norm|\Pi \Upsilon \psi'_T|^2 
=\sum_{\rho\in R} \normC|{ (M_\rho\otimes \Pi_\rho)  \sum_{(\mu, y)\in Y} \sqrt{p_{\mu,y}} \ket Y |\mu, y>\ket A|\psi'_{y,T}>}|^2
= \sum_{\rho\in R}\; \sum_{(\mu, y)\in Y\colon \rho\in R_\mu} p_{\mu,y} \normA|\Pi_\rho \psi'_{y,T}|^2,
\]
which equals~\rf{eqn:successProbability} with inverted order of summation.
\end{proof}

\subsection{Anti-concentration and Small Success Probability}
\label{sec:overallAntiConcentration}

Our next goal is an upper bound on the success probability of the algorithm given by \rf{clm:successProbability}.
The following notion is of importance for this task.

\begin{defn}[Anti-concentration]
\label{defn:anticoncentration}
Let $M$, $R$, $\cY$, and $M_\rho\colon \cY\to\cY$ be as in \rf{sec:spaceY}.
Then, for a real number $0<\gamma<1$, we say that a vector $\phi$ in $\cY$ or $\cY\otimes \cA$ is \emph{$\gamma$-anti-concentrated} if $\|M_\rho\phi\|\le \gamma\|\phi\|$ for all $\rho\in R$.
We say that a subspace $\cH\subseteq \cY$ is $\gamma$-anti-concentrated if all vectors in $\cH$ are $\gamma$-anti-concentrated.
\end{defn}

Let us note that while the set of $\gamma$-anti-concentrated vectors in $\cY$ is closed under multiplication by scalars, it is \emph{not} closed under addition, and does \emph{not} form a subspace in general.
As a simplified example, the elements of the Fourier basis in $\bC^q$ are $1/\sqrt{q}$-anti-concentrated with respect to the standard basis, but the whole space $\bC^q$ clearly contains vectors that do not possess this property.

The notion of anti-concentration as in \rf{defn:anticoncentration} turns out to be too weak for the vectors in $\cY\otimes\cA$, and we need the following its strengthening.

\begin{defn}[Strong anti-concentration]
\label{defn:strongAntiConcentration}
We say that a vector $\psi\in\cY\otimes \cA$ is \emph{strongly $\gamma$-anti-concentrated}, if $\psi\in \cH\otimes \cA$ with $\cH$ being a $\gamma$-anti-concentrated subspace of $\cY$.
\end{defn}

For clarity, let us explicitly state out the following:
\begin{clm}
\label{clm:strongAnticoncentration}
If $\psi\in \cY\otimes \cA$ is strongly $\gamma$-anti-concentrated, then it is $\gamma$-anti-concentrated.
\end{clm}

\begin{proof}
Let $\{v_i\}$ be an orthonormal basis in $\cA$.
We can write
$
\psi = \sum_i \ket Y|\phi_i> \ket A |v_i>,
$
where all $\phi_i$ are $\gamma$-anti-concentrated.
Then:
\[
\|M_\rho \psi\|^2 
= \sum_i \|M_\rho \phi_i\|^2 
\le \gamma^2 \sum_i \|\phi_i\|^2
= \gamma^2 \|\psi\|^2.
\qedhere
\]
\end{proof}

The crux of \rf{defn:strongAntiConcentration} is the following result.
\begin{prp}
\label{prp:antiConcentrationMeasurement}
Assume $\psi\in \cY\otimes \cA$ is strongly $\gamma$-anti-concentrated.
Then, for every choice of an orthogonal measurement $\{\Pi_\rho\}_{\rho\in R}$ in $\cA$, we have
\[
\|\Pi\psi\| \le \gamma \|\psi\|,
\]
where $\Pi = \sum_{\rho\in R} M_\rho\otimes \Pi_\rho$ is as defined in~\rf{eqn:Pi}.
\end{prp}

\begin{proof}
Let $\cH$ be a $\gamma$-anti-concentrated subspace of $\cY$ such that $\psi \in \cH\otimes \cA$.
Write an orthogonal decomposition $\psi = \sum_{\rho\in R} u_\rho$, where each $u_\rho \in \cH\otimes \Pi_\rho$ is also strongly $\gamma$-anti-concentrated.
Then, using~\rf{clm:strongAnticoncentration}:
\[
\|\Pi \psi\|^2
=
\sum_{\rho\in R} \|M_\rho u_\rho\|^2
\le
\sum_{\rho\in R} \gamma^2\|u_\rho\|^2
= \gamma^2 \|\psi\|^2.
\qedhere
\]
\end{proof}

\rf{prp:antiConcentrationMeasurement} shows that if the final state of the algorithm is strongly anti-concentrated, the algorithm has small success probability.
The next lemma states that the same holds if the final state is close to being strongly anti-concentrated.

\begin{lem}[Success Probability of the Algorithm]
\label{lem:successProbability}
Assume the final state $\Upsilon\psi_T'\in \cY\otimes \cA$ of the algorithm can be decomposed as $\Upsilon\psi_T' = u + v$, where $\|u\|\le \delta$ and $v$ is strongly $\gamma$-anti-concentrated.
Then, the success probability of the algorithm is at most $(\gamma + \delta)^2$.
\end{lem}

\begin{proof}
By \rf{clm:successProbability}, the success probability is given by $\norm|\Pi \Upsilon\psi_T'|^2$.
By rescaling the vector $v$ so that $u = \Upsilon\psi_T'-v$ has minimal norm, we can achieve that $u$ and $v$ are orthogonal.  In particular, $\|v\|\le 1$, and it is still strongly $\gamma$-anti-concentrated.
Then, using \rf{prp:antiConcentrationMeasurement}:
\[
\|\Pi \Upsilon\psi'_T\|
\le
\|\Pi u\| + \|\Pi v\|
\le
\|u\| + \gamma\|v\|
\le
\delta + \gamma.
\qedhere
\]
\end{proof}

\subsection{Example: Polynomial Method}
\label{sec:polynomial}

In this section, we sketch how the polynomial method~\cite{beals:pol} can be interpreted in our framework.
No notion nor result of this section is used anywhere else in the paper, and it can be safely omitted while reading.
For brevity, we will say a degree-$d$ polynomial or a polynomial of degree $d$ meaning a polynomial of degree \emph{at most} $d$.

We consider the case of total Boolean functions for simplicity.
That is, $q=2$, $R=\{0,1\}$, and $\cX = (\bC^2)^{\otimes n}$.
We will not use types in this section, hence $Y = \cube$ and $\cY = \cX$.
Everywhere in this section, we identify a function $\theta$ on $\cube$ with a vector $\ket |\theta>\in \cX$ given by $\sum_{x\in\cube} \theta(x) \ket|x>\in \cX$.

There are two alternative ways how the polynomial method can be embedded in our framework.
In the first one, we obtain a simple observation that each amplitude of $\psi_{x,t}$ is a polynomial in the input variables of degree $t$.
The second alternative is to use the dual polynomials.
It is more involved, but demonstrates the use of anti-concentration in the polynomial method.
\medskip

Let us start with the first of these two alternatives.
Consider the uniform state of the algorithm from~\rf{eqn:zhandryToIndividual}:
\[
\psi_t = \frac{1}{\sqrt{2^n}} \sum_{x\in \cube} \ket X|x> \ket A|\psi_{x,t}> = \sum_{a} \ket X|\phi_a> \ket A|a>,
\]
where $a$ ranges over some orthonormal basis of the space $\cA$ of the algorithm.
By \rf{prp:inH}, $\psi_t \in \cX_{\le t}\otimes \cA$, hence, $\phi_a \in \cX_{\le t}$ for all $a$.
Recall that we identify a vector $\phi\in \cX$ with a function $\phi\colon \cube\to\bC$.
Assuming this identification, the subspace $\cX_{\le t}$ is precisely the subspace of degree-$t$ polynomials~\cite{odonnell:analysis}.
Therefore, the amplitudes of $\psi_{x,t}$ are degree-$t$ polynomials in $x$.
\medskip

We now turn to the second alternative.
We will show now that the method of dual polynomials essentially uses anti-concentration of the transfer operator $\Upsilon$ for a carefully crafted probability distribution $p_y$.
The construction of our reduction is similar to~\cite{belovs:directReduction}.

A \emph{dual polynomial} of degree $d$ is a function $\theta\colon \cube\to \bR$ satisfying two conditions: 
\begin{equation}
\label{eqn:dualPolynomial}
\sum_{x\in\cube} \abs|\theta(x)|=1
\qqand
\ket|\theta> \perp \cX_{\le d}.
\end{equation}
It is known~\cite{sherstov:patternMethod} that
\begin{equation}
\label{eqn:dualPolynomialIdentity}
\min_p \max_{x\in\cube} \absA|f(x) - p(x)| = \max_\theta \sum_{x\in\cube} \theta(x) f(x),
\end{equation}
where optimisation on the left is over all degree-$d$ polynomials $p$ and on the right over all dual degree-$d$ polynomials $\theta$.
Thus, a dual degree-$d$ polynomial $\theta$ achieving value $\eps$ on the right-hand side of~\rf{eqn:dualPolynomialIdentity} certifies that $f$ cannot be $\eps$-approximated by a polynomial of degree $d$.

From now on, we let the degree $d=2T$, fix a dual degree-$d$ polynomial $\theta$, and denote
\begin{equation}
\label{eqn:dualPolynomialEpsilon}
\eps = \sum_{x\in\cube} \theta(x) f(x).
\end{equation}
Let $p_y = |\theta(y)|$, which is a probability distribution due to the first condition of~\rf{eqn:dualPolynomial}.
We define a Boolean function $g\colon \cube \to \bool$ by $g(y) = 1$ if and only if $\theta(y)\ge 0$, and we let 
$M_0 = \{y\in \cube \mid \theta(y)<0\}$ and $M_1 = \{y\in \cube\mid \theta(y)\ge 0\}$ in accordance to $g$, not $f$.
The idea is that we can prove anti-concentration for the function $g$, and $f$ and $g$ are sufficiently well-correlated due to~\rf{eqn:dualPolynomialEpsilon}.

\begin{clm}
\label{clm:dualPolynomials}
In the above assumptions, the space $\Upsilon\cX_{\le T}$ is $\frac1{\sqrt{2}}$-anti-concentrated.
\end{clm}

\begin{proof}
Let $\phi\colon \cube\to \bC$ be a degree-$T$ polynomial, which is equivalent to $\ket |\phi> \in \cX_{\le T}$.
On the one hand, by~\rf{eqn:UpsilonStandard}:
\[
\|\Upsilon \phi\|^2 = 2^n \sum_{x\in\cube} |\theta(x)| \cdot |\phi(x)|^2.
\]
On the other hand, it is not hard to check that $|\phi|^2 = \phi^*\phi$ is a polynomial of degree $2T$, hence, using the second condition of~\rf{eqn:dualPolynomial}, we obtain
\begin{equation}
\label{eqn:dualPolynomials1}
\sum_{x\in\cube} \theta(x) \cdot |\phi(x)|^2
=
\sum_{x\in M_1} |\theta(x)| \cdot |\phi(x)|^2
-
\sum_{x\in M_0} |\theta(x)| \cdot |\phi(x)|^2
= 0.
\end{equation}
Combining the above two equations, we get
\[
\|M_1 \Upsilon \phi\|^2 =
2^n\sum_{x\in M_1} |\theta(x)| \cdot |\phi(x)|^2
=
\|M_0 \Upsilon \phi\|^2 =
2^n\sum_{x\in M_0} |\theta(x)| \cdot |\phi(x)|^2
=
\frac12 \|\Upsilon \phi\|^2,
\]
proving that $\Upsilon\cX_{\le T}$ is $\frac1{\sqrt2}$-anti-concentrated.
\end{proof}

\rf{lem:successProbability} and \rf{clm:dualPolynomials} imply:
\begin{cor}
\label{cor:dualPolynomialsG}
Any quantum algorithm evaluating the function $g$ in $T$ queries has total success probability exactly $1/2$.
\end{cor}

Note that the anti-concentration result of \rf{clm:dualPolynomials} is the best possible for a Boolean function, and the success probability of \rf{cor:dualPolynomialsG} is the same as simply guessing the output.
%Clearly, in the average case (assuming the probability distribution $p$), the success probability is precisely $1/2$.
Using the correlation from~\rf{eqn:dualPolynomialEpsilon}, we can upper bound the success probability of the algorithm on the function $f$ as well.

\begin{thm}
\label{thm:dualPol}
In the above assumptions, any quantum algorithm evaluating the function $f$ and making $T$ queries has success probability at most $1-\eps$.
\end{thm}

\begin{proof}
All the probabilities in the proof will be for $y$ sampled from $p_y = |\theta(y)|$.
Let 
\begin{align*}
 A &= \Pr\skA[f(y) = 1\; \wedge\; g(y) =1], &&&  B &= \Pr\skA[f(y) = 1\; \wedge\; g(y) =0], \\
 C &= \Pr\skA[f(y) = 0\; \wedge\; g(y) =1], &&\text{and}&  D &= \Pr\skA[f(y) = 0\; \wedge\; g(y) =0].
\end{align*}
Clearly, 
\begin{equation}
\label{eqn:dualPol1}
A+B+C+D = 1.
\end{equation}
Also, the function $\theta$ is orthogonal to the constant function (i.e, \eqrf{eqn:dualPolynomials1} holds with $\phi=1$).  Hence,
\begin{equation}
\label{eqn:dualPol2}
A+C-B-D = \Pr\skA[g(y)=1] - \Pr\skA[g(y)=0] = \sum_{y\in\cube} \theta(y) = 0.
\end{equation}
Finally, condition~\rf{eqn:dualPolynomialEpsilon} is equivalent to
\begin{equation}
\label{eqn:dualPol3}
A - B = \sum_{y\in\cube} \theta(y) f(y) = \eps.
\end{equation}
Combining~\rf{eqn:dualPol3} with~\rf{eqn:dualPol2}, we get $D-C=\eps$, which, together with~\rf{eqn:dualPol1}, gives
\begin{equation}
\label{eqn:dualPol4}
B+C = \frac12-\eps.
\end{equation}

Now, let $q_y$ be the acceptance probability of the algorithm on the input $y$.
By \rf{cor:dualPolynomialsG}:
\[
\Pr\skA[\text{Algorithm successfully evaluates $g$}] 
=
\sum_{y\in g^{-1}(1)} p_y q_y + \sum_{y\in g^{-1}(0)} p_y (1-q_y) = \frac12.
\]
We have a similar expression for the success probability of the same algorithm on the function $f$:
\[
\Pr\skA[\text{Algorithm successfully evaluates $f$}] 
=
\sum_{y\in f^{-1}(1)} p_y q_y + \sum_{y\in f^{-1}(0)} p_y (1-q_y).
\]
Subtracting the success probability on $g$ from the success probability on $f$, we get the following expression:
\[
\sum_{y\in f^{-1}(1)\cap g^{-1}(0)} p_y(2q_y-1) 
+ 
\sum_{y\in f^{-1}(0)\cap g^{-1}(1)} p_y(1-2q_y) 
\le
B+C = \frac12-\eps.
\]
Therefore, the success probability of the algorithm on the function $f$ is at most $1-\eps$.
\end{proof}

\subsection{
\texorpdfstring
	{Knowledge and the Operator $\Upsilon^+$}
	{Knowledge and the Operator Upsilon+}
}
\label{sec:overallKnowledge}

We saw in the previous section that the polynomial method works by showing anti-concentration of $\Upsilon \cX_{\le t}$.
This is achieved at the cost of having a quite unnatural probability distribution $p_y$.
Even for the OR function, finding the right probability distribution is not trivial at all, see~\cite{spalek:dualPolForOR}.
For natural probability distributions, anti-concentration of $\Upsilon\cX_{\le t}$ fails badly as the example in \rf{sec:overviewKDistAntiConcentration} illustrates.

To keep natural probability distributions, we adopt a different approach.
We will divide the state $\Upsilon \psi_t$ into two parts 
\[
\Upsilon \psi_t = \Upsilon^+ \psi_t + \Upsilon^- \psi_t,
\]
where the first part has ``knowledge'' of the output and the second one does not.
Then, we will show two things:
\begin{itemize}\noseps
\item (anti-concentration) the part $\Upsilon^- \psi_t$ is anti-concentrated, so it has small success probability regardless of its size;
\item (query gain) the part $\Upsilon^+ \psi_t$ grows slowly with $t$, so, if the total number of queries is small, the algorithm has small success probability due to the size of this part of the state.
\end{itemize}
Technically, the two parts will play the roles of the vectors $v$ and $u$ in \rf{lem:successProbability}.
The anti-concentration and the query gain statements are independent from each other and can be proven separately.

We are now going to define both operators $\Upsilon^+$ and $\Upsilon^-$ formally, but for that, we need some machinery.
First, we define a knowledge system $L^+$ that is individual for each type $\mu\in M$.
Informally, it consists of the subsets of input variables whose values reveal something important about the input string.

\begin{defn}[(Pseudo-)knowledge system]
\label{defn:knowledgeSystem}
Let $M$ be as in \rf{defn:problem}.
A \emph{pseudo-knowledge system} in $M$ is a mapping $L^\bigtriangleup$ that assigns for each $\mu\in M$ a set $L_\mu^\bigtriangleup\subseteq 2^{[n]}$ consisting of subsets of $[n]$.
A \emph{knowledge system} is a pseudo-knowledge system $L^+$ that satisfies two additional properties:
\begin{itemize}\noseps
\item none of $L^+_\mu$ contain the emptyset $\emptyset$; and
\item all $L^+_\mu$ are upwards closed: for all $\mu\in M$ and $S\subseteq S'\subseteq [n]$, if $S\in L^+_\mu$, then also $S'\in L^+_\mu$.
\end{itemize}
\end{defn}

The two properties of the knowledge system in \rf{defn:knowledgeSystem} are rather natural. 
The first one states that we know nothing if we know nothing, and the second one states that the knowledge can only increase when more input variables are revealed.
Often, the set $L^+_\mu$ is formed by the subsets $S\subseteq [n]$ with the following property: 
If the input string $y$ is of type is $\mu$, knowledge of the values of the input variables $y_i$ for $i\in S$ implies knowledge of some $\rho\in R$ such that $\mu\in M_\rho$ (akin to certificate structures from~\cite{belovs:onThePower}).

It will be helpful to introduce other pseudo-knowledge systems derived from $L^+$, and we will use different superscripts (like $-$ or $\partial i$) to distinguish them.
For instance, $L^{-}$ is defined as the complement of $L^+$, i.e., $L^{-}_\mu = 2^{[n]} \setminus L^+_\mu$ for all $\mu\in M$.
For each such superscript $\bigtriangleup$, we define a linear operator $\Upsilon^\bigtriangleup\colon \cX\to\cY$ in the following way.

\begin{defn}[(Pseudo-)Knowledge operator]
Assume $L^\bigtriangleup$ is a (pseudo-)knowledge system in $M$.
For each type $\mu\in M$, define the \emph{(pseudo-)knowledge operator} for this type as the following linear operator
\begin{equation}
\label{eqn:pseudoKnowledgeOperatorMu}
\Upsilon^\bigtriangleup_\mu
\colon \cX\to\cY_\mu\colon\quad 
\ket X|\chit \sigma> \longmapsto 
\begin{cases} \Upsilon_\mu \ket X|\chit \sigma>,& \text{if $\supp(\sigma)\in L^\bigtriangleup_\mu$;}\\ 
0,&\text{otherwise};
\end{cases}
\end{equation}
where $\Upsilon_\mu$ are the transfer operators from \rf{defn:transfer}.
For the whole collection $M$, the \emph{(pseudo-)knowledge operator} corresponding to $L^\bigtriangleup$ is a linear operator $\Upsilon^\bigtriangleup$ given by
\begin{equation}
\label{eqn:pseudoKnowledgeOperator}
\Upsilon^\bigtriangleup
\colon \cX\to\cY\colon\quad 
\phi \mapsto \bigoplus_{\mu\in M} \sqrt{p_\mu}\; \Upsilon_\mu^\bigtriangleup \phi,
\end{equation}
where $p_\mu$ are defined in~\rf{eqn:inducedProbabilities}.
\end{defn}

Compared to the transfer operator $\Upsilon$ from~\rf{eqn:UpsilonDirectSum}, the operator $\Upsilon^\bigtriangleup$ chops off the $\Upsilon_\mu \ket X|\chit \sigma>$ part of the state if the support of $\sigma$ is not an element of $L^\bigtriangleup_\mu$.
Note that 
\begin{equation}
\label{eqn:UpsilonSum}
\Upsilon^+ + \Upsilon^- = \Upsilon.
\end{equation}

The quantity $\|\Upsilon^+\psi_t\|$ is an important progress measure of the algorithm.
We call it \emph{knowledge} of the algorithm.
The following result demonstrates its importance.

\begin{lem}
[Lower Bound Framework]
\label{lem:framework}
Consider a computational problem from \rf{defn:problem} and a quantum algorithm~\rf{eqn:zhandryAlgorithm} making $T$ queries.
Assume we can define a knowledge system $L^+$ on $M$ such that the following two conditions hold:
\begin{itemize}\noseps
\item (anti-concentration) the space $\Upsilon^-\cX_{\le T}$ is $\gamma$-anti-concentrated for some $\gamma$;
\item (small knowledge) the final uniform state $\psi'_T$ of the algorithm~\rf{eqn:zhandryAlgorithm} making $T$ queries satisfies $\|\Upsilon^+ \psi'_T\| \le \delta$ for some $\delta$.
\end{itemize}
Then, the total success probability of the algorithm is at most $(\gamma + \delta)^2$.
\end{lem}

\begin{proof}
By~\rf{eqn:UpsilonSum}, $\Upsilon\psi'_T = \Upsilon^+\psi'_T  + \Upsilon^-\psi'_T $.
By the second assumption, $\|\Upsilon^+ \psi'_T\| \le \delta$.
By \rf{prp:inH}, $\psi'_T\in \cX_{\le T} \otimes \cA$.
Hence, by the first assumption, $\Upsilon^-\psi'_T\in (\Upsilon^-\cX_{\le T}) \otimes \cA$ is strongly $\gamma$-anti-concentrated.
The result now follows from \rf{lem:successProbability}.
\end{proof}

\subsection{Relation to Zhandry's Method}
\label{sec:zhandry}
In this section, we briefly discuss the relation of our method to Zhandry's compressed oracle technique~\cite{zhandry:record}.
Unlike the polynomial method, neither of the two is a rigidly defined bound, but rather a general framework, and we will only consider the original version of the method as presented in a survey paper by Hamoudi~\cite{hamoudi:introduction}, which is based on the exposition from~\cite{hamoudi:multipleCollisionPairs} with several additional examples.
None of the notions introduced in this section is used anywhere else in this paper.

\mycutecommand{\rec}{\mathrm{rec}}
Zhandry's method adds new dimension to the space $\cX_i$ from~\rf{eqn:X}.
Namely, let
\[
\cX_\rec = \cX_{1,\rec} \otimes \cdots \otimes \cX_{n,\rec},
\]
where each $\cX_{i,\rec} = \bC^{\bZ_q \cup \{\perp\}}$ with $\perp$ being some additional index.
The \emph{recording operator} is a unitary
\[
R = R_1\otimes \cdots \otimes R_n,
\]
where
$R_i\colon \cX_{i,\rec}\to \cX_{i,\rec}$ is defined by $R_i|{\perp}\rangle = \ket|\chit 0>$, $R_i\ket|\chit 0> = |{\perp}\rangle$ and $R_i\ket|\chit a> = \ket |\chit a>$ for $a>0$.
After that, the \emph {record state} of the algorithm is defined as 
$
R\psi_t
$, where $\psi_t$ is as in \rf{defn:algorithmStates}.

Zhandry's method proceeds by defining a progress measure.
For instance, in~\cite[Proposition 42]{hamoudi:introduction} the search problem is considered, where the task is to find a symbol 1 in a uniformly random input string.
The corresponding projector $\Pi\colon \cX_\rec \to \cX_\rec$ is defined as the projector onto the span of $\ket|x>$ with $x\in (\q\cup \{\perp\})^n$ satisfying $x_i = 1$ for some $i\in [n]$.
The progress measure is defined by
\begin{equation}
\label{eqn:zhandryDelta}
\norm|\Pi R \psi_t|.
\end{equation}

Let us apply the framework from \rf{sec:overviewKnowledge} to the same problem.
We can achieve this by letting $\Upsilon = I$, which means that the probability distribution $p$ is uniform.
The knowledge system $L^+_y$ is defined as the set of $S\subseteq[n]$ such that $y_i = 1$ for some $i\in S$.
Then, it is not hard to check that, for $\Upsilon^+$ as defined in~\rf{eqn:overviewKnowledgeAlternative}, we have
\[
\Pi R \ket X|\chit \sigma> = R_{\supp(\sigma)} \Upsilon^+ \ket X|\chit \sigma>
\]
for every $\sigma\in \q^n$, where, for $S\subseteq [n]$,
\[
R_S = \bigotimes_{i\in[n]} \begin{cases}
I,& \text{if $i\in S$;}\\
R_i, & \text{otherwise}.
\end{cases}
\]
However, in general, the knowledge $\|\Upsilon^+\psi\|$ can be different from $\|\Pi R\psi\|$.
As an example, consider the case $n=q=2$.
Then,
\[
\Upsilon^+ \ket X|\chit{1\mapsto 1}> = -\frac12 \sB[\ket |10> + \ket |11>]
\qqand
\Pi R \ket X|\chit{1\mapsto 1}> = -\frac1{\sqrt 2} \ket |1{\perp}>.
\]
Also,
\[
\Upsilon^+ \ket X|\chit{2\mapsto 1}> = -\frac12 \sB[\ket |01> + \ket |11>]
\qqand
\Pi R \ket X|\chit{2\mapsto 1}> = -\frac1{\sqrt 2} \ket |{\perp}1>.
\]
However, the vectors 
\begin{align*}
\Upsilon^+ \sB[ \ket X|\chit{1\mapsto 1}> +  \ket X|\chit{2\mapsto 1}> ] 
&= -\frac12 \sB[\ket |10> + \ket |01> + 2 \ket |11>]
\intertext{and}
\Pi R \sB[ \ket X|\chit{1\mapsto 1}> +  \ket X|\chit{2\mapsto 1}> ] 
&= - \frac1{\sqrt 2} \sB[\ket |1{\perp}> + \ket |{\perp}1>]
\end{align*}
have different norms.
It is an interesting question, however, whether this change is substantial, and whether it is possible to perform the same analysis as in Zhandry's method using the framework of \rf{sec:overviewKnowledge}.

\subsection{
\texorpdfstring
	{Query Gain and the Operator $\Psi^\partial$}
	{Query Gain and the Operator Psi d}
}
\label{sec:queryGain}
\mycutecommand{\Qi}{^{\partial i}}
In this section, we formulate some results on how knowledge changes during the algorithm.
First, we show that the algorithm can gain knowledge only through queries.

\begin{prp}
\label{prp:knowledgeProperties}
Knowledge satisfies the following properties:
\begin{itemize}\itemsep=0pt
\item[(a)] At the beginning, the knowledge $\|\Upsilon^+\psi_0\|$ of the algorithm is zero.
\item[(b)] A unitary $U$ applied in the space $\cA$ of the algorithm does not change the knowledge.\\
More generally, for any such unitary, any vector $\psi\in \cX\otimes \cA$, and any pseudo-knowledge operator $\Upsilon^\bigtriangleup$, we have $\|\Upsilon^\bigtriangleup U \psi \| = \|\Upsilon^\bigtriangleup \psi \|$.
\end{itemize}
\end{prp}

\begin{proof}
By \rf{eqn:zhandryInitialFourier}, the initial uniform state of the algorithm if $\psi_0 = \ket X|\chit \emptyset> \ket A|0>$.
By \rf{defn:knowledgeSystem}, $\emptyset\notin L_\mu^+$ for every $\mu\in M$, hence, $\Upsilon^+ \ket X|\chit \emptyset> = 0$.  This proves~(a).

To prove~(b), it is sufficient to note that $\Upsilon^\bigtriangleup$ and $U$ commute as acting on different registers:
\[
\normA|(\Upsilon^\bigtriangleup\otimes I_\cA) (I_\cX\otimes U)\psi |
= \normA|(I_\cY\otimes U)(\Upsilon^\bigtriangleup\otimes I_\cA) \psi |
= \normA|(\Upsilon^\bigtriangleup\otimes I_\cA) \psi |.
\qedhere
\]
\end{proof}

The next pseudo-knowledge system focuses on how much knowledge the algorithm can obtain during a query.

\begin{defn}[Query Gain]
\label{defn:queryGain}
Let $L^+$ be a knowledge system on $M$.
For $i\in [n]$, the pseudo-knowledge system $L^{\partial i}$ is made of $L^{\partial i}_\mu$ consisting of all $S\in L^-_\mu$ such that $S\cup\{i\}\in L^+_\mu$.
The \emph{query gain} operator corresponding to $L^+$ is a linear operator defined by
\[
\Psi^{\partial}\colon \cX\otimes \cI \to \cY\otimes \cI
\colon\quad \ket X|\phi> \ket I|i> \longmapsto 
\ketA Y|\Upsilon^{\partial i}\phi> \ket I|i>,
\]
where $\Upsilon^{\partial i} \colon \cX\to\cY$ are the pseudo-knowledge operators corresponding to $L^{\partial i}$.
\end{defn}

We use a different letter to emphasise that $\Psi^\partial$ is \emph{not} a pseudo-knowledge operator, and, in particular, \rf{prp:knowledgeProperties}(b) does \emph{not} apply to it.
The set $L^{\partial i}_\mu$ consists of $S\subseteq [n]$ that are on the verge of being in $L^+_\mu$:
For $\sigma$ with $\supp(\sigma)\in L^{\partial i}_\mu$, the support of $\sigma$ is not in $L^+_\mu$, but, for any $c\ne 0$, the support of $\sigma + \{i\mapsto c\}$ is.

\begin{lem}
\label{lem:queryGainMu}
For all $\mu\in M$, $i\in [n]$, and $c\in\q$, we have the following identity between linear operators from $\cX$ to $\cY_\mu$:
\begin{equation}
\label{eqn:queryGainMu}
\Upsilon^+_\mu O_{i,c} - O_{i,c} \Upsilon_\mu^+
=
O_{i,c} \Upsilon_\mu^{\partial i} - \Upsilon^{\partial i}_\mu O_{i,c}.
\end{equation}
\end{lem}

\eqrf{eqn:queryGainMu} is an identity between two commutators.
The one on the left-hand side captures the knowledge gained by one query, while, on the right-hand side, we have the same commutator but with all the terms that cancel out removed.
In particular, each of the two terms on the right has much smaller norm than each of the two terms on the left.
Therefore, we can get a meaningful upper bound on the change of norm on the left using the sum of the norms on the right via the triangle inequality, see \rf{cor:queryGainSimple} below.

\begin{proof}[Proof of \rf{lem:queryGainMu}]
If $c=0$, then $O_{i,c}$ is the identity, and~\rf{eqn:queryGainMu} is obvious.  So, we will further assume that $c\ne 0$.
We will compute all four terms in~\rf{eqn:queryGainMu} on an arbitrary $\ket X|\chit \sigma>$.
Let $\tau\in\q^\mu$ be such that $\Upsilon_\mu \ket X|\chit \sigma> = \ket Y|\mu, \chit \tau>$, 
and $\bigtriangleup$ be either $+$ or $\partial i$.
By~\rf{eqn:zhandryOicFourier}, and \rf{eqn:pseudoKnowledgeOperatorMu}:
\begin{equation}
\label{eqn:queryGainMu + O}
\Upsilon^\bigtriangleup_\mu O_{i,c} \ket X|\chit \sigma>
=
\Upsilon^\bigtriangleup_\mu \ket X|\chit {\sigma+\{i\mapsto c\}}>
=
1_{\supp\sA[\sigma+\{i\mapsto c\}]\in L^\bigtriangleup_\mu} 
\Upsilon_\mu \ket X|\chit {\sigma+\{i\mapsto c\}}>,
\end{equation}
where $1_P$ is the indicator variable equal to 1 is the predicate $P$ is true, and to 0 otherwise.
Similarly, using also that $O_{i,c}$ and $\Upsilon_\mu$ commute by \rf{fact:Upsilon Oic commute}:
\begin{equation}
\label{eqn:queryGainMu O +}
O_{i,c} \Upsilon^\bigtriangleup_\mu \ket X|\chit \sigma>
=
1_{\supp(\sigma)\in L^\bigtriangleup_\mu} \Upsilon_\mu O_{i,c} \ket X|\chit \sigma>
=
1_{\supp(\sigma)\in L^\bigtriangleup_\mu}
\Upsilon_\mu \ket X|\chit {\sigma+\{i\mapsto c\}}>.
\end{equation}
For $+$ instead of $\bigtriangleup$, subtracting~\rf{eqn:queryGainMu O +} from~\rf{eqn:queryGainMu + O}, we obtain:
\begin{equation}
\label{eqn:queryGainMu 1}
\sA[\Upsilon^+_\mu O_{i,c} - O_{i,c} \Upsilon_\mu^+] \ket X|\chit \sigma>
=
\sC[
1_{\substack{
\supp(\sigma)\notin L^+_\mu\\
\supp\sA[\sigma+\{i\mapsto c\}]\in L^+_\mu
}}
-
1_{\substack{
\supp\sA[\sigma+\{i\mapsto c\}]\notin L^+_\mu\\
\supp(\sigma)\in L^+_\mu
}}
]
\Upsilon_\mu \ket X|\chit {\sigma+\{i\mapsto c\}}>.
\end{equation}
Consider the pair of conditions of the first indicator variable in~\rf{eqn:queryGainMu 1}.
As explained in the paragraph below \rf{defn:queryGain}, they are equivalent to $\supp(\sigma) \in L_\mu^{\partial i}$.
In the same way, the second indicator variable is equal to $1_{\supp(\sigma+\{i\mapsto c\})\in L^{\partial i}_\mu}$.
Thus, using also~\rf{eqn:queryGainMu + O} and~\rf{eqn:queryGainMu O +} for $\partial i$ instead of $\bigtriangleup$, we get from~\rf{eqn:queryGainMu 1}:
\begin{align*}
\sA[\Upsilon^+_\mu O_{i,c} - O_{i,c} \Upsilon_\mu^+] \ket X|\chit \sigma>
&=
\sB[
1_{\supp(\sigma) \in L_\mu^{\partial i}}
-
1_{\supp\sA[\sigma+\{i\mapsto c\}]\in L^{\partial i}_\mu}
]
\Upsilon_\mu \ket X|\chit {\sigma+\{i\mapsto c\}}>\\
&=
\sA[O_{i,c} \Upsilon_\mu^{\partial i} - \Upsilon^{\partial i}_\mu O_{i,c}]\ket X|\chit \sigma>.
\qedhere
\end{align*}
\end{proof}

Let us argue starting from \rf{lem:queryGainMu} upwards.
First, taking the direct sum over all $\mu\in M$ and using the definition of $\Upsilon^\bigtriangleup$ from~\rf{eqn:pseudoKnowledgeOperator}, we obtain from~\rf{eqn:queryGainMu} the following identity between linear operators from $\cX$ to $\cY$:
\begin{equation}
\label{eqn:queryGain 1}
\Upsilon^+ O_{i,c} - O_{i,c} \Upsilon^+
=
O_{i,c} \Upsilon^{\partial i} - \Upsilon^{\partial i} O_{i,c}.
\end{equation}
Next, using that $O = \bigoplus_{i,c} O_{i,c}$ and $\Psi^\partial = \bigoplus_i \Upsilon^{\partial i}$, we get the following identity between linear operators from $\cX\otimes \cA$ to $\cY\otimes \cA$:
\begin{equation}
\label{eqn:queryGain 2}
\Upsilon^+ O - O \Upsilon^+
=
O \Psi^\partial - \Psi^\partial O. 
\end{equation}
Finally, applying the above identity to $\psi'_t$ and using that $\psi_{t+1} = O\psi'_t$ by~\rf{eqn:zhandryStates}, we get the following identity between the states of the algorithm:

\begin{cor}
[Query Identity]
\label{cor:queryIdentity}
It holds that
\begin{equation}
\label{eqn:queryIdentity}
\Upsilon^+\psi_{t+1} - O\Upsilon^+\psi'_t = O \Psi^\partial \psi'_t - \Psi^\partial \psi_{t+1}.
\end{equation}
\end{cor}

Using the triangle inequality, we obtain an upper bound on the increase of knowledge of the algorithm:
\begin{cor}
[Query Gain Bound]
\label{cor:queryGainSimple}
It holds that
\[
\|\Upsilon^+\psi_{t+1}\| - \|\Upsilon^+\psi_{t}\| \le \normA|\Psi^\partial\psi'_t| + \normA|\Psi^\partial \psi_{t+1}|.
\]
\end{cor}

\begin{proof}
We have the following string of identities and inequalities:
\begin{align*}
\|\Upsilon^+\psi_{t+1}\| - \|\Upsilon^+\psi_{t}\|
&=
\|\Upsilon^+\psi_{t+1}\| - \|O\Upsilon^+\psi'_{t}\|
&&\text{\rf{prp:knowledgeProperties}(b) and $O$ is unitary}\\
&\le
\normA|\Upsilon^+\psi_{t+1} - O\Upsilon^+\psi'_{t}|
&&\text{triangle inequality}\\
& = 
\normA|O \Psi^\partial \psi'_t - \Psi^\partial \psi_{t+1}|
&&\text{\rf{cor:queryIdentity}}\\
&\le \normA|\Psi^\partial\psi'_t| + \normA|\Psi^\partial \psi_{t+1}|
&&\text{triangle inequality and $O$ is unitary}
\qedhere
\end{align*}
\end{proof}

Finally, combining this with \rf{prp:knowledgeProperties}, we get the following result.

\begin{cor}
[Knowledge Upper Bound]
\label{cor:knowledgeUpperBound}
We have
\[
\| \Upsilon^+\psi'_t \| = 
\| \Upsilon^+\psi_t \|
\le
\sum_{j=0}^{t-1} \s[\normA|\Psi^\partial\psi'_j| + \normA|\Psi^\partial \psi_{j+1}|].
\]
\end{cor}

\section{Search for Equal Elements}
\label{sec:equalElements}

The framework of \rf{sec:framework} feels rather empty at this point.
It is not at all clear how to prove anti-concentration nor query gain.
Starting with this section, we make things more concrete.
From now on, all problems we consider will be of specific form we call the Equal Element Problem.
In \SuperRef Sections\ref{sec:settings} and~\ref{sec:spaceYRevisited}, we concentrate on the input of the problem: the sets $M$ and $Y$ from \rf{defn:problem}.
The first of these sections is devoted to general definitions.
We give two versions of the problem: the strict one, which is a more conventional variant, and the relaxed one, which is more suitable for our framework.
The latter yields a very clean description of the transfer operator in the Fourier basis, which we describe in \rf{sec:spaceYRevisited}.
In \rf{sec:equalElementsKnowledge}, we concentrate on the output of the problem: sets $R$ and $R_\mu$.
Also, we define the $k$-Distinctness problem.

\subsection{Formulation of the Problem}
\label{sec:settings}

The following notion will be crucial for the remaining part of the paper.

\begin{defn}[Partition]
\label{defn:partition}
A \emph{partition} is a collection of \emph{blocks} $\mu = \{B_1,\dots,B_m\}$ that satisfies:
\begin{itemize}\noseps
\item Each block $B\in\mu$ is a non-empty subset of $[n]$, 
\item their union $\bigcup_{B\in\mu} B$ equals $[n]$, and
\item they are pairwise disjoint: $B\cap C = \emptyset$ for distinct $B,C\in\mu$.
\end{itemize}
\end{defn}

A \emph{singleton} is a block of size 1, a \emph{pair} a block of size 2, and an  \emph{$\ell$-block} (or an $\ell$-tuple) has size $\ell$.
Unless specifically told otherwise, all partitions will be of $[n]$, the set of input indices.

The problems we consider from now on follow the Hidden Problem framework from \rf{defn:problem}, where it is additionally assumed that $M$, the set of types, is a set of partitions.
In particular, everywhere in this and the next sections, we use the term `partition' instead of the term `type'.
These partitions will correspond to blocks of equal elements in the input string.

We first define the strict version of the problem.
For that, we say that an input $x\in \q^n$ \emph{agrees} to a partition $\mu$ if, for all $i,j\in [n]$,
$x_i = x_j$ if and only if $i$ and $j$ belong to the same block $B$ of the partition $\mu$.

\begin{defn}[Equal Element Problem, strict version]
\label{defn:strict}
The \emph{Equal Element Problem} is a special case of the Hidden Computational Problem from \rf{defn:problem} where each type $\mu\in M$ is a partition of $[n]$ and the set of inputs $Y$ consists of \emph{all} pairs $(\mu,y)$ with $\mu\in M$ and $y\in \q^n$ agreeing to $\mu$.

In the average-case settings, a partition $\mu$ is sampled from some probability distribution $(p_\mu)_{\mu\in M}$, and each $y \in Y_\mu$ is sampled uniformly at random, where $Y_\mu$ consists of all inputs agreeing to $\mu$.
Using notation from~\rf{eqn:inducedProbabilities}, $p'_{\mu,y} = 1/|Y_\mu|$ and $p_{\mu, y} = p_\mu/|Y_\mu|$.
\end{defn}

Each input string $y$ in \rf{defn:strict} agrees to exactly one partition.
Thus, the types are not strictly necessary for the formulation of the problem in this case, since the input type can be uniquely deduced from the input string.
This is not the case for the relaxed version of the problem, where types/partitions are essential.
To define that, we need a slightly different piece of notation.
Let $\mu$ be a partition and $z\in \q^\mu$.
The string $z^{\uparrow\mu} \in \q^n$ is defined by
\begin{equation}
\label{eqn:zuparrowmu}
\text{$(z^{\uparrow\mu})_i = z_B$ for all $B\in\mu$ and $i\in B$.}
\end{equation}

\begin{defn}[Equal Element Problem, relaxed version]
\label{defn:relaxed}
In the \emph{relaxed version} of the Equal Element Problem, the set $Y$ is given by
\begin{equation}
\label{eqn:relaxedY}
Y = \sfig{ \sA[\mu, z^{\uparrow \mu}] \midA \mu\in M,\; z\in \q^\mu }.
\end{equation}
The average-case version is defined similarly to \rf{defn:strict}, where $Y_\mu$ is now the set $\sfigA{(\mu, z^{\uparrow \mu}) \mid z\in \q^\mu}$.
\end{defn}

Note that $y \in Y_\mu$ need not to agree (in the sense of \rf{defn:strict}) to $\mu$ any more since $z$ can contain equal elements.

We are primarily interested in the strict version of the problem.
However, all the analysis in the paper is done in the relaxed settings.
This is justified: as the following result shows, if the size $q$ of the input alphabet is sufficiently large, the average-case versions of the problems from \SuperRef Definitions\ref{defn:strict} and~\ref{defn:relaxed} are essentially equivalent.

\begin{prp}
\label{prp:relaxedToStrict}
Consider an algorithm solving an instance of Equal Element Problem.
Let $p^{\mathrm{strict}}$ be its success probability in the average-case settings of the strict version from \rf{defn:strict}.
Define $p^{\mathrm{relaxed}}$ similarly for the relaxed version from \rf{defn:relaxed}.
Assuming $q = \Omega(n^2)$, we have 
\[
p^{\mathrm{strict}} \le \sA[1+\OO (n^2/q)]p^{\mathrm{relaxed}}.
\]
\end{prp}

\begin{proof}
Let $s^{\mathrm{strict}}_\mu$ and $s^{\mathrm{relaxed}}_\mu$ be the success probabilities of the algorithm when $x$ is generated from $\mu\in M$ as in the strict version of \rf{defn:strict} and in the relaxed version of \rf{defn:relaxed}, respectively.
The first probability distribution over $x$ is equal to the second probability distribution conditioned on all the elements of $z\in \q^\mu$ being distinct.
The latter happens with probability at least $1 - \OO (n^2/q)$, which gives
$s^{\mathrm{relaxed}}_\mu \ge \sA[1-\OO (n^2/q)]s^{\mathrm{strict}}_\mu$, or, under the assumption $q = \Omega(n^2)$:
\[
s^{\mathrm{strict}}_\mu \le \sC[{1 + \OO \s[\frac{n^2}{q}]}]s_\mu^{\mathrm{relaxed}}.
\]
Summing over all $\mu\in M$, we get
\[
p^{\mathrm{strict}} = \sum_\mu p_\mu s^{\mathrm{strict}}_\mu \le \sC[{1 + \OO \s[\frac{n^2}{q}]}] \sum_\mu p_\mu s_\mu^{\mathrm{relaxed}} = \s[{1 + \OO \s[\frac{n^2}{q}]}] p^{\mathrm{relaxed}}.
\qedhere
\]
\end{proof}

The last proposition means that, if $q = \Omega(n^2)$ and we can show that the average success probability of the algorithm is small in the relaxed model, then its average success probability is also small in the strict model.
The latter, of course, implies the same upper bound on the success probability in the worst case.
From now on, we will assume that $q = \Omega(n^2)$ and will consider the relaxed settings from \rf{defn:relaxed}.

\subsection{
\texorpdfstring
	{Space $\cY$ Revisited}
	{Space Y Revisited}
}
\label{sec:spaceYRevisited}
In this section, we show that, assuming the relaxed version of the problem from \rf{defn:relaxed}, the transfer operators from \rf{defn:transfer} assume a particularly nice form if we introduce the Fourier basis in the space $\cY_\mu$.

We assume the standard basis of the space $\cY$ from \rf{sec:spaceY} consists of the vectors $\ket Y|\mu, z>$, where $\mu\in M$ is a partition and $z\in \q^\mu$.
This tuple is interpreted as the input $(\mu, z^{\uparrow\mu})$ from~\rf{eqn:relaxedY}.
Fixing a partition $\mu$, we get a subregister $\cY_\mu$ storing $z\in \q^\mu$.
For the latter, we can write
\begin{equation}
\label{eqn:Ymu}
\cY_\mu = \bigotimes_{B\in \mu} \cY_B,
\end{equation}
where each $\cY_B$ is isomorphic to $\bC^q$ and corresponds to a block $B$ of the partition $\mu$.
Since different $\mu$ have different such decompositions, we use a single register $\cY$ to store both $\mu$ and $z$.

We also define the Fourier basis in $\cY_\mu$ by
\begin{equation}
\label{eqn:YmuFourier}
\ket Y |\mu, \chit \tau> = \ket |\mu>\otimes \bigotimes_{B\in \mu} \ket Y_B|\chit {\tau(B)}>,
\end{equation}
where $\tau$ ranges over $\q^\mu$ and individual $\ket Y_B|\chit {\tau(B)}>$ are as in~\rf{eqn:FourierDef}.
It is an orthonormal basis of $\cY_\mu$.
In particular, for arbitrary $\beta_{\mu,\tau}\in \bC$, we have Parseval's identity
\begin{equation}
\label{eqn:YParseval}
\normC|{\sum_{\tau\in\q^\mu} \beta_{\mu,\tau} \ket Y|\mu,\chit \tau > }|^2 
=
\sum_{\tau\in\q^\mu} |\beta_{\mu,\tau}|^2.
\end{equation}

\begin{prp}
[Transfer Operators]
\label{prp:transfer}
For a partition $\mu$, the transfer operator $\Upsilon_\mu$ from \rf{defn:transfer} has the following form in the Fourier basis (where $\sigma$ ranges over $\q^n$):
\begin{equation}
\label{eqn:UpsilonMuFourier}
\Upsilon_\mu\colon \cX\to\cY_\mu\colon\quad 
\ket X|\chit \sigma> \mapsto \ket Y|\mu, \chit \tau>\quad\text{given by}\quad
\tau(B) = \sum_{i\in B} \sigma(i)\quad \text{for all $B\in\mu$}.
\end{equation}
\end{prp}

\begin{proof}
First, using~\rf{eqn:FourierDef} and~\rf{eqn:FourierSigma}, we can write
\[
\ket X|\chit \sigma> =
\frac{1}{\sqrt{q^n}} \sum_{x\in \q^n} \omega_n^{\sum_{i=1}^n \sigma(i)x_i} \ket X|x>.
\]
Now, recall that we identify $\ket Y|\mu, y>$ from~\rf{eqn:UpsilonMuStandard} (where $y\in \q^n$) with $\ket Y|\mu, z>$ (where $z\in \q^\mu$) via the identity $y = z^{\uparrow \mu}$.
Recall also that all $y\in Y_\mu = \sfigA{(\mu, z^{\uparrow \mu}) \mid z\in \q^\mu}$ have the same probability $p'_{\mu, y} = 1/|Y_\mu| = 1/q^{|\mu|}$.
Therefore, we obtain from~\rf{eqn:UpsilonMuStandard} and~\rf{eqn:zuparrowmu}:
\[
\label{eqn:transfer 1}
\Upsilon_\mu \ket X|\chit \sigma> 
=
\frac{1}{\sqrt{q^{|\mu|}}} \sum_{z\in \q^\mu} \omega_n^{\sum\limits_{i=1}^n \sigma(i)z^{\uparrow \mu}_i} \ket Y|\mu, z>
=
\frac{1}{\sqrt{q^{|\mu|}}} \sum_{z\in \q^\mu} \omega_n^{\sum\limits_{B\in \mu} z_B\sum\limits_{i\in B}\sigma(i)} \ket Y|\mu, z>,
\]
which equals $\ket Y|\mu, \chit \tau>$ from~\rf{eqn:YmuFourier}.
\end{proof}

We will be applying $\Upsilon_\mu^\bigtriangleup$ to specific vectors in $\cX$.
The following claim, which follows immediately from \rf{prp:transfer} and~\rf{eqn:pseudoKnowledgeOperatorMu}, helps in this task.
In its formulation, we use two additional pieces of notation.
First, for $\mu\in M$ and $\tau\in \q^\mu$, the set $\S[\mu,\tau]$ consists of all $\sigma\in \q^n$ satisfying $\sum_{i\in B} \sigma(i) = \tau(B)$ for all blocks $B\in \mu$.
Second, for $\bigtriangleup$ equal to $+$, $-$, or $\partial i$, the set $\S^\bigtriangleup[{\mu,\tau}]$ consists of all $\sigma\in \S[{\mu,\tau}]$ additionally satisfying $\supp(\sigma) \in L_\mu^{\bigtriangleup}$.

\begin{clm}
\label{clm:UpsilonTriangle}
For complex numbers $\alpha_\sigma$, we have
\[
\Upsilon^{\bigtriangleup}_\mu
\colon \cX\to\cY_\mu\colon\quad
\sum_{\sigma\in \q^n} \alpha_\sigma \ket X|\chit \sigma>
\longmapsto
\sum_{\tau\in \q^\mu} \sC[\sum_{\sigma\in \S^\bigtriangleup[{\mu,\tau}]} \alpha_\sigma ] \ket Y|\mu, \chit \tau>,
\]
where $\bigtriangleup$ can be either empty, or equal to $+$, $-$, or $\partial i$.
\end{clm}

\subsection{Knowledge}
\label{sec:equalElementsKnowledge}
In the previous two sections, we did not specify in detail the problem being solved, as the transfer operator is oblivious to it.
In this section, we define the output of the problem (sets $R$ and $R_\mu$), and this will also give us the corresponding knowledge system.
In principle, there are many different problems following \rf{defn:strict} like the Hidden Subgroup Problem~\cite{kitaev:HSP} or the Collision problem~\cite{brassard:collision}.
We can search for $\ell$ pairs of equal elements like in~\cite{hamoudi:multipleCollisionPairs}, and so on.
In this paper, our main goal is \rf{thm:main}, which limits our scope to the problem of finding $k$ equal elements.

\begin{defn}
[Finding $k$ Equal Elements]
\label{defn:findKEqual}
The problem of finding $k$ equal elements is a special case of the Equal Element Problem (either the strict version of \rf{defn:strict} or the relaxed one from \rf{defn:relaxed}) with the following assumptions on $R$ and $R_\mu$.
The set $R$ consists of all subsets of $[n]$ of size $k$.
For a partition $\mu\in M$, the set $R_\mu$ consists of all $\rho\in R$ that are subsets of blocks in $\mu$, i.e, $R_\mu = \sfigA{\rho\in R\mid \exists B\in \mu\colon \rho\subseteq B}$.
\end{defn}

The corresponding knowledge system $L^+$ for this problem follows this definition.
Namely, for $\mu\in M$, the set $L^+_\mu$ consists of all $S\subseteq [n]$ that have intersection of size at least $k$ with some block of the partition $\mu$.

The $k$-Distinctness problem we study in depth in \rf{sec:kDistQueryGain} follows this convention.
The most conventional formulation of this problem is as \rf{defn:findKEqual} with $M$ consisting of all possible partitions of $[n]$.
We concentrate on the version of the problem where there is unique block of $k$ equal elements.

\begin{defn}[(Unique) $k$-Distinctness.]
\label{defn:kDist}
The \emph{$k$-Distinctness} problem follows \rf{defn:findKEqual}, where each $\mu\in M$ contains exactly one $k$-block, and all other blocks of $\mu$ are of smaller size.
The \emph{Element Distinctness} problem is a special case of $k$-Distinctness with $k=2$.
\end{defn}

In this case, for each $\mu$, there is unique correct response: $R_\mu$ consists of the only $k$-block of $\mu$.
Also, the knowledge system $L^+_\mu$ consists of all subsets of $[n]$ that are simultaneously supersets of the only $k$-block of $\mu$.

\rf{defn:findKEqual} is not limited to  $k$-Distinctness.
With different choices of $M$, we get other problems.
For example, if $M$ consists of all matchings of $[n]$ (partitions into $n/2$ pairs) and $k=2$, we get the Collision problem.
\medskip

%Before we go into the lower bound for $k$-Distinctness in the next two sections, let us end this section with a simple example related to the Element Distinctness problem, that illustrates some of the concepts introduced above.

\section{Anti-Concentration}
\label{sec:kDistAntiConcentration}

In this section, we start applying the framework of \SuperRef Sections\ref{sec:framework} and~\ref{sec:equalElements} for the $k$-Distinctness problem.
We begin with two anti-concentration results.
In the first one, we prove $\OO(1/n)$-anti-concentration for the Element Distinctness problem.
The second anti-concentration result is more general.
It works for the problem of finding $k$ equal elements as in \rf{defn:findKEqual} provided that the problem is invariant under permutations and each partition $\mu\in M$ has $\Omega(n)$ singletons.
Under these assumptions, we prove $\OO(1/\sqrt n)$-anti-concentration.
In \rf{sec:kDistQueryGain}, we combine these results with the query gain estimates to obtain lower bounds for Element Distinctness and $k$-Distinctness.

All collections of partitions we use in this and the next sections satisfy the following permutation-invariance property.

\begin{defn}[Orbit]
\label{defn:orbit}
We say that a collection $M$ of partitions is an \emph{orbit} of a partition $\mu$ if the following two conditions hold:
\begin{enumerate}\noseps
\item
The collection $M$ is the orbit of $\mu$ under the action of the symmetric group on $[n]$.

\item
The associated probability distribution is uniform: $p_\mu = \frac1{|M|}$ for all $\mu\in M$.
\end{enumerate}
\end{defn}

More formally, let $S_n$ be the symmetric group consisting of all permutations of $[n]$.
For $\pi\in S_n$, and partition $\mu = \{B_1,\dots,B_m\}$, we define $\pi(\mu) = \{\pi(B_1),\dots, \pi(B_m)\}$ with $\pi(B) = \{\pi(b) \mid b\in B\}$.
Thus, in Point~1 of \rf{defn:orbit}, $M = \{\pi(\mu) \mid \pi\in S_n\}$.%, which holds for any $\mu\in M$.

\def\nodegapA{0.9}
\def\marginsizeA{0.37}

\def\drawelementA#1#2{ % column, content
    \draw (#1*\nodegapA, 0) node[draw, circle, inner sep=5] {};
    \draw (#1*\nodegapA, 0) node[] {$#2$};
}
\def\drawblockI#1#2#3#4{ % column start, column end, subtitle, filling
    \draw[#4] (#1*\nodegapA-\marginsizeA, -\marginsizeA) rectangle (#2*\nodegapA+\marginsizeA, \marginsizeA);
    \path (#1*\nodegapA, -\marginsizeA) to node[brown] {$#3$} (#2*\nodegapA, -\marginsizeA-0.4);
}
\def\drawblockA#1#2#3{ % column start, column end, subtitle
    \drawblockI{#1}{#2}{#3}{}
}
\def\drawblockF#1#2#3{ % column start, column end, subtitle
    \drawblockI{#1}{#2}{#3}{fill=gray!50}
}

\subsection{Element Distinctness}
\label{sec:EDAnti}
In this section, we prove an anti-concentration result for the special case $k=2$ of $k$-Distinctness, known as Element Distinctness.
Everywhere in this section, $M$ will be the collection of partitions corresponding to this problem.
It is defined as follows.
First, each partition $\mu\in M$ contains exactly one pair and $n-2$ singletons, and the probability distribution $p_\mu$ on $M$ is uniform.
The corresponding knowledge system $L^+$ is defined by the sets $L^+_\mu$ consisting of all subsets of $[n]$ that are supersets of the unique pair in $\mu$.
In this section, we prove the following result, fulfilling a half of the program in \rf{lem:framework}.

\begin{thm}[Element Distinctness, anti-concentration]
\label{thm:EDAntiConcentration}
Assume $t\le n/2$.
In the above assumptions, the subspace $\Upsilon^-\cX_{\le t}$ is $\OO \s[1/n]$-anti-concentrated.
\end{thm}

Consider an arbitrary vector
\begin{equation}
\label{eqn:EDAnti phi}
\phi = \sum_{\sigma\in \q^n\colon |\sigma|\le t} \alpha_\sigma \ket X|\chit\sigma> \in \cX_{\le t}.
\end{equation}
To simplify notation, we will let $\alpha_\sigma=0$ for all $\sigma$ with support size greater than $t$.
By \rf{defn:anticoncentration} of anti-concentration, we have to show that
\[
\|M_\rho\Upsilon^-\phi\| = \OO \s[\frac 1n]\|\Upsilon^-\phi\|
\]
for every $\rho\in R$.
By \rf{clm:UpsilonTriangle} and Point~2 of \rf{defn:orbit}, this is equivalent to
\[
\norm|\sum_{\mu\in M_\rho} \frac1{\sqrt{|M|}} \sum_{\tau\in \q^\mu} \sC[\sum_{\sigma\in \S^-[\mu, \tau]} \alpha_\sigma] {\ket Y|\mu,\chit\tau> }|^2
=
\OO \s[\frac1{n^2}]
\norm|\sum_{\mu\in M} \frac1{\sqrt{|M|}} \sum_{\tau\in \q^\mu} \sC[\sum_{\sigma\in \S^-[\mu, \tau]} \alpha_\sigma] {\ket Y|\mu,\chit\tau> }|^2
\]
(The only difference between the vectors on the left and the right is in the range of the sum over $\mu$.)
Note that all $\sigma \in \S^-[\mu,\tau]$ satisfy $|\sigma|\ge |\tau|$.
Hence, the sum over $\sigma$ in the above equation is zero whenever $|\tau|> t$.
Using this, and applying also Parserval's identity~\rf{eqn:YParseval}, we get that the above equation is equivalent to
\[
\sum_{\mu\in M_\rho} \sum_{\tau\in \q^\mu \colon |\tau|\le t} \absC|\sum_{\sigma\in \S^-[\mu, \tau]} \alpha_\sigma|^2
=
\OO \s[\frac1{n^2}]
\sum_{\mu\in M} \sum_{\tau\in \q^\mu \colon |\tau|\le t} \absC|\sum_{\sigma\in \S^-[\mu, \tau]} \alpha_\sigma|^2.
\]

Let us simplify notation.
First, denote by $P$ all pairs $(\mu,\tau)$, where $\mu\in M$, $\tau\in \q^\mu$, and $|\tau|\le t$.
Next, for $(\mu,\tau)\in P$, denote 
\begin{equation}
\label{eqn:beta}
\beta^-[\mu,\tau] = \sum_{\sigma\in \S^-[\mu,\tau]} \alpha_\sigma.
\end{equation}
Under these simplifications, we have to prove the following bound
\begin{equation}
\label{eqn:EDAnti Goal}
\sum_{(\mu,\tau)\in P\colon \mu\in M_\rho} \absA|\beta^-[\mu,\tau]|^2
=
\OO \s[\frac{1}{n^2}] 
\sum_{(\mu,\tau)\in P} \absA|\beta^-[\mu,\tau]|^2.
\end{equation}

The general proof strategy is clear.
We have to show that each $(\mu,\tau)\in P$ with $\mu\in M_\rho$ gives rise to a bunch (more precisely, $\Omega(n^2)$) of $(\mu',\tau')\in P$ such that $\absA|\beta^-[\mu,\tau]| = \OO(1) \absA|\beta^-[\mu',\tau']|$.
This is a bit too strong, but replacing a single pair $(\mu',\tau')\in P$ with a family of pairs gives a provable claim.
We formulate this proof strategy in the following way.

\begin{lem}[Local Alterations]
\label{lem:alterations}
Let $M$ be a collection of partitions with uniform probability distribution as in Point~2 of \rf{defn:orbit}, $M_\rho\subseteq M$, 
and $L^+$ be a knowledge system in $M$.
We assume $\phi$ and $\alpha_\sigma$ are as in~\rf{eqn:EDAnti phi}, while $P$ and $\beta^-[\mu,\tau]$ are as described near~\rf{eqn:beta}.

Assume that, for each pair $(\mu, \tau)\in P$ with $\mu\in M_\rho$, there exists a set $\Delta_{\mu,\tau}$ of some labels, and, for each $d\in \Delta_{\mu,\tau}$, there exists a subset $N_{\mu,\tau,d}\subseteq P$ satisfying 
\begin{equation}
\label{eqn:alterationsEstimate}
\absA|\beta^-[\mu,\tau]|^2 = \OO (1) \sum_{(\mu', \tau')\in N_{\mu,\tau,d}} \absA|\beta^-[\mu',\tau']|^2.
\end{equation}

Let $K$ be the minimum size of $\Delta_{\mu,\tau}$ over all $(\mu,\tau)\in P$ with $\mu\in M_\rho$.
Let $K'$ be the maximum number, taken over $(\mu',\tau')\in P$, of triples $(\mu,\tau,d)$ such that $(\mu',\tau') \in N_{\mu,\tau,d}$.
Then:
\begin{equation}
\label{eqn:alterationsResult}
\|M_\rho \Upsilon^-\phi\| = \OO \s[\sqrt{\frac{K'}{K}}] \|\Upsilon^-\phi\|.
\end{equation}
\end{lem}

\begin{proof}
Using exactly the same argument as at the beginning of this section, it suffices to show that
\[%begin{equation}
%\label{eqn:alterationsGoal}
\sum_{(\mu,\tau)\in P\colon \mu\in M_\rho} \absA|\beta^-[\mu,\tau]|^2
=
\OO \s[\frac{K'}{K}] 
\sum_{(\mu,\tau)\in P} \absA|\beta^-[\mu,\tau]|^2.
\]%end{equation}
This is shown as follows:
\begin{align*}
K \sum_{(\mu,\tau)\in P\colon \mu\in M_\rho} \absA|\beta^-[\mu,\tau]|^2
&\le
\sum_{(\mu,\tau)\in P\colon \mu\in M_\rho} \; \sum_{d\in \Delta_{\mu,\tau}} \absA|\beta^-[\mu,\tau]|^2
\\
&\le
\OO (1) \sum_{(\mu,\tau)\in P\colon \mu\in M_\rho} \sum_{d\in \Delta_{\mu,\tau}} \sum_{(\mu', \tau')\in N_{\mu,\tau,d}} \absA|\beta^-[\mu',\tau']|^2
\\
&\le
\OO (K') \sum_{(\mu', \tau')\in P} \absA|\beta^-[\mu',\tau']|^2.\qedhere
\end{align*}
\end{proof}

Let us now apply the local alteration \rf{lem:alterations} to Element Distinctness in order to prove \rf{thm:EDAntiConcentration}.
Fix $\rho = \{a,b\}$, and assume that $\mu\in M_\rho$, i.e., $\mu$ contains $\{a,b\}$ as a pair.
Fix also $\tau\in \q^\mu$ with $|\tau|\le t$, and let $v = \tau(\{a,b\})$.
We define $\Delta_{\mu,\tau}$ as the collection of pairs $\{c,d\}$ such that $\{c\}$ and $\{d\}$ are singletons of $\mu$ and $\tau(\{c\}) = \tau(\{d\})=0$.
By our assumption  $|\tau|\le t \le n/2$, the size of $\Delta_{\mu,\tau}$ is $\Omega(n^2)$, hence, in notation of \rf{lem:alterations}, $K=\Omega(n^2)$.

Temporarily fix $\{c,d\}\in \Delta_{\mu, \tau}$.
We define a partition $\mu_0$ of $[n]\setminus \{a,b,c,d\}$ and a mapping $\tau_0\in \q^{\mu_0}$ on it by (using conventions of \rf{sec:prelim}):
\[
\mu_0 = \mu \setminus\sfigA{\{a,b\}, \{c\}, \{d\}}
\qqand
\tau_0 = \tau \setminus\sfigA{ \{a,b\}\mapsto v }.
\]
From that, we define a partition $\mu'\in M$ and two mapping $\tau_a, \tau_b\in \q^{\mu'}$ on it by (see \rf{fig:EDAnticoncentration}):
\[
\mu' = \mu_0\cup\sfigA{\{a\}, \{b\}, \{c,d\}},
\qquad
\tau_a = \tau_0 \cup\sfigA{\{a\}\mapsto v},
\qqand
\tau_b = \tau_0 \cup\sfigA{\{b\}\mapsto v}.
\]
We let 
\[
N_{\mu,\tau,\{c,d\}} = \sfigA{(\mu',\tau_a), (\mu',\tau_b)}.
\]
(If $v=0$, then $\tau_a = \tau_b$, and this set is of cardinality 1.)
If $(\mu',\tau')\in N_{\mu, \tau, \{c,d\}}$ for \emph{some} pair $(\mu,\tau)$, the latter can be uniquely reconstructed from the former.
Indeed, $\{a,b\} = \rho$,  $\{c,d\}$ is the only pair of $\mu'$, and the value of $\tau(\{a,b\})$ is the only non-zero (if any) of $\tau'(\{a\})$ and $\tau'(\{b\})$.
Therefore, in notation of \rf{lem:alterations}, $K' = 1$.
To prove \rf{thm:EDAntiConcentration} using this lemma, it suffices to show~\rf{eqn:alterationsEstimate}.
For that, we have to consider two cases, see also \rf{fig:EDAnticoncentration}.

\def\drawabcd{
    \drawelementA0a
    \drawelementA1b
    \drawelementA2c
    \drawelementA3d
}

\myfigure
\label{fig:EDAnticoncentration}
=====
\[
\begin{tikzpicture}
\begin{scope}[shift={(0,0)}]
    \node[blue] at (-1.2, 0) {$(\mu,\tau)$};
    \drawabcd
    \drawblockA 010
    \drawblockA 220
    \drawblockA 330
\end{scope}
\begin{scope}[shift={(0,-2.5)}]
    \node[blue] at (-1.2, 0) {$(\mu',\tau')$};
    \drawabcd
    \drawblockA 000
    \drawblockA 110
    \drawblockA 230
\end{scope}
\draw[very thick, double distance=2pt, gray, arrows={-Latex[length=2 2.5 0]}] (1.4,-0.9) to (1.4, -2);
\begin{scope}[shift={(6.5,0)}]
    \node[blue] at (-1.2, 0) {$(\mu,\tau)$};
    \drawabcd
    \drawblockA 01v
    \drawblockA 220
    \drawblockA 330
\begin{scope}[shift={(-2.5,-3.2)}]
    \node[blue] at (1.5, 0.8) {$(\mu',\tau_a)$};
    \drawabcd
    \drawblockA 00v
    \drawblockA 110
    \drawblockA 230
\end{scope}
\begin{scope}[shift={(2.5,-3.2)}]
    \node[blue] at (1.5, 0.8) {$(\mu',\tau_b)$};
    \drawabcd
    \drawblockA 000
    \drawblockA 11v
    \drawblockA 230
\end{scope}
\draw[very thick, double distance=2pt, gray, arrows={-Latex[length=2 2.5 0]}] (1.2,-0.9) to (-0.9, -2);
\draw[very thick, double distance=2pt, gray, arrows={-Latex[length=2 2.5 0]}] (2.2,-0.9) to (4, -2);
\node[gray] at (1.3, -3.2) {\Huge$+$};
\end{scope}
\end{tikzpicture}
\]
\negbigskip
-----
An illustration to the two cases from the proof of \rf{thm:EDAntiConcentration}.
It only depicts the elements $a,b,c,d$, since for the remaining elements of $[n]$ all the partitions and the functions on them are identical.
On the left, a pair $(\mu,\tau)$ with $\tau(\{a,b\})=0$ gets replaced by a pair $(\mu',\tau')$ as in~\rf{eqn:ED Anti 1}.
On the right, a pair $(\mu,\tau)$ with $\tau(\{a,b\})=v\ne 0$ gets replaced by two pairs $(\mu', \tau_a)$ and $(\mu', \tau_b)$ like in~\rf{eqn:ED Anti 15}.
=====

\begin{itemize}
\item 
If $v=0$, then $\tau_a = \tau_b$.
Denote their common value by $\tau'$.
All $\sigma\in \S^-[\mu,\tau]$ satisfy $\sigma(a)=\sigma(b)=\sigma(c)=\sigma(d) = 0$.
The same is true for any $\sigma\in \S^-[\mu', \tau']$, which gives us $\S^-[\mu,\tau] = \S^-[\mu',\tau']$, resulting in 
\begin{equation}
\label{eqn:ED Anti 1}
\absA|\beta^-[\mu,\tau]|^2
= 
\absA|\beta^-[\mu',\tau']|^2.
\end{equation}

\item If $v\ne 0$, then all $\sigma\in \S^-[\mu,\tau]$ satisfy $\sigma(c)=\sigma(d)=0$ together with $\sigma(a)=v$ xor $\sigma(b)=v$.
This yields a decomposition into a disjoint union: $\S^-[\mu,\tau] = \S^-[\mu',\tau_a] \sqcup \S^-[\mu',\tau_b]$.  
Therefore,
\begin{equation}
\label{eqn:ED Anti 15}
\beta^-[\mu, \tau]
= 
\beta^-[\mu', \tau_a]
+
\beta^-[\mu', \tau_b],
\end{equation}
which by the QM-AM inequality gives
\begin{equation}
\label{eqn:ED Anti 2}
\absA|\beta^-[\mu, \tau]|^2
\le
2\sB[{
\absA|\beta^-[\mu', \tau_a]|^2
+
\absA|\beta^-[\mu', \tau_b]|^2
}].
\end{equation}
\end{itemize}

\eqsrf{eqn:ED Anti 1} and~\rf{eqn:ED Anti 2} are of the form~\rf{eqn:alterationsEstimate}, which finishes the proof of \rf{thm:EDAntiConcentration}.

\subsection{General Anti-Concentration Theorem}
\label{sec:anticoncentration}
In this section, we prove a general anti-concentration theorem for the task of finding $k$ equal elements from \rf{defn:findKEqual} under two additional assumptions: the problem is invariant under permutations and each partition contains many singletons.
Formally, our assumptions (including the ones from \rf{sec:equalElementsKnowledge}) are as follows:

\begin{itemize}\noseps
\item The collection $M$ is an orbit of a partition $\mu$ such that
    \begin{itemize}\noseps
    \item it has at least $cn$ singletons for some constant $c>0$, and
    \item the size of its largest block is $\OO (1)$.
    \end{itemize}
\item The set $R$ consists of all $k$-subsets of $[n]$ for some $2\le k=\OO (1)$.
\item The set $M_\rho$ with $\rho\in R$ is given by $M_\rho = \{\mu\in M\mid \exists B\in \mu\colon \rho\subseteq B\}$.
%\item The set $R_\mu$ is given by $R_\mu = \{\rho\in R\mid \exists B\in \mu\colon \rho\subseteq B\}$.
\item The knowledge system $L^+$ is given by 
$L^+_\mu = \sfigA{S\subseteq [n]\mid \exists B\in \mu\colon |B\cap S|\ge k}$.
\end{itemize}

The condition that $M$ forms a single orbit is without loss of generality.
If $M$ is a disjoint union of several orbits, each of which satisfies $\gamma$-anti-concentration, the whole collection $M$ also meets this condition.
In this section, we prove the following result.

\begin{thm}
[Anti-concentration with singletons]
\label{thm:anticoncentration}
In the above assumptions, the subspace $\Upsilon^-\cX_{\le t}$ is $\OO (1/\sqrt{n})$-anti-concentrated for all $t\le cn/2$.
\end{thm}

Fix $\rho\in R$.
It is a subset of $[n]$ of size $k$.
We again assume that $\phi\in \cX_{\le t}$ is given by~\rf{eqn:EDAnti phi},
$\beta^-[\mu,\tau] = \sum_{\sigma\in \S^-[\mu,\tau]} \alpha_\sigma$ as in~\rf{eqn:beta}, and we aim to use \rf{lem:alterations}.
The set $\S^-[\mu,\tau]$ is easy to describe.
It consists of those $\sigma\in\q^n$ that, for each $B\in \mu$, satisfy
\begin{equation}
\label{eqn:anticoncentrationLocalCondition}
\sum_{i\in B} \sigma(i) = \tau(B)
\qqand
\absA|\supp(\sigma) \cap B| < k.
\end{equation}
A nice property of these conditions is that they are independent between different $B$.

The main idea of the proof is as follows.
For every $(\mu,\tau)\in P$ with $\mu\in M_\rho$ and for every $\sigma\in \S^{-}[\mu,\tau]$, at least one $\sigma(i)$ with $i\in\rho$ is zero.
It is tempting, as we did in \rf{sec:EDAnti}, to exchange $i$ with a singleton $\{c\}\in \mu$ satisfying $\tau(\{c\})=0$.
Unfortunately, just by exchanging one element, it is hard to prove~\rf{eqn:alterationsEstimate}, as the exchange affects many different $\sigma$ in complicated ways.
Luckily, using a variant of the inclusion-exclusion principle, it is still possible to apply \rf{lem:alterations}, which we will demonstrate now.

Fix $\mu\in M_\rho$ and $\tau\in \q^\mu$ with $|\tau|\le t$.
Let $A\in\mu$ be the block that is a superset of $\rho$, and let $v = \tau(A)$.
We define $\Delta_{\mu,\tau}$ as consisting of subsets $D = \{d_1,\dots,d_k\}\subseteq [n]$ of size $k$ such that every $\{d_i\}$ is a singleton of $\mu$ satisfying $\tau(\{d_i\}) = 0$.

Fix $D = \{d_1,\dots,d_k\} \in \Delta_{\mu,\tau}$.
Let $A_0 = A\setminus \rho$, and define a partition $\mu_0$ of $[n]\setminus (A\cup D)$ and a function $\tau_0\in \q^{\mu_0}$ on it by
\[
\mu_0 = \mu\setminus\sfigA{A, \{d_1\},\dots,\{d_k\}}
\qqand
\tau_0 = \tau\setminus \{A\mapsto v\}.
\]
For each $J\subseteq \rho \cup D$ of size $k$, we define $A_J = A_0\cup J$, a partition $\mu_J\in M$, and a function $\tau_J \in \q^{\mu_J}$ on it by  (see \rf{fig:Anticoncentration})
\[
\mu_J = \mu_0\cup \{A_J\} \cup \sfigA{\{b\}\mid b \in (\rho\cup D)\setminus J}
\qqand
\tau_J = \tau_0\cup \{A_J\mapsto v\}.
\]
(Additionally to $\mu_0$, the partition $\mu_J$ contains a single block of size $|A_J|=|A|$, and $\absA|(\rho\cup D)\setminus J| = k$ singletons; hence, it is a permutation of $\mu$; hence, indeed belongs to $M$.)
In particular, $\mu_\rho = \mu$ and $\tau_\rho = \tau$.

\def\drawabcd{
    \drawelementA0{a_1}
    \drawelementA1{a_2}
    \drawelementA2{a_3}
    \drawelementA3{a_4}
    \drawelementA4{a_5}
    \drawelementA5{d_1}
    \drawelementA6{d_2}
    \drawelementA7{d_3}
    \drawelementA8{d_4}
    \drawelementA9{d_5}    
}

\myfigure
\label{fig:Anticoncentration}
=====
\[
\begin{tikzpicture}
\begin{scope}[shift={(0,0)}]
    \node[blue, anchor=east] at (-0.8, 0) {$(\mu,\tau) = (\mu_\rho, \tau_\rho)$};
    \drawabcd
    \drawblockA 04v
    \drawblockA 550
    \drawblockA 660
    \drawblockA 770
    \drawblockA 880
    \drawblockA 990
\end{scope}
\draw[very thick, double distance=2pt, gray, arrows={-Latex[length=2 2.5 0]}] (4.5*\nodegapA,-0.9) to (4.5*\nodegapA, -2);
\begin{scope}[shift={(0,-2.7)}]
    \node[blue, anchor=east] at (-0.8, 0) {$(\mu_J,\tau_J)$};
    \drawabcd
    \drawblockA 000
    \drawblockA 110
    \drawblockA 220
    \drawblockA 37v
    \drawblockA 880
    \drawblockA 990
\end{scope}
\end{tikzpicture}
\]
\negbigskip
-----
An illustration to the proof of \rf{thm:anticoncentration}.
Here $k=5$, and $A = \rho = \{a_1,\dots,a_5\}$ (so that $A_0 = \emptyset$), while $D = \{d_1,\dots,d_5\}$.
We express $\beta^-[\mu,\tau] = \beta^-[\mu_\rho,\tau_\rho]$ through identity~\rf{eqn:anticoncentration1} that involves a linear combination of $\beta^-[\mu_J, \tau_J]$ over $k$-subsets of $J\subseteq \rho\cup D$ different from $\rho$.
The pair $(\mu_J, \tau_J)$ for one such $J=\{a_4,a_5,d_1,d_2,d_3\}$ is depicted below. 
Again, the blocks of partitions and the values of the functions on them are not depicted outside of $A\cup D$ because they are identical in all partitions.
=====

We will later prove the following identity
\begin{equation}
\label{eqn:anticoncentration1}
\sum_{i=0}^k \frac{(-1)^i}{\binom ki} \sum_{J\subseteq \rho\cup D\colon |J|=k,\; |J\cap D| = i} \beta^-[\mu_J, \tau_J] = 0,
\end{equation}
but, first, let us show how it implies \rf{thm:anticoncentration} via \rf{lem:alterations}.
Using the QM-AM inequality, and that $\beta^-[\mu_\rho, \tau_\rho] = \beta^-[\mu,\tau]$, we arrive at
\begin{equation}
\label{eqn:anticoncentration2}
\absA|\beta^-[\mu,\tau]|^2 \le 
\OO \sC[{\sum_{J\subseteq \rho\cup D\colon |J|=k,\; J\ne \rho} \absA|\beta^-[\mu_J, \tau_J]|^2 }].
\end{equation}
This is a variant of~\rf{eqn:alterationsEstimate} with 
\[
N_{\mu,\tau,D} = \{(\mu_J,\tau_J) \mid J\subseteq \rho\cup D,\; |J|=k,\; J\ne \rho\}.
\]
Let us estimate $K$ and $K'$.
For the first one, using our assumption on $\mu$ and $|\tau|\le t$, we get that the number of singletons $\{d\}\in\mu$ satisfying $\tau(\{d\})=0$ is at least $cn/2 = \Omega(n)$, hence, $K = \Omega(n^k)$.
For the second one, let $(\mu',\tau')\in N(\mu,\tau,D)$ for some $(\mu,\tau)$.
Denote by $\ell$ the number of $a\in \rho$ such that $\{a\}$ is a singleton of $\mu'$.
We necessarily have $\ell\ge 1$.
There are two cases.
\begin{itemize}\noseps
\item If $\ell < k$, then $\ell$ elements of $D$ are contained in the same block of $\mu'$ as the remaining $k-\ell$ elements of $\rho$.
For each of the $k-\ell$ remaining elements of $D$, there are at most $n$ choices.
Hence, the number of choices for $D$ is $\OO (n^{k-\ell})$.
\item If $\ell = k$, then $D$ is contained in one block of $\mu'$, hence, the number of choices for $D$ is $\OO (n)$.
\end{itemize}
(In both statements, we used that the maximal size of a block in $\mu'$ is $\OO (1)$).
Knowing $D$, we can uniquely reconstruct $(\mu,\tau)$, hence, $K' = \OO (n^{k-1})$.
The theorem now follows from \rf{lem:alterations}.

It remains to prove~\rf{eqn:anticoncentration1}.
The set $\S^-[\mu_J,\tau_J]$ consists of all $\sigma\in\q^n$ such that
\begin{itemize}\itemsep=0pt
\item $\sum_{i\in B} \sigma(i) = \tau(B)$ and $|\supp(\sigma) \cap B|<k$ for all $B\in \mu_0$;
\item $\sum_{i\in B} \sigma(i) = v$ and $|\supp(\sigma) \cap B|<k$ for $B = A\cup D$; and
\item $\sigma(i) = 0$ for all $i\in (\rho\cup D)\setminus J$.
\end{itemize}
For $C\subseteq \rho\cup D$, let
$\gamma[C] = \sum_{\sigma} \alpha_\sigma$ where the sum is over those $\sigma$ that satisfy the first two bullets above and $\supp(\sigma)\cap (\rho\cup D) = C$.
By~\rf{eqn:beta} and the third bullet above:
\begin{equation}
\label{eqn:anticoncentration5}
\beta^-[\mu_J,\tau_J] = \sum_{C\subseteq J\colon |C|< k} \gamma[C].
\end{equation}
Take any $C\subseteq \rho\cup D$ with $|C|<k$.
Denote $\ell_1 = |\rho\cap C|$ and $\ell_2 = |D\cap C|$.
Replacing $\beta^-[\mu_J,\tau_J]$ in~\rf{eqn:anticoncentration1} with~\rf{eqn:anticoncentration5}, the coefficient of $\gamma[C]$ becomes (where in both binomial coefficients below, we choose elements that do \emph{not} belong to $J$):
\begin{align*}
\sum_{i=0}^k \frac{(-1)^i}{\binom ki}{\binom {k-\ell_1}i} {\binom {k-\ell_2}{k-i}}
&=
\sum_{i=\ell_2}^{k-\ell_1} (-1)^i\frac{i! (k-i)!}{k!} \frac{(k-\ell_1)!}{i!(k-\ell_1-i)!}\frac{(k-\ell_2)!}{(k-i)! (i-\ell_2)!}\\
&=\frac{(k-\ell_1)!(k-\ell_2)!}{k!(k-\ell_1-\ell_2)!} 
\sum_{i=\ell_2}^{k-\ell_1} (-1)^i
\frac{(k-\ell_1-\ell_2)!}{(k-\ell_1-i)!(i-\ell_2)!}\\
&=\frac{(k-\ell_1)!(k-\ell_2)!}{k!(k-\ell_1-\ell_2)!} \sA[1+(-1)]^{k-\ell_1-\ell_2} = 0
\end{align*}
since $\ell_1+\ell_2<k$.
This proves~\rf{eqn:anticoncentration1} and finishes the proof of \rf{thm:anticoncentration}.

\def\nodegap{0.7}
\def\rowgap{0.9}
\def\marginsize{0.25}

\def\drawelement#1#2{ % row, column
    \draw (#2*\nodegap, #1*\rowgap) node[draw, circle, inner sep=2] {};
}
\def\partitionblock#1#2#3{ %Internal function
    (#2*\nodegap-\marginsize, #1*\rowgap-\marginsize) rectangle (#3*\nodegap+\marginsize, #1*\rowgap+\marginsize);
}
\def\drawblock#1#2#3{ % row, column start, column end
    \draw\partitionblock{#1}{#2}{#3}
}
\def\fillblock#1#2#3{ % row, column start, column end
    \draw[fill=gray!50]\partitionblock{#1}{#2}{#3}
}
\def\drawbrace#1#2#3#4{
    \draw [decorate, decoration={brace,amplitude=5pt}] (#2*\nodegap-0.3, #1*\rowgap + 0.3) --node[above=4pt]{$#4$} (#3*\nodegap+0.3, #1*\rowgap +0.3);
}

\section{
\texorpdfstring
	{$k$-Distinctness Lower Bound: Query Gain}
	{k-Distinctness Lower Bound: Query Gain}
}
\label{sec:kDistQueryGain}

In this section, we finish the proof of the $k$-Distinctness lower bound by proving the query gain bound.
We will use several collections of types/partitions at once.
To distinguish between them, we will use subscripts, which also carry over to all collection-related notations, like $\cY$, $\Upsilon$, $\Upsilon^+$, and so on.

An important distinction with \SuperRef Sections\ref{sec:equalElements} and~\ref{sec:kDistAntiConcentration} is that our main type will be a \emph{highlighted partition} rather than an ordinary partition.
A highlighted partition $\mu$ is defined similarly to a partition from \rf{sec:settings}, but exactly one of its blocks is highlighted.
The set $\q^\mu$ is exactly the same for highlighted and usual partitions.
This also does not change $\Upsilon_\mu$ nor $\S[\mu,\tau]$ for any $\tau\in \q^\mu$.
What is affected by highlighting is the knowledge system: $L_\mu^+$ consists of the supersets of the highlighted block.

\subsection{Highlighted Partitions and Our Collections}
\label{sec:kDistPartitions}

In this section, we define several collections of partitions: $M_k,\dots, M_1$ and $M_{\circ(k-1)},\dots, M_{\circ 1}$.
Moreover, the first ones will consist of highlighted partitions, a notion we introduce in \rf{defn:highlighted}.
The collection we care about the most is $M_k$ that directly corresponds to the $k$-Distinctness problem.
All other collections are auxiliary and their only raison d'être is estimation of the query gain of $M_k$.

Thus, our main collection is $M_k$ as in the \rf{defn:kDist} of the $k$-Distinctness problem.
More precisely, we will restrict it to an orbit of a single partition:
$M_k$ is the orbit (in the sense of \rf{defn:orbit}) of a partition $\mu_k$ that contains exactly one $k$-block, and $\Omega(n)$ $\ell$-blocks for all $\ell\in [k-1]$.
Following \rf{defn:findKEqual}, the set $R$ of responses consists of all subsets of $[n]$ of size $k$, and for $\rho = \{a_1,\dots,a_k\}\in R$, the subset $M_\rho$ consists of all partitions from $M$ with the unique $k$-block equal to $\{a_1,\dots,a_k\}$.

The corresponding knowledge system $L^+_k$ is as in \rf{sec:equalElementsKnowledge}, which was also used for the anti-concentration \rf{thm:anticoncentration}.
For this particular collection, $L^+_{\mu}$ consists of all supersets of the only $k$-block of $\mu$.
(We do not add $k$ to the notation $L^+_\mu$ as every partition will belong to at most one collection, so there will be no risk of confusion.)
Knowledge systems for other collections will be organised similarly, as explained in the next definition.

\begin{defn}
[Highlighted partitions and knowledge systems]
\label{defn:highlighted}
A \emph{highlighted partition} is a partition $\mu$ that has exactly one block $B$ that has been \emph{highlighted}.
If $\mu$ is a highlighted partition, then the set $L^+_\mu$ consists of all subsets of $[n]$ that are simultaneously supersets of the highlighted block of $\mu$.
\end{defn}

More formally, we can define a highlighted partition as a pair $\mu = \sA[B_1, \{B_2,\dots,B_m\}]$ where $\{B_1,\dots,B_m\}$ is the corresponding partition in the usual sense of \rf{defn:partition}, and $B_1$ is the highlighted block.

\begin{defn}
[Unhighlighting]
\label{defn:unhighlight}
The \emph{unhighlighting} operation takes a highlighted partition and removes the highlighting, turning it into an ordinary partition in the sense of \rf{defn:partition}.
In the above formalisation, it acts as $\sA[B_1, \{B_2,\dots,B_m\}] \mapsto \{B_1,\dots,B_m\}$.
\end{defn}

The symmetric group in Point~1 of \rf{defn:orbit} acts on highlighted partitions in the obvious way: $\pi(\mu) = \sA[\pi(B_1), \sfigA{\pi(B_2),\dots,\pi(B_m)}]$.
The orbit of the highlighted partition $\mu$ is larger than the orbit of the corresponding unhighlighted partition (the ratio of their sizes being equal to the number of blocks with size $|B_1|$).

The collection $M_k$ follows the convention of \rf{defn:highlighted} if we highlight the only $k$-block of each partition.
For this collection, there is really no difference between the highlighted and the unhighlighted versions.
We construct other collections from $M_k$ using the following operation.

\begin{defn}[Splitting]
\label{defn:splitting}
Let $\mu_2$ be a highlighted partition, $B$ its highlighted block of size at least 2, and $i\in B$.
We say that a highlighted partition $\mu_1$ is obtained from $\mu_2$ by \emph{splitting off} the element $i$, if $\mu_1$ is obtained from $\mu_2$ by removing $B$, and adding two blocks $B\setminus \{i\}$ and $\{i\}$ instead, of which the first one is highlighted (see \rf{fig:splitting}).
In the above formalism, it acts as 
$\sA[B, \{B_2,\dots,B_m\}] \mapsto \sA[B\setminus\{i\}, \sfigA{\{i\}, B_2,\dots,B_m}]$.
\end{defn}

The motivation behind this construction is the following result.
(Recall that $L^{\partial i}_\mu$ consists of those $S\in L^-_\mu$ that satisfy $S\cup\{i\}\in L^+_\mu$, see \rf{defn:queryGain}.)
\begin{clm}
\label{clm:splitting}
If $\mu_1$ is obtained from $\mu_2$ by splitting off $i$, then, for all $\phi\in \cX$:
\[
\|\Upsilon^{\partial i}_{\mu_2} \phi\| \le 
\|\Upsilon^{+}_{\mu_1} \phi\|.
\]
\end{clm}

\begin{proof}
Let $\phi = \sum_{\sigma\in \q^n} \alpha_\sigma \ket X|\chit \sigma>$ with $\alpha_\sigma\in\bC$, and $B$ be the highlighted block of $\mu_2$.
For $\tau_2\in \q^{\mu_2}$, we define $\tau_1\in \q^{\mu_1}$ by (see also \rf{fig:splitting}):
\[
\tau_1 = \sA[\tau_2 \setminus\sfigA{ B\mapsto v }]\cup \sfigA{(B\setminus\{i\}) \mapsto v},
\]
where $v=\tau_2(B)$.
Note that $\sigma\in \S[{\mu_1,\tau_1}]$ if and only if $\sigma\in \S[{\mu_2,\tau_2}]$ and $\sigma(i)=0$.
Hence, $\S^+[{\mu_1,\tau_1}] = \S^{\partial i}[{\mu_2,\tau_2}]$.
The mapping $\tau_2\mapsto\tau_1$ is injective, thus, using \rf{clm:UpsilonTriangle} and Parseval's identity~\rf{eqn:YParseval}:
\[
\|\Upsilon^{\partial i}_{\mu_2} \phi\|^2
=
\sum_{\tau_2\in \q^{\mu_2}} \absC|\sum_{\sigma\in \S^{\partial i}[{\mu_2,\tau_2}]} \alpha_\sigma |^2
\le
\sum_{\tau_1\in \q^{\mu_1}} \absC|\sum_{\sigma\in \S^+[{\mu_1,\tau_1}]} \alpha_\sigma |^2
=
\|\Upsilon^{+}_{\mu_1} \phi\|^2.
\qedhere
\]
\end{proof}

\def\drawabcd{
    \drawelementA0{b_1}
    \drawelementA1{b_2}
    \drawelementA2{b_3}
    \drawelementA3{b_4}
    \drawelementA4{i}
}
\def\nodegapA{1.1}

\myfigure
\label{fig:splitting}
=====
\[
\begin{tikzpicture}
\begin{scope}[shift={(0,0)}]
    \node[blue, anchor=east] at (-0.8, 0) {$(\mu_2,\tau_2)$};
    \drawblockF 04v
    \drawabcd
\end{scope}
\draw[very thick, double distance=2pt, gray, arrows={-Latex[length=2 2.5 0]}] (2*\nodegapA,-0.9) to (2*\nodegapA, -2);
\begin{scope}[shift={(0,-2.7)}]
    \node[blue, anchor=east] at (-0.8, 0) {$(\mu_1,\tau_1)$};
    \drawblockF 03v
    \drawblockA 440
    \drawabcd
\end{scope}
\end{tikzpicture}
\]
\negbigskip
-----
A transformation from the pair $(\mu_2,\tau_2)$ to $(\mu_1,\tau_1)$ from the proof of \rf{clm:splitting}.
The highlighted blocks are filled in grey.
Here, $B = \{b_1,b_2,b_3,b_4,i\}$.
The set $\S^{\partial i}[{\mu_2,\tau_2}]$ is equal to the set $\S^+[{\mu_1,\tau_1}]$.
=====

Starting from $M_k$, we define two hierarchies of collections $M_k,\dots, M_1$ and $M_{\circ(k-1)},\dots, M_{\circ 1}$.
The former consists of highlighted partitions, and the latter of usual unhighlighted partitions.
(See \SuperRef Figures\ref{fig:ED} and~\ref{fig:3Dist} for schematic depictions of these collections for the special cases of $k=2$ and $k=3$.)
\begin{itemize}
\item The collections $M_{k-1},\dots, M_1$ are defined by the following inductive construction.
If $M_\ell$ is defined and $\ell>1$, we define $M_{\ell-1}$ as the orbit of $\mu_{\ell-1}$, where $\mu_{\ell-1}$ is obtained from $\mu_\ell\in M_\ell$ by splitting off one of the elements of its highlighted block.
In other words, $\mu_{\ell}$ is obtained from $\mu_k$ by removing the highlighted $k$-block of $\mu_k$ and replacing it with one highlighted $\ell$-block and $k-\ell$ singletons.
We will use knowledge of $M_\ell$ to estimate the query gain of $M_{\ell+1}$, see \rf{cor:querySpecial} below.

\item
For $\ell\in[k-1]$, we define $M_{\circ\ell}$ as the orbit of $\mu_{\circ\ell}$ that is obtained from $\mu_\ell\in M_\ell$ by unhighlighting.
It is an ordinary partition, and it does not have an associated knowledge system.
We will make use of it via its norm, see~\rf{eqn:UpsilonNorms}.
\end{itemize}

\begin{clm}
\label{clm:setSizes}
For every $\ell\in [k-1]$, we have $|M_{\ell+1}| \ge \Omega(n) |M_{\ell}|$ and $|M_{\ell}| \ge \Omega(n) |M_{\circ\ell}|$.
\end{clm}

\begin{proof}
Let us start with the first statement.
Consider the following partially defined map from $M_{\ell+1}$ to $M_\ell$.
If $\mu_{\ell+1}\in M_{\ell+1}$ contains 1 in the highlighted block, split it off and obtain a highlighted partition $\mu_\ell\in M_\ell$.
This operation is reversible, hence, it is a bijection between the subset of $\mu_{\ell+1}\in M_{\ell+1}$ containing 1 in the highlighted block and the subset of $\mu_\ell\in M_\ell$ where $\{1\}$ is a non-highlighted singleton.
This shows that these two subsets are of equal size.
The probability that 1 is in the highlighted block of a uniformly random $\mu_{\ell+1}\in M_{\ell+1}$ is $\OO (1/n)$, and the probability that $\{1\}$ is a non-highlighted singleton of a uniformly random $\mu_\ell\in M_{\ell}$ is $\Omega(1)$.
This proves that $|M_{\ell+1}| \ge \Omega(n) |M_{\ell}|$.

To prove the second statement, note that each $\mu_\circ\in M_{\circ\ell}$ gives rise to $\Omega(n)$ highlighted partitions in $M_{\ell}$ by highlighting different $\ell$-blocks of $\mu_{\circ}$.
\end{proof}

As in~\rf{eqn:pseudoKnowledgeOperator}, for each $\ell\in \{1\dots, k\}\cup\{\circ 1,\dots,\circ (k-1)\}$, we define
\begin{equation}
\label{eqn:Upsilon ell}
\Upsilon_\ell^\bigtriangleup\colon \cX\to\cY_\ell\colon\quad
\phi \mapsto \frac{1}{\sqrt{|M_\ell|}} \bigoplus_{\mu\in M_\ell} \Upsilon_\mu^\bigtriangleup \phi,
\end{equation}
where $\bigtriangleup$ is either empty (works for all $\ell$) or stands for $+$ or $\partial i$ (only works for highlighted collections).
We also define for $\ell\in[k]$:
\begin{equation}
\label{eqn:Psi ell}
\Psi^{\partial}_\ell \colon \cX\otimes \cI \to \cY_\ell\otimes \cI
\colon\quad \ket X|\phi> \ket I|i> \longmapsto 
\ketA Y_\ell|\Upsilon^{\partial i}_\ell \phi> \ket I|i>
\end{equation}
as in \rf{defn:queryGain}.

Let us make one final remark about highlighting.
Since highlighting does not change the partition as such, if a highlighted partition $\mu_{\bullet}$ is obtained from $\mu_\circ$ by highlighting one of its blocks, then $\cY_{\mu_\bullet}$ and $\cY_{\mu_\circ}$ are canonically isomorphic, and, under this isomorphism, $\Upsilon_{\mu_\bullet} = \Upsilon_{\mu_\circ}$.  
In particular, $\|\Upsilon_{\mu_\bullet} \phi\| = \|\Upsilon_{\mu_\circ} \phi\|$ for every $\phi\in \cX$.
The difference between $\mu_\bullet$ and $\mu_\circ$ manifests itself in two aspects.
The first one is knowledge, of which the latter is bereft. 
The second one is the normalisation factor in front of the direct sum in~\rf{eqn:Upsilon ell}.
If $M_\bullet$ and $M_\circ$ are the orbits of $\mu_\bullet$ and $\mu_\circ$, respectively, then $|M_\bullet|$ will be larger than $|M_\circ|$ and the coefficient will be smaller, see \rf{clm:setSizes}.

\subsection{Two Lemmata}

In this section, we prove two lemmata that hold for collections that are orbits of a highlighted partition, and which are crucial in our $k$-Distinctness lower bound proof.

\mycutecommand{\H}{\mathrm{2}}
\mycutecommand{\L}{\mathrm{1}}

\begin{lem}
\label{lem:query}
Assume $M_{\H}$ and $M_\L$ are collections that are orbits of highlighted partitions $\mu_\H$
and $\mu_\L$, respectively.
Moreover, $\mu_\L$ is obtained from $\mu_\H$ by splitting off one element of its highlighted block.
Then, for every state $\psi\in \cX\otimes \cA$, we have
\[
\|\Psi^\partial_\H \psi \| \le \sqrt{\frac{|M_\L|}{|M_\H|}} \|\Upsilon^+_\L \psi \|.
\]
\end{lem}

Returning briefly to the settings of \rf{sec:kDistPartitions}, we can combine this with \rf{clm:setSizes} to obtain the following corollary.
\begin{cor}
\label{cor:querySpecial}
For collections $M_\ell$ defined in \rf{sec:kDistPartitions}, every $\ell\in[k-1]$ and $\psi\in \cX\otimes \cA$:
\[
\|\Psi^\partial_{\ell+1} \psi \| 
\le 
\OO \s[\frac{1}{\sqrt n}]
\|\Upsilon^+_\ell \psi \|.
\]
\end{cor}

\begin{proof}[Proof of \rf{lem:query}]
We argue upwards, starting with $\Upsilon^{\partial i}_\H$ acting on $\cX$ and ending with $\Psi^\partial_\H$ acting on $\cX\otimes \cA$.
Fix $i\in [n]$, and let $\phi = \sum_{\sigma\in \q^n} \alpha_\sigma \ket X|\chit \sigma>$ be an arbitrary vector in $\cX$.

\begin{itemize}\noseps
\item 
For each $\mu_\H\in M_\H$ that does not contain $i$ in the highlighted block, we have 
$L^{\partial i}_{\mu_\H} = \emptyset$, hence, $\Upsilon^{\partial i}_{\mu_\H} \phi = 0$.

\item
For each $\mu_\H\in M_\H$ that contains $i$ in the highlighted block, we construct $\mu_\L\in M_\L$ by splitting off $i$.
By \rf{clm:splitting}, $\|\Upsilon_{\mu_\H}^{\partial i} \phi\| \le \|\Upsilon_{\mu_\L}^+ \phi\|$.
Moreover, the mapping $\mu_\H\mapsto\mu_\L$ is injective.
\end{itemize}
Using this, we obtain via~\rf{eqn:Upsilon ell}:
\begin{equation}
\label{eqn:queryProof1}
\|\Upsilon^{\partial i}_\H \phi\|^2 = 
\frac1{|M_\H|} \sum_{\mu_\H \in M_\H} \|\Upsilon^{\partial i}_{\mu_\H} \phi \|^2
\le
\frac1{|M_\H|} \sum_{\mu_\L \in M_\L} \|\Upsilon^{+}_{\mu_\L} \phi \|^2
=
\frac{|M_\L|}{|M_\H|} \| \Upsilon^{+}_\L \phi \|^2.
\end{equation}
Now we can finish the proof.  Write
\[
\psi = \sum_{i,c,w} \ket X|\phi_{i,c,w}> \ket I |i>\ket C|c>\ket W|w>,
\]
where $w$ ranges over an arbitrary orthonormal basis of $\cW$, and $\phi_{i,c,w}$ are some non-normalised vectors.
By the definition~\rf{eqn:Psi ell} of $\Psi^\partial_\H$ and~\rf{eqn:queryProof1}, we get
\[
\normA|\Psi^\partial_\H \psi|^2
=
\sum_{i,c,w} \normA|\Upsilon^{\partial i}_\H \phi_{i,c,w}|^2
\le
\frac{|M_\L|}{|M_\H|}\sum_{i,c,w} \normA|\Upsilon^{+}_\L \phi_{i,c,w}|^2
=
\frac{|M_\L|}{|M_\H|}\normA|\Upsilon^+_\L \psi|^2.
\qedhere
\]
\end{proof}

To estimate $\|\Upsilon^+_1\psi\|$, which is at the bottom of our hierarchy of \rf{sec:kDistPartitions}, we use the following result.

\begin{lem}
\label{lem:Upsilon1Norm}
Assume $M_\circ$ and $M_\bullet$ are collections that are orbits of partitions $\mu_\circ$
and $\mu_\bullet$, respectively, where 
$\mu_\circ$ is a partition having $\Omega(n)$ singletons, and $\mu_\bullet$ is the same partition with one of the singletons highlighted.
Then, for every $\phi\in \cX_{\le t}$:
\[
\|\Upsilon^+_\bullet \phi\| \le \OO \s[\sqrt{\frac{t}{n}}] \|\Upsilon_\circ \phi\|.
\]
\end{lem}

\begin{proof}
Fix arbitrary $\phi = \sum_{\sigma\in \q^n} \alpha_\sigma \ket X|\chit \sigma> \in \cX_{\le t}$ with $\alpha_\sigma\in\bC$ satisfying $\alpha_\sigma = 0$ if $|\sigma|>t$.
Denote by $S(\mu_{\circ})$ the set of singletons of $\mu_{\circ}\in M_{\circ}$, and let $|S| = |S(\mu_{\circ})| = \Omega(n)$ stand for their common size.
For every $\mu_{\circ}\in M_{}$, there are exactly $|S|$ highlighted partitions $\mu_\bullet\in M_\bullet$ that are obtained from $\mu_{\circ}$ by highlighting one of its singletons.
Denote the latter set by $H(\mu_{\circ})$.
Also, $|M_\bullet| = |S|\cdot|M_{\circ}|$.

Fix temporarily $\mu_{\circ}\in M_{\circ}$ and $\mu_\bullet \in H(\mu_\circ)$.
For $\tau \in \q^{\mu_{\circ}} = \q^{\mu_{\bullet}}$, let 
\[
\beta_{\tau} 
=\sum_{\sigma\in \S[{\mu_{\circ},\tau}]} \alpha_\sigma
=\sum_{\sigma\in \S[{\mu_{\bullet},\tau}]} \alpha_\sigma.
\]
For a highlighted partition $\mu_\bullet \in M_\bullet$ with the highlighted block $\{a\}$, we have 
\[
\S^+[\mu_\bullet, \tau] = 
\begin{cases}
\S[\mu_\bullet, \tau], &\text{if $\tau\sA[\{a\}]\ne 0$};\\
\emptyset, & \text{otherwise}.
\end{cases}
\]
Hence, using \rf{clm:UpsilonTriangle}:
\[
\sum_{\mu_\bullet\in H(\mu_{\circ})} \!\!\!\!\!\|\Upsilon_{\mu_\bullet}^+\phi\|^2
=
\!\!\!\!\sum_{\mu_\bullet\in H(\mu_{\circ})} \sum_{\tau} \absC|\sum_{\sigma\in \S^+[\mu_\bullet, \tau]} \!\!\!\!\alpha_\sigma|^2
=
\!\!\!\!\sum_{a\in S(\mu_{\circ})}\; \sum_{\tau\colon \tau(\{a\})\ne 0} \!\!\!\!\! |\beta_\tau|^2
\le
t \sum_{\tau} |\beta_\tau|^2
= t \|\Upsilon_{\mu_{\circ}} \phi\|^2,
\]
where, in the only inequality, we used that each $\tau$ with non-zero $\beta_\tau$ has support size at most $t$, hence, satisfies $\tau(\{a\})\ne 0$ for at most $t$ values of $a$.

Summing over all $\mu_{\circ}\in M_{\circ}$, we get
\[
\|\Upsilon_\bullet^+ \phi\|^2
= \frac{1}{|M_{\bullet}|} \sum_{\mu_\bullet \in M_\bullet} \|\Upsilon_{\mu_\bullet}^+\phi\|^2
\le
\frac{t}{|M_{\bullet}|} \sum_{\mu_\circ \in M_\circ} \|\Upsilon_{\mu_\circ}\phi\|^2
= \frac{t}{|S|} \|\Upsilon_{\circ}\phi\|^2
\le \OO \s[\frac{t}{n}] \|\Upsilon_{\circ} \phi\|^2.
\qedhere
\]
\end{proof}

\subsection{Element Distinctness}
\label{sec:ED}
Now we are able to prove a tight lower bound for Element Distinctness, which is the special case of $k$-Distinctness for $k=2$.
Following \rf{sec:kDistPartitions}, we have three collections of partitions, see \rf{fig:ED} for a schematic depiction.

\myfigure
\label{fig:ED}
=====
\[
\def\nodegap{1.1}
\begin{tikzpicture}
\node at (-0.5,2*\rowgap) {$M_{2}$};
\node at (-0.5,\rowgap) {$M_{1}$};
\node at (-0.5,0) {$M_{\circ 1}$};
\fillblock212
\fillblock111
\drawblock122
\drawblock011
\drawblock022
\foreach \y in {0,1,2}
{
    \foreach \x in {3,...,5, 7} {\drawblock \y\x\x};
    \draw (6*\nodegap, \y*\rowgap) node {$\cdots$};
    \foreach \x in {1,...,5, 7} {\drawelement \y\x};
}
\end{tikzpicture}
\]
-----
Schematic depiction of partitions whose orbits under the symmetric group give collections used in the query gain bound for the Element Distinctness problem.
The highlighted blocks are filled in grey.
=====

\begin{itemize}\noseps
\item[$M_2$] 
Each highlighted partition $\mu\in M_2$ has exactly one highlighted pair, and all other elements are singletons.

\item[$M_1$]
Each highlighted partition $\mu\in M_1$ consists only of singletons.
Exactly one singleton is highlighted.

\item[$M_{\circ1}$]
This collection consists of a single partition $\mu = \sfigA{\{1\},\dots,\{n\}}$ made out of singletons.
There is no highlighted block, and no associated knowledge system.
\end{itemize}

In particular, we have $|M_2| = \binom n2$, $|M_1| = n$, and $|M_{\circ1}|=1$.
The main result concerning query gain for Element Distinctness is as follows.

\begin{thm}[Element Distinctness, Query Gain]
\label{thm:EDQueryGain}
For the collections defined above, 
\[
\| \Upsilon_2^+\psi_t \| = \|\Upsilon_2^+\psi'_t\| = \OO \s[\frac{t^{3/2}}n]
\]
 for every $t$.
\end{thm}

\begin{proof}
By~\rf{eqn:UpsilonNorms}, we have $\|\Upsilon_{\circ 1} \psi_t\| = 1$ for all $t$.
From this, using \rf{lem:Upsilon1Norm} together with the observation that $\psi_t\in \cX_{\le t}\otimes \cA$ by \rf{prp:inH}:
\begin{equation}
\label{eqn:EDQueryGain1}
\|\Upsilon^+_1 \psi'_t \| = \|\Upsilon^+_1 \psi_t \| \le \OO \s[\sqrt{\frac{t}{n}}] \|\Upsilon_{\circ 1} \psi_t\| = \OO \s[\sqrt{\frac{t}{n}}],
\end{equation}
where we also used \rf{prp:knowledgeProperties}(b).
Hence, by~\rf{cor:querySpecial}, we get
\[
\|\Psi_2^\partial\psi_t '\|, \|\Psi_2^\partial\psi_t \|  =
\OO \s[\frac{1}{\sqrt n}]
\|\Upsilon^+_1 \psi_t \|
= \OO \s[\frac{\sqrt t}{n}].
\]
Thus, by \rf{cor:knowledgeUpperBound}:
\begin{equation}
\label{eqn:EDQueryGainLast}
\|\Upsilon^+_2 \psi_t' \| = \|\Upsilon^+_2 \psi_t \| \le 
\sum_{j=0}^{t-1} \sA[\|\Psi_2^\partial\psi_j' \|+ \|\Psi_2^\partial\psi_{j+1} \|]
\le \OO \sC[\sum_{j=1}^t \frac{\sqrt j}{n}] = \OO \s[\frac{t^{3/2}}n].
%\qedhere
\end{equation}
\end{proof}

Combining this with the anti-concentration \rf{thm:EDAntiConcentration}, we get the following result via our lower bound framework \rf{lem:framework} and \rf{prp:relaxedToStrict}.

\begin{thm}
Assuming the alphabet size $q = \Omega(n^2)$, any quantum algorithm solving the search version of the Element Distinctness problem in $T$ queries has success probability
\[
\OO \s[\frac{1}{n^2} + \frac{T^3}{n^2}].
\]
\end{thm}

Thus, to achieve constant success probability, the algorithm has to make $\Omega(n^{2/3})$ queries, matching the upper~\cite{ambainis:distinctness} and the best known lower bound~\cite{shi:collisionLower}.

\subsection{3-Distinctness: First Attempt}
\label{sec:3DPartI}

Let us perform similar estimations as in \rf{sec:ED} for the 3-Distinctness problem, and see what we obtain.
The collections $M_3$, $M_2$, $M_1$, $M_{\circ 2}$, and $M_{\circ 1}$ were already described in \rf{sec:kDistPartitions}, see also \rf{fig:3Dist}.

\myfigure
\label{fig:3Dist}
=====
\[
\begin{tikzpicture}
\node at (-0.5,4*\rowgap) {$M_{3}$};
\node at (-0.5,3*\rowgap) {$M_{2}$};
\node at (-0.5,2*\rowgap) {$M_{1}$};
\node at (-0.5,\rowgap) {$M_{\circ 2}$};
\node at (-0.5,0) {$M_{\circ 1}$};
\fillblock413
\fillblock312
\drawblock333
\fillblock211
\drawblock222
\drawblock233
\drawblock112
\drawblock133
\drawblock011
\drawblock022
\drawblock033
\foreach \y in {0,...,4}
{
    \foreach \x in {1,...,7,9,10,11,12,13,15} {\drawelement \y\x};
    \foreach \x/\z in {4/5,6/7,9/10} {\drawblock\y\x\z};
    \foreach \x in {11,12,13,15} {\drawblock\y\x\x};
    \draw (8*\nodegap, \y*\rowgap) node {$\cdots$};
    \draw (14*\nodegap, \y*\rowgap) node {$\cdots$};
}
\drawbrace44{10}{\Omega(n)}
\drawbrace4{11}{15}{\Omega(n)}
\end{tikzpicture}
\]
-----
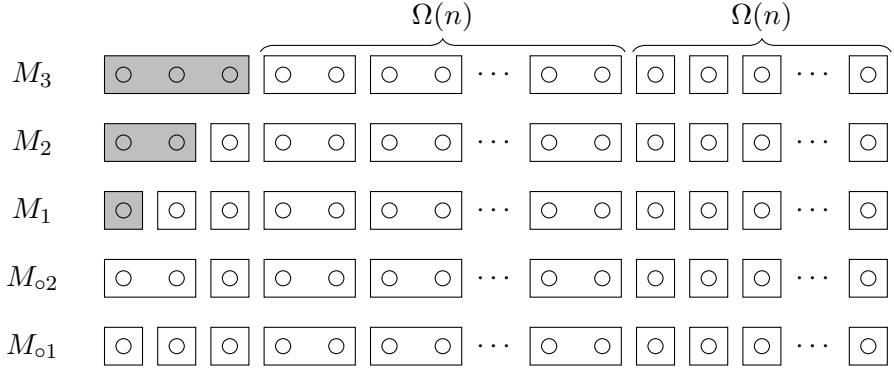
Schematic depiction of partitions whose orbits under the symmetric group give collections used in the query gain bound for the 3-Distinctness problem.
The highlighted blocks are filled in grey.
=====

Mimicking the proof of \rf{thm:EDQueryGain}, we obtain in exactly the same way, for all $t$:
\begin{equation}
\label{eqn:3D1}
\|\Upsilon_{\circ 1} \psi_t\|=1,
\qquad
\|\Upsilon^+_1 \psi'_t \| = \|\Upsilon^+_1 \psi_t \| = \OO \s[\sqrt{\frac{t}{n}}],
\qquad
\|\Psi_2^\partial\psi_t \|, \|\Psi_2^\partial\psi_t' \|  = \OO \s[\frac{\sqrt t}{n}]
\end{equation}
and
\begin{equation}
\label{eqn:3D2}
\norm|\Upsilon_2^+ \psi_t'| = 
\norm|\Upsilon_2^+ \psi_t| \le 
\sum_{j=0}^{t-1} \sA[\|\Psi_2^\partial\psi_j' \|+ \|\Psi_2^\partial\psi_{j+1} \|]
\le \OO \sC[\sum_{j=1}^t \frac{\sqrt j}{n}] = \OO \s[\frac{t^{3/2}}n].
\end{equation}
Another application of~\rf{cor:querySpecial} gives us
\[
\|\Psi_3^\partial\psi_t \|, \|\Psi_3^\partial\psi_t' \|  = \OO \s[\frac{t^{3/2}}{n^{3/2}}].
\]
Which, by \rf{cor:knowledgeUpperBound}, yields:
\[
\|\Upsilon^+_3 \psi_t' \| = 
\|\Upsilon^+_3 \psi_t \| \le 
\sum_{j=0}^{t-1} \sA[\|\Psi_3^\partial\psi_j' \|+ \|\Psi_3^\partial\psi_{j+1} \|]
\le \OO \sC[\sum_{j=1}^t \frac{j^{3/2}}{n^{3/2}}] = \OO \s[\frac{t^{5/2}}{n^{3/2}}].
\]
This proves a lower bound of $\Omega(n^{3/5})$ on the number of queries required to achieve knowledge $\Omega(1)$.
This is disappointing  as it is even worse than the lower bound $\Omega(n^{2/3})$ obtained in \rf{sec:ED} for Element Distinctness.
Obviously, our analysis is suboptimal.
In retrospect, this is not surprising at all as we have never used that the input partition has $\Omega(n)$ pairs --- a crucial requirement for a tight lower bound.
\medskip

To get some intuition what we are doing wrong here, let us take a step back and reconsider the Element Distinctness problem we analysed in \rf{sec:ED}.
More precisely, consider the bound~\rf{eqn:EDQueryGain1}.
Instead of the direct proof we performed in \rf{lem:Upsilon1Norm}, we could use a proof strategy similar to the one used in~\rf{eqn:EDQueryGainLast}.
First, reasoning as in \rf{lem:query}, it is not hard to show that
\begin{equation}
\label{eqn:3DIntution2}
\|\Psi^\partial_1 \psi_t \| \le \OO \sC[\frac1{\sqrt n}] \|\Upsilon_0 \psi_t \| = \OO \sC[\frac1{\sqrt n}].
\end{equation}
Then, using a similar argument as in~\rf{eqn:EDQueryGainLast}, we obtain
\begin{equation}
\label{eqn:3DIntution3}
\|\Upsilon^+_1 \psi_t' \| = 
\|\Upsilon^+_1 \psi_t \| \le 
\sum_{j=0}^{t-1} \sA[\|\Psi_1^\partial\psi_j' \|+ \|\Psi_1^\partial\psi_{j+1} \|]
\le \OO \sC[\sum_{j=1}^t \frac{1}{\sqrt{n}}] = \OO \s[\frac{t}{\sqrt n}],
\end{equation}
much worse a bound than the one in~\rf{eqn:EDQueryGain1}.
It is easy to construct an algorithm that saturates the bound~\rf{eqn:3DIntution2}.
For instance, take an algorithm that blindly queries one variable after another: first $x_1$, then $x_2$, and so on.
It satisfies $\|\Psi^\partial_1\psi_t\|=1/\sqrt{n}$ for all $t$, but~\rf{eqn:EDQueryGain1} still holds, because the new part added to $\Upsilon_1^+\psi_t$ is \emph{orthogonal} to the existing state:
First, we add the part with the singleton $\{1\}$ highlighted, then the part with the singleton $\{2\}$ highlighted, and so on.
Reasoning by analogy, this might make us suggest that the bounds in~\rf{eqn:3D1} are tight, but the one in~\rf{eqn:3D2} is loose because the query gain bound from \rf{cor:queryGainSimple} is inadequate in this settings.
We will show that this is indeed the case.

\subsection{Refined Query Gain Bound}

Our next goal is to prove a more precise version of the query gain bound tailored to our case of highlighted partitions.
For this section, we will assume that $M$ is some collection of highlighted partitions and the corresponding knowledge systems are as in \rf{defn:highlighted}.
We do \emph{not} require that they are permutation invariant or have uniform probability distribution.
Similarly to \rf{sec:framework}, all notation related to $M$ (like $\cY$, $\Upsilon^+$) will be without subindices in this section.
The main observation (already stated in the first bullet of the proof of \rf{lem:query}) is that each query can change only a small part of the state.

\begin{defn}
\label{defn:queryProfile}
The \emph{query profile} operator $\Xi\colon \cY \otimes \cI \mapsto \cY\otimes \cI$ is defined as the orthogonal projector onto the span of all states $\ket Y|\mu, \chit{\tau}>\ket I|i>$, where $i$ is contained in the highlighted block of $\mu$.
\end{defn}

As the name suggests, query gain only happens in the image of $\Xi$.

\begin{clm}
\label{clm:queryProfileKnowledge}
We have $\Xi \Psi^\partial = \Psi^\partial$.
Also, $O\Xi = \Xi O$.
\end{clm}

\begin{proof}
We start with the second identity.
Recall that $O$ is decomposable into a direct sum $O = \bigoplus_{\mu, i,c} O_{\mu, i,c}$ with $O_{\mu, i,c}\colon \cY_\mu\to\cY_\mu$.
(For clarity, we write $O_{\mu,i,c}$ for what was $O_{i,c}$ in~\rf{eqn:YOic}).
We have a similar decomposition $\Xi = \bigoplus_{\mu,i} \Xi_{\mu, i}$ with $\Xi_{\mu,i}\colon \cY_\mu\to\cY_\mu$ being either the identity or the zero operator.
Hence, $O_{\mu, i,c}$ commutes with $\Xi_{\mu,i}$, implying that $O$ commutes with $\Xi$.

To prove the first statement, note that $\Xi_{\mu,i}$ is zero if and only if $i$ is not contained in the highlighted block of $\mu$.
If this holds, then $L^{\partial i}_{\mu} = \emptyset$, thus, $\Upsilon^{\partial i}_\mu = 0$.
Therefore, $\Xi_{\mu, i} \Upsilon^{\partial i}_\mu = \Upsilon^{\partial i}_\mu$.
Taking the direct sum over $i$ and $\mu$, we get $\Xi \Psi^\partial = \Psi^\partial$.
\end{proof}

We now start working towards our refinement of the Query Gain Bound, \rf{cor:queryGainSimple}.
Denote the common value from the Query Identity~\rf{eqn:queryIdentity} by $w$:
\begin{equation}
\label{eqn:RQG Identity}
w = \Upsilon^+\psi_{t+1} - O\Upsilon^+\psi'_t = O \Psi^\partial \psi'_t - \Psi^\partial \psi_{t+1}.
\end{equation}

Applying $\Xi$ to the right-hand side and using \rf{clm:queryProfileKnowledge}, we get
\[
\Xi\sA[O \Psi^\partial \psi'_t - \Psi^\partial \psi_{t+1}]
=
O \Xi\Psi^\partial \psi'_t - \Xi\Psi^\partial \psi_{t+1}
=
O \Psi^\partial \psi'_t - \Psi^\partial \psi_{t+1},
\]
i.e., $\Xi w = w$.
Denote $\Xi^\perp = I_{\cY\otimes \cA} - \Xi$ and consider the following three vectors
\[
u = \Xi^\perp \Upsilon^+ \psi_{t+1} = \Xi^\perp O \Upsilon^+ \psi'_{t},
\qquad
v_2 = \Xi \Upsilon^+ \psi_{t+1}
\qqand
v_1 = \Xi O \Upsilon^+ \psi_{t}'.
\]
From $w = v_2 - v_1$, we get $\|v_2\|\le \|v_1\|+\|w\|$, which gives
$\|v_2\|^2 \le \|v_1\|^2 + 2\|v_1\|\cdot\|w\|+\|w\|^2$, or
\[
\|\Upsilon^+\psi_{t+1}\|^2
=
\|u\|^2 +\|v_2\|^2
\le
\|u\|^2 + \|v_1\|^2 +2\|v_1\|\cdot\|w\|+\|w\|^2
= \|O \Upsilon^+ \psi'_{t} \|^2 + 2\|v_1\|\cdot\|w\|+\|w\|^2.
\]
For the individual terms above, we have the following identities.
First, using that $O$ is unitary and \rf{prp:knowledgeProperties}(b) with~\rf{eqn:zhandryStates}, $\|O \Upsilon^+ \psi'_{t} \| = \|\Upsilon^+ \psi_t\|$.
Next, by \rf{clm:queryProfileKnowledge} and unitarity of $O$, $\|v_1\| = \|\Xi \Upsilon^+ \psi'_t\|$.
Also, $\|w\| \le \|\Psi^\partial \psi'_t\| + \|\Psi^\partial \psi_{t+1}\|$ by~\rf{eqn:RQG Identity}.
Putting everything together, we obtain the following result.

\begin{lem}
[Refined Query Gain Bound]
\label{lem:queryGainRefined}
It holds that
\[
\|\Upsilon^+\psi_{t+1}\|^2 - \|\Upsilon^+\psi_{t}\|^2 
\le 
2 \|\Xi \Upsilon^+ \psi'_t\| \sB[\|\Psi^\partial \psi'_t\| + \|\Psi^\partial \psi_{t+1}\|] + \sB[\|\Psi^\partial \psi'_t\| + \|\Psi^\partial \psi_{t+1}\|]^2.
\]
\end{lem}

\subsection{k-Distinctness}
\label{sec:kDist}

In this section, we use the above machinery to prove a tight lower bound on the quantum query complexity of the $k$-Distinctness problem.
We return to the collections $M_k,\dots, M_1$ and $M_{\circ(k-1)},\dots, M_{\circ 1}$ defined in \rf{sec:kDistPartitions}.
Our main improvement compared to \rf{sec:3DPartI} is the Refined Query Gain Bound,~\rf{lem:queryGainRefined}.
Thus, we first show that the new player of that lemma, $\Xi \Upsilon^+ \psi'_t$, has small norm.

\begin{lem}
\label{lem:Xik}
For each $\ell=2,\dots,k-1$, it holds that $\normA|\Xi \Upsilon^+_\ell \psi'_t| = \OO \s[\frac1{\sqrt n}]$.
\end{lem}

\begin{proof}
Let
\begin{equation}
\label{eqn:Xi3 psi}
\psi_t' = \sum_{\sigma, i, c, w} \alpha_{\sigma,i,c,w} \ket X|\chit{\sigma}> \ket I|i> \ket C|c> \ket W|w>,
\end{equation}
where $\sigma$ ranges over $\q^n$, $i$ over $[n]$, $c$ over $\q$, and $w$ over some orthonormal basis of $\cW$.
During a larger part of the proof, we will fix some values of $i,c,w$, so, for $\mu\in M_\ell$ and $\tau\in\q^\mu$, denote
\begin{equation}
\label{eqn:Xi3 beta+}
\beta^+[\mu,\tau] = \beta^+_{i,c,w}[\mu,\tau] = \sum_{\sigma\in \S^+[\mu,\tau]} \alpha_{\sigma, i,c,w}.
\end{equation}
Let $G_i$ be the set of those $\mu\in M_\ell$ that contain $i$ in the highlighted block.
Using \rf{clm:UpsilonTriangle} and the \rf{defn:queryProfile} of $\Xi$, we get
\begin{equation}
\label{eqn:Xi3Norm1}
\|\Xi \Upsilon^+_\ell \psi'_t\|^2
= 
\frac1{|M_\ell|}\sum_{i,c,w} \sum_{\mu\in G_i} \sum_{\tau\in \q^\mu} \absA|\beta^+_{i,c,w}[\mu,\tau]|^2.
\end{equation}

The idea is to upper bound this norm using $\Upsilon_{\circ j}\psi'_t$ for $j\in[\ell]$, which are normalised vectors by~\rf{eqn:UpsilonNorms}.
For that, we will use a variant of local alterations.
Towards that end, let us denote for $\mu_j\in M_{\circ j}$ and $\tau_j\in \q^{\mu_j}$:
\begin{equation}
\label{eqn:Xi3 beta}
\beta[\mu_j,\tau_j] =
\beta_{i,c,w}[\mu_j,\tau_j] =
\sum_{\sigma\in \S[\mu_j,\tau_j]} \alpha_{\sigma, i,c,w}.
\end{equation}
Using \rf{clm:UpsilonTriangle}, we obtain for $j\in [\ell]$:
\begin{equation}
\label{eqn:Xi3Norm2}
\|\Upsilon_{\circ j}\psi'_t\|^2=
\frac1{|M_{\circ j}|} \sum_{i,c,w} \sum_{\mu_j\in M_{\circ j}} \sum_{\tau_j\in \q^{\mu_j}} \absA|\beta_{i,c,w}[\mu_j,\tau_j]|^2 = 1.
\end{equation}

Fix now $i,c,w$, as well as $\mu\in G_i$, and $\tau\in \q^{\mu}$.
Let $B$ be the highlighted block of $\mu$, and $v = \tau(B)$.
In particular, $|B|=\ell$.
Let $\mu_\circ \in M_{\circ\ell}$ be the unhighlighted version of $\mu$.

\def\drawabcd{
    \drawelementA0{b_1}
    \drawelementA1{b_2}
    \drawelementA2{b_3}
    \drawelementA3{b_4}
    \drawelementA4{i}
}
\def\nodegapA{0.9}

\myfigure
\label{fig:Xik}
=====
\[
\begin{tikzpicture}
\begin{scope}[shift={(0,0)}]
    \node[blue, anchor=east] at (-0.8, 0) {$(\mu,\tau)$};
    \drawblockF 04v
    \drawabcd
\end{scope}
\draw[very thick, double distance=2pt, gray, arrows={-Latex[length=2 2.5 0]}] (2*\nodegapA,-0.9) to (2*\nodegapA, -2);
\begin{scope}[shift={(0,-2.7)}]
    \node[blue, anchor=east] at (-0.8, 0) {$(\mu_J,\tau_J)$};
    \drawblockA 000
    \drawblockA 110
    \drawblockA 23v
    \drawblockA 440
    \drawabcd
\end{scope}
\begin{scope}[shift={(8, 0)}]
    \begin{scope}[shift={(0,0)}]
        \node[blue, anchor=east] at (-0.8, 0) {$(\mu,\tau)$};
        \drawblockF 040
        \drawabcd
    \end{scope}
    \draw[very thick, double distance=2pt, gray, arrows={-Latex[length=2 2.5 0]}] (2*\nodegapA,-0.9) to (2*\nodegapA, -2);
    \begin{scope}[shift={(0,-2.7)}]
        \node[blue, anchor=east] at (-0.8, 0) {$(\mu_\emptyset,\tau_\emptyset)$};
        \drawblockA 000
        \drawblockA 110
        \drawblockA 220
        \drawblockA 330
        \drawblockA 440
        \drawabcd
    \end{scope}
\end{scope}
\end{tikzpicture}
\]
\negbigskip
-----
A transformation of the pair $(\mu,\tau)\in G_i$ into a pair $(\mu_J,\tau_J)$ from the proof of \rf{lem:Xik}.
Here, $\ell=5$ and the highlighted block $B = \{b_1,b_2,b_3,b_4,i\}$ of $\mu$ contains $i$.
On the left, $J = \{b_3, b_4\}$, and the corresponding $\mu_J$ belongs to $M_{\circ 2}$.
On the right, $J = \emptyset$, and the corresponding $\mu_\emptyset$ belongs to $M_{\circ 1}$.
The case on the right only makes appearance if $v = \tau(B) = 0$.
=====

For $J\subseteq B$, we will define a partition $\mu_J$ and a function $\tau_J\in \q^{\mu_J}$ on it by turning all the elements of $B\setminus J$ into zero-valued singletons.
Formally, there are two cases in dependence on $j = |J|$ (see \rf{fig:Xik}):

\begin{itemize}\noseps
\item If $j>0$, then
\[
\mu_J = \sA[\mu_\circ \setminus B]\cup J \cup \sfigA{\{b\} \mid b\in B\setminus J} \in M_{\circ j}
\qqand
\tau_J = \sA[\tau \setminus \sfigA{B\mapsto v}] \cup \sfigA{J\mapsto v}.
\]
\item If $j = 0$ (which only makes sense for $v=0$), then
\[
\mu_J = \sA[\mu_\circ \setminus B] \cup \sfigA{\{b\} \mid b\in B} \in M_{\circ 1}
\qqand
\tau_J = \tau.
\]
\end{itemize}

Let us also introduce the following quantity
\[
\gamma[J] = \sum_{\sigma\in \S[\mu,\tau]\colon \supp(\sigma)\cap B = J} \alpha_\sigma.
\]
Note that
$
\beta^+[\mu,\tau] = \gamma[B]
$
and
$
\beta[\mu_J, \tau_J] = \sum_{C\subseteq J} \gamma[C].
$
From this, we get by inclusion-exclusion:
\[
\beta^+[\mu,\tau] = \sum_{J\subseteq B} (-1)^{\ell - |J|} \beta[\mu_J, \tau_J].
\]
Using the QM-AM inequality (and that $\ell=\OO (1)$):
\begin{equation}
\label{eqn:Xik Main}
\absA|\beta^+[\mu,\tau]|^2 \le \OO \sC[{ \sum_{J\subseteq B} \absA|\beta[\mu_J, \tau_J]|^2 }].
\end{equation}
Now sum over all the pairs $(\mu,\tau)$ with $\mu\in G_i$ and $\tau\in \q^\mu$.
Consider $j\in [\ell]$ and a pair $(\mu_j, \tau_j)$ with $\mu_j \in M_{\circ j}$ and $\tau_j\in \q^{\mu_j}$.
Let us estimate for how many pairs $(\mu,\tau)$ it appears on the right-hand side of~\rf{eqn:Xik Main}.
The highlighted partition $\mu$ can be reconstructed by merging one $j$-block of $\mu_j$ with its $\ell-j$ singletons.
One of these $\ell-j+1$ blocks must contain $i$, and, for each of the others, there are at most $n$ choices.
For a fixed reconstructed $\mu$, the function $\tau$ can be reconstructed uniquely.
Hence, $(\mu_j, \tau_j)$ appears at most $\OO (n^{\ell-j})$ times on the right-hand side of~\rf{eqn:Xik Main}.
Therefore,
\[
\sum_{\mu\in G_i} \sum_{\tau\in \q^\mu} \absA|\beta^+[\mu,\tau]|^2
\le
\sum_{j=1}^\ell \OO (n^{\ell-j}) \sum_{\mu_j\in M_{\circ j} } \sum_{\tau_j\in \q^{\mu_j}} \absA|\beta[\mu_j,\tau_j]|^2.
\]

Summing over all $i,c,w$, and using~\rf{eqn:Xi3Norm1} and~\rf{eqn:Xi3Norm2}, we get
\[
\|\Xi \Upsilon^+_\ell \psi'_t\|^2 
\le 
\sum_{j=1}^\ell 
\OO \s[\frac{n^{\ell-j}|M_{\circ j}|}{|M_\ell|}] \|\Upsilon_{\circ j}\psi'_t\|^2
=
\OO \s[\frac1 n].
\]
In the second inequality above, we used that $|M_\ell| = \Omega(n^{\ell-j+1})|M_{\circ j}|$ by \rf{clm:setSizes}, and that $\|\Upsilon_{\circ j}\psi'_t\| =1$ for all $j$ by~\rf{eqn:UpsilonNorms}.
\end{proof}

\begin{thm}
\label{thm:kDistQueryGain}
For the $k$-Distinctness problem, using the collection $M_k$ as defined in \rf{sec:kDistPartitions}:
\[
\|\Upsilon_k^+\psi_t' \| = \|\Upsilon_k^+\psi_t\| = 
\OO \s[{\sqrt[2^{k-1}]{\frac{t^{2^k-1}}{n^{3\cdot 2^{k-2} -1}}}}]
\]
for every $t$.
\end{thm}

\begin{proof}
By \rf{prp:knowledgeProperties}(b), $\|\Upsilon_\ell^+\psi_t \| = \|\Upsilon_\ell^+\psi'_t\|$ for all $\ell$, so we will use them interchangeably.
By~ \rf{eqn:UpsilonNorms}, $\|\Upsilon_{\circ 1} \psi_t\| = 1$.
From this, using \rf{lem:Upsilon1Norm} and that $\psi_t\in \cX_{\le t}\otimes\cA$ by \rf{prp:inH}:
\begin{equation}
\label{eqn:kDist1}
\|\Upsilon^+_1 \psi'_t \| = \|\Upsilon^+_1 \psi_t \| \le \OO \s[\sqrt{\frac{t}{n}}] \|\Upsilon_{\circ 1} \psi_t\| = \OO \s[\sqrt{\frac{t}{n}}].
\end{equation}
Now we use induction on $\ell$.
For $\ell \in [k-2]$, assume that $\alpha = \alpha_\ell>0$ is such that $\|\Upsilon^+_\ell \psi_t \| = \OO \sB[{\s[\frac{t}{n}]^\alpha}]$ for all $t$.
Then, by~\rf{cor:querySpecial}, we get
\[
\|\Psi_{\ell+1}^\partial\psi_t' \| + \|\Psi_{\ell+1}^\partial\psi_{t+1} \|  
\le 
\OO \s[\frac{1}{\sqrt n}]
\sA[\|\Upsilon^+_\ell \psi'_t \| + \|\Upsilon^+_\ell \psi_{t+1} \|]
=\OO \s[\frac{(t+1)^\alpha}{n^{\alpha+\frac12}}].
\]

The Refined Query Gain Bound from~\rf{lem:queryGainRefined} combined with our bound $\|\Xi \Upsilon^+_{\ell+1} \psi'_t\| = \OO (1/\sqrt{n})$ from~\rf{lem:Xik} reads as
\[
\|\Upsilon^+_{\ell+1}\psi_{t+1}\|^2 
-
\|\Upsilon^+_{\ell+1}\psi_{t}\|^2
\le
\OO \s[\frac{(t+1)^\alpha}{n^{\alpha+1}} + \frac{(t+1)^{2\alpha}}{n^{2\alpha+1}}]
=
\OO \s[\frac{(t+1)^{\alpha}}{n^{\alpha+1}}],
\]
where we used that $\s[\frac{t+1}{n}]^\alpha \le 1$.
Thus, using also \rf{prp:knowledgeProperties}(a), we get
\[
\|\Upsilon^+_{\ell+1}\psi_t\|^2 
\le
\sum_{j=1}^{t} \OO \s[\frac{j^{\alpha}}{n^{\alpha+1}}]
=
\OO \s[{\s[\frac{t}{n}]}^{\alpha+1}].
\]
Therefore, we have the recursive relation
\[
\alpha_{\ell+1} = \frac{\alpha+1}{2}.
\]
Denoting now $\alpha = \alpha_{k-1}$, and using that $\alpha_1 = 1/2$ by~\rf{eqn:kDist1}, we have:
\[
\alpha = \alpha_{k-1} = \frac12 + \frac 14 + \cdots + \frac1{2^{k-1}} = 1 - \frac1{2^{k-1}}.
\]

Let us now estimate $\|\Upsilon_k^+\psi_t\|$.
Another application of \rf{cor:querySpecial} gives us
\[
\|\Psi_{k}^\partial\psi_t \|, \|\Psi_{k}^\partial\psi_t' \|  
\le 
\OO \s[\frac{1}{\sqrt n}]
\|\Upsilon^+_{k-1} \psi_t \|
=\OO \s[\frac{t^{\alpha}}{n^{\alpha+\frac12}}].
\]
We cannot use the Refined Query Gain Bound any more because \rf{lem:Xik} fails for $\ell=k$, but we can use the original version, \rf{cor:knowledgeUpperBound}.
It gives
\[
\|\Upsilon^+_k \psi_t \| \le 
\sum_{j=0}^{t-1} \sA[\|\Psi_k^\partial\psi_j' \|+ \|\Psi_k^\partial\psi_{j+1} \|]
\le \sum_{j=1}^t \OO\s[\frac{t^{\alpha}}{n^{\alpha+\frac12}}]
= \OO \s[\frac{t^{\alpha+1}}{n^{\alpha+\frac12}}].
\]
To finish the proof, it suffices to evaluate both exponents:
\[
\alpha+1 = 2 - \frac{1}{2^{k-1}} = \frac{2^k-1}{2^{k-1}}
\qqand
\alpha+\frac12 = \frac32 - \frac{1}{2^{k-1}} = \frac{3\cdot 2^{k-2}-1}{2^{k-1}}.
\qedhere
\]
\end{proof}

Combining this with the general anti-concentration bound from \rf{thm:anticoncentration} in our framework from \rf{lem:framework}, we get that any quantum algorithm solving the $k$-Distinctness problem in the relaxed settings of \rf{defn:relaxed} using $T$ queries has success probability 
\[
\OO \s[{\frac{1}{n} + \sqrt[2^{k-2}]{
\frac{T^{2^k-1}}{n^{3\cdot 2^{k-2} -1}}
}}].
\]
This bound is sub-optimal (at least in the first term), but, nonetheless, yields the following theorem, which matches the upper bound from~\cite{belovs:learningKDist}.

\begin{thm}[$k$-Distinctness Lower Bound]
Assuming $q = \Omega(n^2)$, any quantum algorithm solving the search version of the $k$-Distinctness problem with constant success probability makes 
\[
\Omega\s[n^{\frac{3\cdot 2^{k-2}-1}{2^k-1}}]
=
\Omega\s[n^{\frac 34 - \frac{1}{4(2^k-1)}} ]
\]
queries.
The lower bound holds even if the algorithm knows that there is exactly one $k$-tuple of equal elements and is given in advance the precise number of $\ell$-tuples of equal elements in the input string for every $\ell\in[k-1]$.
\end{thm}

By the standard reduction from the search to the decision version~\cite{ambainis:adv}, this gives the same lower bound on solving the decision version with bounded error.

\section{Summary and Future Work}

In this paper, we proved a tight quantum query lower bound for the $k$-Distinctness problem.
However, we made a number of assumptions that significantly simplified our proofs.
One of the assumptions is that the input partition contains $\Omega(n)$ singletons.
It is interesting whether this condition can be dropped.
In particular, can we prove the same lower bound when $M$ consists of the orbit of a partition with one $k$-block, $\Omega(n)$ $(k-1)$-blocks, and no blocks of other sizes?
More generally, can we tightly characterise complexity of $k$-Distinctness when $M$ is the orbit of an arbitrary partition (assuming possibly, it contains a unique $k$-block).
Related problems are to get a tight trade-off between number of queries and success probability, or to drop the condition on the size of the input alphabet.

Another open problem is to generalise these techniques beyond Equal Element Problems.
One candidate could be the Sum Problems from~\cite{belovs:onThePower}.  We have tight characterisation of their bounded-error quantum query complexity, but the techniques of this paper could prove tighter trade-offs between number of queries and success probability of the algorithm.
It is also interesting to understand relation between our framework and the original version of Zhandry's compressed oracle technique.

\subsection*{Acknowledgements}

I am grateful to Fr\'ed\'eric Magniez for introducing me to the subtleties of Zhandry's method, Amin Shiraz Gilani for many useful discussions, and to the anonymous reviewers for their valuable suggestions.
I would also like to thank Sebastian Zur for bringing reference~\cite{hamoudi:introduction} to my attention.

This research is supported by University of Latvia project ``Latvian excellence and competence centre in quantum technologies''.

\bibliographystyle{habbrvM}
{
\small
\bibliography{belov}
}

\end{document}